\documentclass{amsart}
\usepackage[margin=3cm]{geometry}
\usepackage{color}  
\usepackage{xcolor}
\usepackage{faktor}
\usepackage{tikz}
\usepackage{tikz-cd}
\usepackage{spverbatim}
\usepackage{float}
\usepackage[hyperfootnotes=false,colorlinks=false,hypertexnames=false]{hyperref}
\usepackage{cancel}
\usepackage{faktor}
\usepackage{ytableau}
\usepackage{subcaption}
\usepackage[textsize=footnotesize]{todonotes}
\presetkeys%
    {todonotes}%
    {color=red}{}%
\usepackage{amsmath,amssymb,anyfontsize,enumerate,stmaryrd,mathtools, thmtools,tikz,txfonts,nccmath}
\usepackage[utf8]{inputenc}
\usepackage[oldstyle]{libertine}
\usepackage[final]{microtype}

\usepackage[libertine,libaltvw,liby]{newtxmath}

\usepackage{csquotes}
\makeatletter
\renewcommand*\libertine@figurestyle{LF}
\makeatother
\makeatletter
\renewcommand*\libertine@figurestyle{OsF}
\makeatother
\usepackage[noerroretextools,backend=biber,style=alphabetic,maxalphanames=5,giveninits,sorting=nyt,%
maxbibnames=99]{biblatex}
\usepackage{mmacells}
\mmaSet{
  morefv={gobble=2},
  linklocaluri=mma/symbol/definition:#1,
  morecellgraphics={yoffset=1.9ex}
}

\usepackage{tikz-cd}
\usepackage[noabbrev,capitalise]{cleveref}
\usepackage{autonum}
\usepackage{braket}
\usepackage{graphicx} 

\newtheorem{prop}{Proposition}[section]
\theoremstyle{plain}
    \newtheorem{theorem}{Theorem}[section]
    \newtheorem{construction/theorem}[theorem]{Construction/Theorem}
    
    \newtheorem{lemma}[theorem]{Lemma}
    \newtheorem{proposition}[theorem]{Proposition}
    
    \newtheorem{algorithm}[theorem]{Algorithm}
\theoremstyle{definition}
    \newtheorem{remark}[theorem]{Remark}
    
    \newtheorem{example}[theorem]{Example}
    \newtheorem{definition}[theorem]{Definition}

\newcommand{\comment}[1]{}

\newcommand{\Le}{\reflectbox{L}}
\newcommand{\LeBox}{\tikz[baseline=-0.1ex]{
    \draw (0,0) rectangle (1.5ex, 1.5ex);
}}
\newcommand{\LeCircle}{\tikz[baseline=-0.1ex]{
    \draw (0,0) rectangle (1.5ex, 1.5ex);
    \draw (0.75ex, 0.75ex) circle (0.55ex); 
}}

\newcommand{\drawcross}[2]{
    \draw (#1+0.5, #2) -- (#1+0.5, #2-1);
    \draw (#1, #2-0.5) -- (#1+1, #2-0.5);
}

\newcommand{\drawbridge}[2]{
    \draw (#1, #2-0.5) arc (90:0:0.5);
    \draw (#1+0.5, #2) arc (180:270:0.5);
    
    \draw (#1+0.353, #2-0.646) -- (#1+0.646, #2-0.353);
    
    \filldraw[black] (#1+0.353, #2-0.646) circle (2.5pt);
    \filldraw[fill=white, draw=black, line width=0.5pt] (#1+0.646, #2-0.353) circle (2.5pt);
}

\newcommand{\drawboundarybridge}[2]{
    \draw (#1, #2-0.5) arc (90:45:0.5);
    
    \draw (#1+0.5, #2) arc (180:225:0.5);
    
    \draw (#1+0.353, #2-0.646) -- (#1+0.646, #2-0.353);
}

\newcommand{\drawwhitestone}[2]{
    \draw (#1, #2) rectangle (#1+1, #2-1);
    
    \filldraw[fill=white, draw=black, line width=0.5pt] (#1+0.5, #2-0.5) circle (0.35);
}

\newcommand{\inlinebridge}{\tikz[baseline=-0.2ex, x=2ex, y=2ex]{
    \draw (0,0) rectangle (1,1);
    
    \draw (0, 0.5) arc (90:0:0.5);
    \draw (0.5, 1) arc (180:270:0.5);
    
    \draw (0.354, 0.354) -- (0.646, 0.646);
    
    \filldraw[black] (0.354, 0.354) circle (0.15ex);
    \filldraw[fill=white, draw=black, line width=0.4pt] (0.646, 0.646) circle (0.15ex);
}}

\newcommand{\inlinecross}{\tikz[baseline=-0.2ex, x=2ex, y=2ex]{
    \draw (0,0) rectangle (1,1);
    
    \draw (0.5, 1) -- (0.5, 0);
    
    \draw (0, 0.5) -- (1, 0.5);
}}

\title{Combinatorial geometry of the 2D Toda lattice and Davey Stewartson equation}
\author[R.~N.~Betancourt]{Rodrigo Navarro Betancourt}
\address{R.~N.~Betancourt: School of Mathematics, 17 Westland Row, Trinity College Dublin, Dublin 2, Ireland}
\email{navarror@tcd.ie}

\author[M.~A.~Hahn]{Marvin Anas Hahn}
\address{M.~A.~Hahn: School of Mathematics, 17 Westland Row, Trinity College Dublin, Dublin 2, Ireland}
\email{hahnma@tcd.ie}

\author[V.~Posch]{Vera Posch}
\address{V.~Posch: School of Mathematics, 17 Westland Row, Trinity College Dublin, Dublin 2, Ireland}
\email{poschv@tcd.ie}

\author[V.~Reda]{Vincenzo Reda}
\address{V.~Reda: School of Mathematics, 17 Westland Row, Trinity College Dublin, Dublin 2, Ireland}
\email{redav@tcd.ie}

\begin{document}

\begin{abstract}
    The KP equation is a prototypical $(2+1)$--dimensional integrable PDE. Its soliton solutions are famously parametrized by the Sato Grassmannian. In seminal work, Kodama and Williams \cite{kodama2014kp} made the surprising discovery that the combinatorics of soliton solutions are intimately related to the combinatorics of the totally positive Grassmannian as pioneered by Postnikov \cite{postnikov2006total}. They introduced novel algorithmic methods inspired by polyhedral structures arising from tropical geometry. Soliton solutions to the 2D Toda lattice and the Davey--Stewartson equation, two closely related integrable systems with soliton solutions, are also classified by the Sato Grassmannian. Kodama suggested in \cite{kodama} that the methods of his work with Williams could generalize to these two integrable equations. In this work, we show that this is indeed the case. We derive algorithms to produce contour plots from elements in the totally nonnegative Grassmannian in both cases. In the asymptotic setting, we recover and refine previous work of Biondini and Wang \cite{Biondini_2010}; as well as Biondini, Kireyev and Maruno \cite{biondini2022soliton}.
\end{abstract}

\maketitle

\section{Introduction}

The ubiquitous Kadomtsev--Petviashvili (KP) equation,
\begin{equation}
    (-4u_t+6uu_x+u_{xxx})_x+3u_{yy}=0,
\end{equation}
is an integrable equation describing waves in shallow water. Here, the function $u(x,y,t)$ yields the wave amplitude at $(x,y)$ at time $t$. The study of its soliton solutions has received much attention and is a prototype for integrable systems in dimension $2+1$.

Work of the Kyoto school \cite{sato1981soliton} revealed that the soliton solutions of the KP equation (or more generally, the KP hierarchy) are parametrized by the \emph{Sato Grassmannian}. A classical and complementary approach to KP solutions is through algebraic geometry: finite-gap solutions of the KP hierarchy are associated with algebraic curves, and (tropical) degenerations of these curves give rise to multi-soliton solutions \cite{Abenda_2017,abenda2025tropicalkptheorybanana,agostini2021kpsolitonstropicallimits}. In the last decade, a series of seminal works by Kodama and Williams \cite{kodama2014kp,kodama2013deodhar,kodama2011combinatorics} revealed an intimate relation between the rich combinatorics of totally positive/non--negative Grassmannians and regular KP solitons.

More precisely, given a soliton solution $\tau_A$ corresponding to a matrix $A$ in the Sato Grassmannian, one may associate to it a "tropical limit" called its \emph{contour plot} $C_A$. The contour plot $C_A$ is a piecewise linearly embedded graph on the plane $\mathbb{R}^2$, which decomposes the plane into the regions of dominance of $\tau_A$. An ingenious insight of Kodama and Williams was an unexpected relationship between the combinatorics of contour plots and the combinatorics of certain decompositions of (totally non--negative) Grassmannians called the \emph{Deodhar decomposition} \cite{deodhar1985some} or \emph{Postnikov's positroid stratification} \cite{postnikov2006total}. The cells in this decomposition are indexed by combinatorial games on Young diagrams called $\Le$--diagrams. Assuming $A$ lies in the cell indexed by a given $\Le$--diagram $\Gamma$, Kodama and Williams proved that the combinatorial shape of $C_A$ is determined by the \emph{reduced pipedream} -- a graph drawn onto a Young diagram -- associated to $\Gamma$. Moreover, they provided an algorithm to determine the embedding of $C_A$ into $\mathbb{R}^2$. Building on these methods, Kodama and Williams made a connection to cluster algebras and studied the inverse problem: Given a contour plot, can we reconstruct the original soliton solution? Fixing the time $t$ and under some genericity assumption, they gave an affirmative answer.

As noted above, the KP equation is considered a prototype for integrable systems for two spatial and one time dimension. Two closely related systems are the \emph{2D Toda lattice} equation, given by
\begin{equation}
    \frac{\partial^2}{\partial x_{-1}\partial x_1}\ln(1-V_n(x_{-1},x_1))=V_{n+1}(x_{-1},x_1)-2V_{n}(x_{-1},x_1)+V_{n-1}(x_{-1},x_1).
\end{equation}
for a discrete parameter $n$, as well as the \emph{Davey--Stewartson equation},
\begin{equation}
    \begin{split}
        &iq_t+\frac{1}{2}q_{xx}-\frac{1}{2}q_{yy}+2qQ+4|q|^2q=0\\
        &Q_{xx}+Q_{yy}=-4(|q|^2)_{xx}.
    \end{split}
\end{equation}
for $q(x,y,t)$ a complex and $Q(x,y,t)$ a real valued function.

The key feature that the 2D Toda lattice and the Davey--Stewartson equation share with the KP equation is that their soliton solutions are parametrized by the Sato Grassmannian. However, the solutions have a different shape. The building blocks of all the soliton solutions considered here are functions $e^{p_1x+p_2y+p_3z}$ for parameters $p_i$, but the three equations admit solutions with different relations between the $p_i$. For the KP equation, we obtain the parabola with $p_2=p_1^2$. The 2D Toda lattice corresponds to the hyperbola $p_1p_2=1$, and the Davey--Stewartson equation yields the unit circle $p_1^2+p_2^2=1$ \cite[\S 3.2]{kodama}.

In \cite[Remark 8.1]{kodama}, Kodama suggests that his work with Williams could extend to the 2D Toda lattice and the Davey--Stewartson (DS) equation. In this work, we confirm this suggestion by devising algorithms that allow us to draw contour plots for both integrable equations directly from a $\Le$--diagram. The algorithm relies on certain analytic behaviors of specific exponential polynomials. Moreover, we find new phenomena, such as the appearance of parallel lines in contour plots for both the 2D Toda lattice and Davey--Stewartson equation. In the asymptotic case, by which we mean that the variables in the 2D Toda Lattice and Davey--Stewartson equations are large enough, we recover and refine previous results obtained in \cite{Biondini_2010,biondini2022soliton}. Finally, whenever the matrix parameter is in the totally positive Grassmannian, we also solve the inverse problem for the two integrable systems in question.

\subsection{Structure} In \cref{sec-preliminaries}, we recall basic notions required for our work. To begin with, we describe the Wronskian type solutions of the 2D Toda lattice and Davey--Stewartson equation. Moreover, we define the notion of \emph{contour plots} for both integrable systems. Lastly, we discuss relevant decompositions of Grassmannians and introduce the notions of $\Le$--diagrams and (reduced) pipedreams. In \cref{sec-dual}, we show how the duality between Grassmannians $\mathrm{Gr}(k,n)$ and $\mathrm{Gr}(n-k,n)$ translates to a duality of $\tau$--functions arising from these respective linear subspaces. In \cref{sec-2DTodalattice}, we first discuss the asymptotic behavior of contour plots for the 2D Toda lattice. Then, we derive the desired algorithm in \cref{algorithm-TL1} in the $(x,t)$--plane. Moreover, we show that a similar discussion produces analogous results in \cref{algorithm-TL2} for the $(x,n)$--plane, which recovers and refines previous work of Biondini and Wang in \cite{Biondini_2010}. We then turn our attention to the Davey--Stewartson equation in \cref{sec:Davey-Stewartson}, where again we first study the asymptotic behavior of contour plots, reproducing and refining previous results of Biondini, Kireyev and Maruno in \cite{biondini2022soliton}, as well as obtaining an algorithm in the general case in \cref{algorithm-TL2}. Finally, we resolve the inverse problem in \cref{sec-inverseproblem}. In the appendix, we provide -- for the sake of completeness -- several proofs of results we require, where the arguments however are similar to the KP case.

\subsection{Acknowledgements}
The authors would like to thank Claudia Fevola, Yuji Kodama, and  Marius de Leeuw for their valuable feedback on an earlier draft of the manuscript. VP was supported by ERC-2022-CoG - FAIM 101088193.

\section{Preliminaries}
\label{sec-preliminaries}
\subsection{Wronskian type \texorpdfstring{$\tau$}{τ}-functions}\label{sec:wronsk}
The KP equation belongs to the family of differential equations that can be bilinearized in the sense of Hirota \cite{Hirota2004}. Two other equations in this family are the 2D Toda lattice equation and the Davey--Stewartson equation. Here we introduce these two equations and show that they allow for Wronskian type solutions. Expressly, these solutions are of the form:
\begin{equation}\label{wronskian}
    \tau^{(n)}(x)=\text{Wr}(f_1^{(n)},\dots,f_N^{(n)})(x)=\det\begin{pmatrix}
        f_1^{(n)}&\dots& f_{1}^{(n+N-1)}\\
        \vdots &\ddots& \vdots\\
        f_N^{(n)} &\dots& f_{N}^{(n+N-1)}
    \end{pmatrix}
\end{equation}
with the boundary conditions $\frac{\partial f^{(n)}}{\partial x_i}=f^{(n+i)}$, $i \in \mathbb{Z}\backslash\{0\}$. Imposing the general ansatz
 $$ f_j^{(n)}=\sum_{i=1}^M a_{ij} \lambda^n_i \; \text{exp}\biggl(\sum_{k\in \mathbf{Z}} \lambda_i^k x_k \biggr),$$
one may write the $\tau$-function as:
\begin{equation}\label{generltau}
    \begin{split}
        \tau^{(n)} 
        =\sum_{I\in \binom{[M]}{N}} \prod_{i\in I} \lambda_i^n \exp\biggl(\sum_{i\in I}\sum_{k\in \mathbf{Z}}  \lambda_i^k x_k \biggr) \Delta_{I}(A) \prod_{i>j}(\lambda_i-\lambda_j),
    \end{split}
\end{equation}
where we have denoted $[M]= \{1,\dots,M\}$ and by $\binom{[M]}{N}$ the set of $N$-element subsets $I$ of $[M]$. We also denoted by $\Delta_I(A)$ the corresponding $N\times N$ minor of $A$, where $I$ selects the columns of $A$. Important bits of initial data needed here are the constant $N\times M$ matrix $A$ and a set of constant,  distinct values $\{\lambda_1<\dots<\lambda_M\}\in \mathbb{R}\backslash \{0\}$. As we will see, specifying these constants determines the soliton graph.
\begin{remark}
     It will be convenient to have the sum in equation \eqref{generltau} not run through the entirety of the set $\binom{[M]}{N}$; rather, only indices in the \emph{matroid} $\mathcal{M}(A)$ contribute, where
    \begin{equation}
        \mathcal{M}(A)=\left\{ I \in \binom{[M]}{N}\, \bigg|\,\Delta_I(A) \neq 0 \right\}.
    \end{equation}
\end{remark}

\begin{definition}
 Let $n$ and $k$ be positive integers with $k\le n$. Then, we define the associated Grassmannian as
 \begin{equation}
     \mathrm{Gr}(k,n)=\{L\subset\mathbb{R}^n\mid\textrm{$L$ is a linear subspace},\,\mathrm{dim}(L)=k\}.
 \end{equation}
\end{definition}

It is well--known that the Grassmannian $\mathrm{Gr}(k,n)$ may be realized as a projective variety. To see this, we note that every linear space $L\in\mathrm{Gr}(k,n)$ may be represented by a full rank $k\times n$ matrix whose rows are basis vectors of $L$. Two matrices $A$ and $A'$ represent the same linear space if and only if there is an invertible $k\times k$ matrix $B$ with $BA=A'$, i.e. $A$ and $A'$ differ by row operations. Thus, we obtain the following well--defined map, which is called the \textbf{Pl\"ucker embedding}

\begin{align}
    \Phi\colon\mathrm{Gr}(k,n)&\to\mathbb{P}_{\mathbb{R}}^{\binom{n}{k}-1}\\
    A&\mapsto(\Delta_I(A))_{I\in\binom{[n]}{k}}
\end{align}

The subspace $\mathrm{Im}(\Phi)$ realizes $\mathrm{Gr}(k,n)$ as a projective variety that is cut out by the Pl\"ucker relations. The vector $(\Delta_I(A))_{I\in\binom{[n]}{k}}$ is called a Pl\"ucker vector and we call its entries Pl\"ucker coordinates. In what follows we will often conflate the matrix $A$ with its corresponding point in the Grassmannian $\mathrm{Gr}(k,n)$.

\begin{definition}\label{def-tnn_and_tp}
  We define the totally nonnegative (positive) Grassmannian  as the subset of $\mathrm{Gr}(k,n)$ consisting of linear subspaces $L\subset\mathbb{R}^n$ represented by a matrix $A$ with all $\Delta_I(A)\ge0$ (resp. $\Delta_I(A)>0$). We will denote the totally nonnegative (positive) Grassmannian by $\mathrm{Gr}(k,n)_{\geq 0}$ (resp. $\mathrm{Gr}(k,n)_{>0}$).
\end{definition}

Equation \eqref{generltau} implies that if two matrices $A$ and $A'$ are related by an invertible transformation, so that $A'=GA$ for $G \in \mathrm{GL}_N(\mathbb{R})$, the $\tau$--functions they parametrize will be the same, up to a scalar. By this reasoning, we think of $\tau$--functions as parametrized by Grassmannian points as opposed to matrices.

\subsubsection{The 2D Toda Lattice}
Let $n\in\mathbb Z$, $x_1,x_{-1}\in \mathbb{R}$ and $V_n: \mathbb{R}^2\rightarrow \mathbb{R}$. The $2$-dimensional Toda lattice (2DTL) equation for the function $V_n(x_{-1},x_1)$ has the following form:
\begin{equation}\label{eq-Toda}
    \frac{\partial^2}{\partial x_{-1}\partial x_1}\ln(1-V_n(x_{-1},x_1))=V_{n+1}(x_{-1},x_1)-2V_{n}(x_{-1},x_1)+V_{n-1}(x_{-1},x_1).
\end{equation}
This is a non-linear differential equation, and therefore hard to handle. Hirota \cite{Hirota2004} introduced the following change of variables, to make life easier: 
\begin{equation}\label{eq-solution}
    V_n(x_{-1},x_1)=\frac{\partial^2}{\partial x_{-1}\partial x_1}\ln\tau^{(n)}(x_{-1},x_1).
\end{equation}
This leads to what is called a \textbf{Hirota bilinear equation}:
\begin{equation}\label{eq-Hirota}
    \tau^{(n)}\frac{\partial^2\tau^{(n)}}{\partial x_{-1}\partial x_1}-\frac{\partial\tau^{(n)}}{\partial x_{-1}}\frac{\partial\tau^{(n)}}{\partial x_1}=(\tau^{(n)})^2-\tau^{(n+1)}\tau^{(n-1)}.
\end{equation}
\begin{remark}\label{rmk-biondini}
    It is important to note that the functions $f^{(n)}(x_{-1},x_1)$ appearing in \cref{wronskian} satisfy boundary conditions that differ slightly from those considered in \cite{Biondini_2010}. More precisely, the derivative with respect to $x_{-1}$ differs by a sign. Consequently, the 2D Toda lattice presented in \cite{Biondini_2010} can be obtained from ours by  simply substituting $x_{-1}\mapsto -x_{-1}$, 
    and then replacing $V_n$ with $-V_n$. Indeed, if $\tilde V_n(x_{-1},x_1)=-V_n(-x_{-1},x_1)$, it is not difficult to show that we can rewrite \cref{eq-Toda} as
    \begin{equation}
        \frac{\partial^2}{\partial x_{-1}\partial x_1}\ln(1+\tilde V_n(x_{-1},x_1))=\tilde V_{n+1}(x_{-1},x_1)-2\tilde V_{n}(x_{-1},x_1)+\tilde V_{n-1}(x_{-1},x_1).
    \end{equation}
\end{remark}

One can show that \eqref{eq-Hirota} allows for a Wronskian type solution \eqref{wronskian}. To do this, we introduce the throughout useful \textbf{Plücker relations}.
\begin{prop}[Plücker relations]
\label{plucker relation} Let $A$ and $B$ be two matrices of size $n\times(n-2)$ and $n\times(n+1)$ respectively and denote $\det(A)=[a_0,\dots, a_{n-2}]$, and similarly for $\det (B)=[b_0,\dots,b_n]$. The following equation, referred to as the Pl\"ucker relations, is satisfied:
    \begin{align}
    \label{plucker}
        \sum_{i=0}^n (-1)^{i-1} [a_0,\dots,a_{n-2}, b_i][b_0,\dots,b_{i-1},b_{i+1},\dots, b_n]=0
    \end{align}
\end{prop}
The following result is well-known:
\begin{theorem}[\S 3.5.1, \cite{Hirota2004}]
    The $\tau$-function as given in equation \eqref{generltau} solves the $2D$-Toda lattice equation \eqref{eq-Hirota}. 
\end{theorem}
\begin{proof}
    With the notation
    $\tau^{(n)}= \text{Wr}(f^{(n)}_1,...,f^{(n)}_N)=[0,\dots,N-1]$, we find:
    \begin{gather}
        \frac{\partial}{\partial x_1}[0,\dots,N-1]= [0,\dots, N-2, N],\quad
         \frac{\partial}{\partial x_{-1}}[0,\dots,N-1]= [-1,1,\dots, N-1],\\
         \frac{\partial^2}{\partial x_{-1}\partial x_1}[0,\dots,N-1]= [-1,1,\dots, N-2,N]+[0,\dots, N-1],\quad
         \tau^{(n+1)}= [1,\dots, N],\\
         \tau^{(n-1)}=[-1,\dots, N-2].
    \end{gather}
    Substituting for the relevant terms, we find that equation \eqref{eq-Hirota} is equivalent to
    \begin{align}
        [0,\dots,N-2,N ][-1,1,\dots,N-1]-[0,\dots,N-1][-1,1,\dots,N-2,N]=[1,\dots,N][-1,\dots,N-2],
    \end{align}
    which is simply the Plücker relation for two equal matrices.
\end{proof}
Knowing that the $\tau$-function solves the $2$DTL equation, one can return to its explicit form.
Only one adaptation has to be made, and that is to restrict ourselves to the parameters $x_1,x_{-1}$, leading to:
\begin{equation}\label{eq:tau-toda}
    \begin{split}
        \tau^{(n)}(x_1,x_{-1})&=\sum_{I\in \binom{[M]}{N}} \prod_{i\in I} \lambda_i^n \exp\biggl(\sum_{i\in I} (\lambda_i x_1+ \lambda_i^{-1} x_{-1}) + \theta_0 \biggr) \Delta_{I}(A) \prod_{i>j}(\lambda_i-\lambda_j)\\
        &= \sum_{I\in \binom{[M]}{N}}  \exp(\Theta_I+ \theta_0 ) \Delta_{I}(A) \prod_{i>j}(\lambda_i-\lambda_j).
    \end{split}
\end{equation}
where $\theta_0$ is some possible constant, which we will forget about for the rest of the paper. Also note that we have conveniently summarized our parameter--dependent components as 
\begin{align}
    \Theta_I(x_{-1},x_1, n)=\sum_{i\in I} \theta_i (x_{-1}, x_1, n)= \sum_{i\in I} (\lambda_i x_1+\lambda_i^{-1}x_{-1} +n\ln(\lambda_i)).
\end{align}
Recall that soliton solutions are realized as logarithms of $\tau$-functions. So that these solutions remain \emph{regular}, we impose the following constraints on our initial data, guaranteeing the positivity of their associated $\tau$-function:
    \begin{enumerate}\label{pre-tnn}
        \item $\Delta_I(A) \geq 0$ for all $I \in \binom{[M]}{N}$;
        \item The real constants satisfy that $0 < \lambda_1< \dots  <\lambda_M$.
    \end{enumerate}
With a slight abuse of notation, the first constraint can be equivalently phrased as asking that $A \in \mathrm{Gr}(N,M)_{\geq 0}$.

Equation \eqref{eq:tau-toda} is our $\tau$-function, the one that will give us soliton solutions to the $2$D Toda lattice equation. But it is still not very clear how we can see it. Let's do an example.
\begin{example}\label{example-linesoliton}
    We set $N=1$ and $a_{ij}=(1,1)$, and find the following solution:
    \begin{equation}
        \tau^{(n)}=f^{(1)}_n=\lambda_1^n  \exp\biggl[\lambda_1 x_1 +\frac{x_{-1}}{\lambda_1}+\theta_0\biggr]+\lambda_2^n\exp\biggl[\lambda_2 x_1 +\frac{x_{-1}}{\lambda_2}+\theta_0\biggr].
    \end{equation}
    It is convenient as well as illuminating to write: $\theta_k=n\ln \lambda_k+\lambda_k x_1 +\frac{x_{-1}}{\lambda_k}+\theta_0$.\\
    Substituting for the $\tau$-function in the equation for $V_n$, we obtain the expression
    \begin{align}\label{ex-V}
        V_n = \frac{\partial^2}{\partial x_1 \partial x_{-1}} \ln (\exp (\theta_1) + \exp(\theta_2))
        =\frac{1}{4}(\lambda_1-\lambda_2)\biggl(\frac{1}{\lambda_1}-\frac{1}{\lambda_2}\biggr)\mathrm{sech}^2 \frac{1}{2}(\theta_1-\theta_2).
    \end{align}
    Since $\lambda_1<\lambda_2$, $V_n$ is negative at each point of the plane. Its absolute value is negligibly small everywhere apart from a single maximum at $\theta_1=\theta_2$. This equality  $n( \ln \lambda_1-\ln \lambda_2)= x_1 (\lambda_2-\lambda_1) +x_{-1} (\lambda_2^{-1}-\lambda^{-1}_1)$, when keeping one of the variables $x_1,x_{-1},n$ fixed describes a line in the plane. This line is a simple example of a contour plot, defined in \cref{sec-contourplots}.
    \begin{figure}[H]
        \centering
        \includegraphics[width=0.4\linewidth]{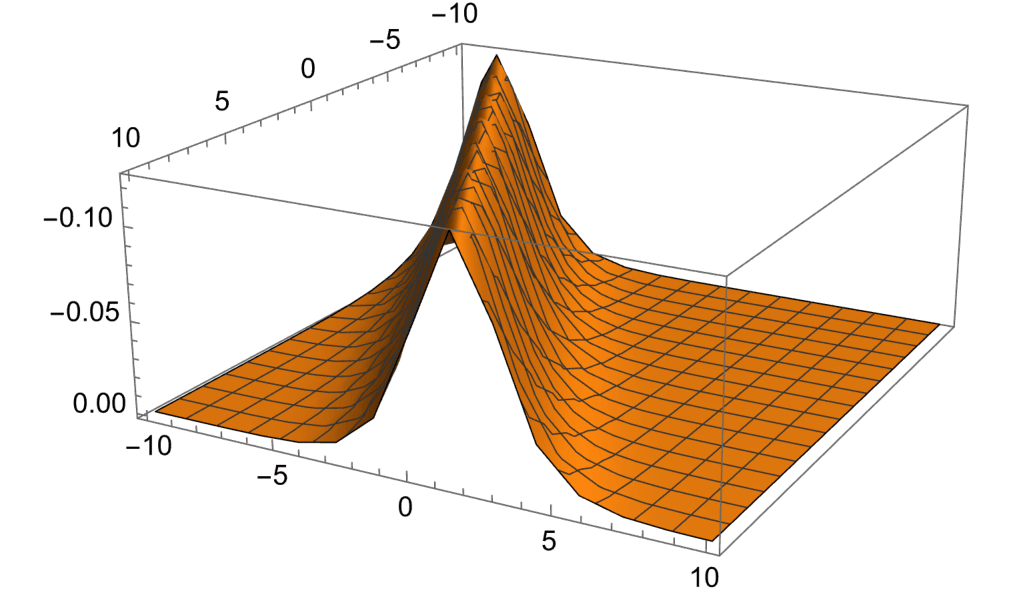}
        \includegraphics[width=0.3\linewidth]{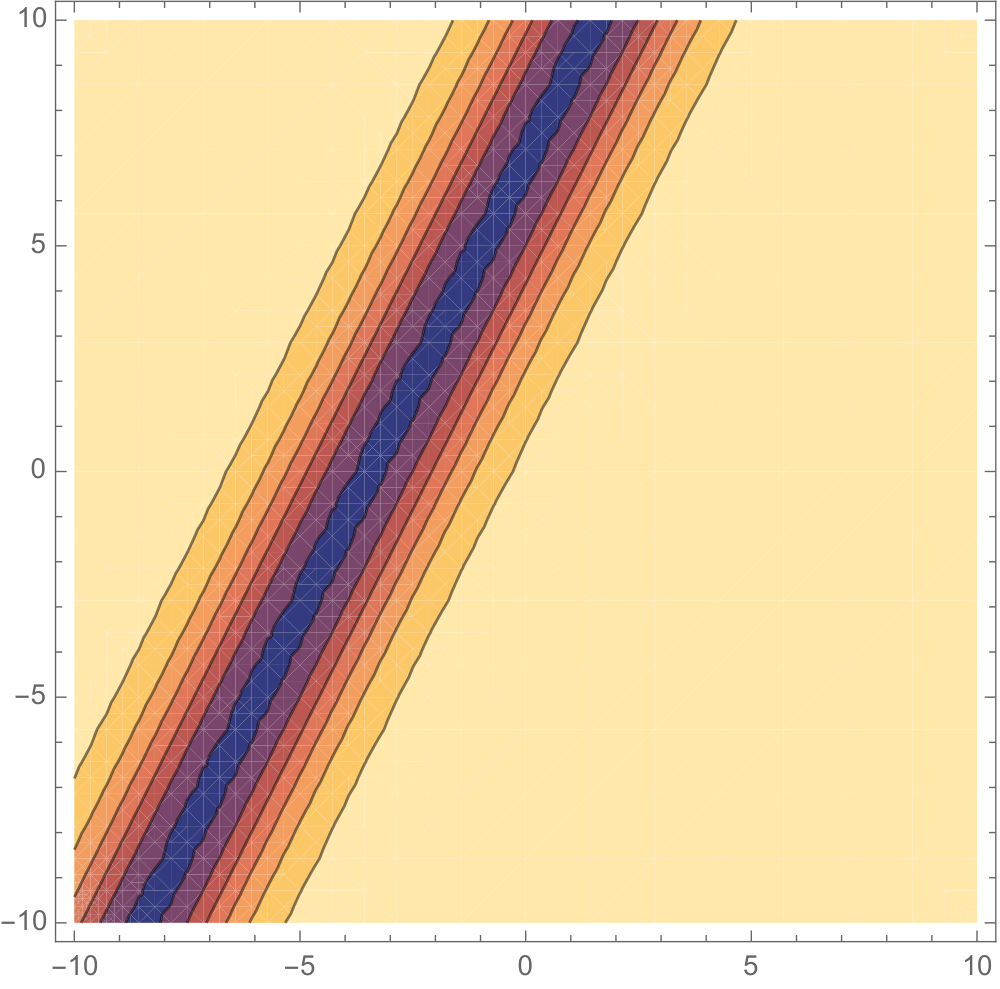}
        \caption{A plot of $V_n$ from equation \eqref{ex-V} (left), and of its corresponding $\tau$--function (right).}
        \label{fig:1sol}
    \end{figure}
\end{example}
For more difficult examples, the structure of these plots will remain: phases $\Theta_I$ domineering sections of the plane and harsh minima where they equate $\Theta_I=\Theta_J$, represented as lines in a contour plot. In order to identify each section in the plane uniquely with a dominant phase $\Theta$, we will demand that these equalities are not trivially satisfied, leading to a genericity condition. This leads us to identifying each region with a subset of indices, so that one might already get the idea that this analytical problem may be translatable into a combinatorial one. Here is the plan on how to see this more clearly and rigorously:
Start in \cref{sec-contourplots} with a formal description of the contour plot. Then an intermission is needed in \cref{sec:grassmannian}, giving some background on Grassmannian representations. Finally, one finds an identification of the soliton graph with some combinatorial object, called a $\reflectbox{L}$-diagram, in \cref{sec-Asymptotic,sec:algorithm}.

In the next subsection we take a break from the 2DTL, and introduce the Davey--Stewartson equation.
\subsubsection{The Davey-Stewartson equation}
In this section, we study the Davey-Stewartson equation II type (DSII) defocusing case and write solutions in terms of $\tau$-functions. 
Moreover, we get a closed expression of the latter using the Wronskian method, as we did for the $2$D Toda lattice, and finally we discuss when such $\tau$-functions provide soliton solutions for this equation.\\
Note that in total, there are four DS models: focusing/defocusing DSI/II, all of which  only differ in some of the signs in equation \eqref{DS}. However, these different signs are physically significant, as only the defocusing DSII experiences resonant line solitons \cite{biondini2022soliton}. Therefore, we only discuss the DSII system of defocusing type in this work.\\
Let $q(x,y,t)$ and $Q(x,y,t)$ be complex and real valued functions, respectively, with $x,y,t\in \mathbb{R}$. We denote by $q^*(x,y,t)$ the complex conjugate of $q(x,y,t)$ and by $|q|^2=qq^*$ the usual norm in $\mathbb C$. Following \cite{biondini2022soliton}, the DSII defocusing case for $q$ and $Q$ is given by:
\begin{equation}
\label{DS}
    \begin{split}
        &iq_t+\frac{1}{2}q_{xx}-\frac{1}{2}q_{yy}+2qQ+4|q|^2q=0\\
        &Q_{xx}+Q_{yy}=-4(|q|^2)_{xx}.
    \end{split}
\end{equation}
We can express $q$ and $Q$ in terms of $\tau$-functions as follows:
\begin{equation}\label{eq-tau-functions}
    q(x,y,t)=\frac{\tau^{(n+1)}(x,y,t)}{\tau^{(n)}(x,y,t)}e^{4it},\qquad Q=\frac{\partial^2}{\partial x^2}\ln(\tau^{(n)}(x,y,t)),
\end{equation}
where  $\tau^{(n+1)}(x,y,t)=(\tau^{(n-1)}(x,y,t))^*$. Similarly to the 2DTL system, one may plug these changes of variables into the differential equations and obtain bilinear expressions. For the sake of brevity and beauty we do not include the direct result, but introduce the Hirota derivatives first:
\begin{equation}
\label{Hirota deriv}
    \begin{split}
        &D_x^mf\cdot g=\left(\frac{\partial}{\partial x}-\frac{\partial}{\partial x'}\right)^mf(x,y,t)g(x',y,t)\bigg|_{x'=x}\\
        &D_y^mf\cdot g=\left(\frac{\partial}{\partial y}-\frac{\partial}{\partial y'}\right)^mf(x,y,t)g(x,y',t)\bigg|_{y'=y}\\
        &D_t^mf\cdot g=\left(\frac{\partial}{\partial t}-\frac{\partial}{\partial t'}\right)^mf(x,y,t)g(x,y,t')\bigg|_{t'=t}
    \end{split}
\end{equation}
Substituting \eqref{eq-tau-functions} in the DSII equation, and using \eqref{Hirota deriv} for further simplifications, one finds:
\begin{align}
    &(2iD_t+D_x^2-D_y^2)\tau^{(n+1)}\cdot\tau^{(n)}=0, \label{eq-Hirota1}\\
    &(D_x^2+D_y^2)\tau^{(n)}\cdot\tau^{(n)}+8(\tau^{(n+1)}\tau^{(n-1)}-(\tau^{(n)})^2)=0 \label{eq-Hirota2}.
\end{align}
In order to identify this $\tau$ with the generic form we defined in \eqref{generltau}, we use the following change of variables:
\begin{equation}
    x_1=ix+y,\qquad x_{-1}=-ix+y,\qquad x_{2}=-it,\qquad x_{-2}=it,
\end{equation}
or, equivalently,
\begin{equation}
    x=\frac{1}{2i}(x_1-x_{-1}),\qquad y=\frac{1}{2}(x_1+x_{-1}),\qquad t=\frac{1}{2i}(x_{-2}-x_2).
\end{equation}
In this setup, \cref{eq-Hirota1,eq-Hirota2} become:
\begin{align}
    &(D_{x_2}-D_{x_1}^2)\tau^{(n+1)}\cdot\tau^{(n)}=0, \label{eq-newHirota1}\\
    &(D_{x_{-2}}-D_{x_{-1}}^2)\tau^{(n+1)}\cdot\tau^{(n)}=0, \label{eq-newHirota2}\\
    &D_{x_1}D_{x_{-1}}\tau^{(n)}\cdot\tau^{(n)}=2((\tau^{(n)})^2-\tau^{(n-1)}\tau^{(n+1)}) \label{eq-newHirota3}.
\end{align}
\begin{remark}
    Note that \eqref{eq-newHirota3} is the Hirota bilinear equation of the $2$D Toda lattice equation, while \cref{eq-newHirota1,eq-newHirota2} are referred to as Bäcklund transformations of the KP hierarchy (see \cite[\S 3.1]{kodama} for details).
\end{remark}
With this remark, it is not hard to believe that:
\begin{lemma}[Lemma 2.1, \cite{biondini2022soliton}]
    The function
    \begin{align}
        \tau^{(n)}(x_{-2},x_{-1},x_1,x_2)=\sum_{I\in \binom{[M]}{N}} \prod_{i\in I} \lambda_i^n \exp\biggl(\sum_{i\in I}\lambda_i^2 x_2 +\lambda_i x_1+ \lambda_i^{-1} x_{-1}+ \lambda_i^{-2} x_{-2} + \theta_0 \biggr) \Delta_{I}(A) \prod_{i>j}(\lambda_i-\lambda_j),
    \end{align}
    of the general form \eqref{wronskian}, is a solution to Equations \eqref{eq-newHirota1}-\eqref{eq-newHirota3}.
\end{lemma}
The proof of this lemma proceeds as in the case of the $2$D Toda lattice, whereby the bilinearization of the Davey--Stewartson equation is reduced to the Plücker relations of the matrix A.\\ 
Having convinced ourselves that the $\tau$-function is of Wronskian type, we may note that $x_{-2},x_{-1},x_1,x_2\in \mathbb{C}$ as well as $\lambda_i\in \mathbb{C}$, $\forall i\in \{1,\dots,M\}$. The formalism in \cite{kodama2014kp}, however, lives in the real plane; it is therefore preferable to work in the original real parameters $x,y,t\in \mathbb{R}$ and introduce a set of real parameters $\{\psi_j\}_{j=1}^M$ satisfying $\lambda_j=e^{-i\psi_j}$. This choice of parameters yields real functions $\theta_j(x,y,t)$. Indeed, we may rewrite
\begin{equation}
    \begin{split}
        \theta_j(x,y,t) &=\lambda_jx_1+\lambda_j^{-1}x_{-1}+\lambda_j^2x_2+\lambda_j^{-2}x_{-2}\\
        &=i(\lambda_j-\lambda_j^{-1})x+(\lambda_j+\lambda_j^{-1})y+i(\lambda_j^{-2}-\lambda_j^2)t\\
        &=\frac{1}{i}(e^{i\psi_j}-e^{-i\psi_j})x+(e^{i\psi_j}+e^{-i\psi_j})y-\frac{1}{i}(e^{2i\psi_j}-e^{-2i\psi_j})t \\
        &=2(\sin(\psi_j)x+\cos(\psi_j)y-\sin(2\psi_j)t).
    \end{split}
\end{equation}
This implies that we may re-frame the Davey--Stewartson $\tau$-functions as follows:
\begin{lemma}[\cite{kodama}, \S 3.2.3]\label{lemma-tau-functionExpression}
    Let $A \in \mathrm{Gr}(N,M)$  and $\{\psi_j\}_{j=1}^M$ be a family of real parameters. Then the following $\tau$-functions solve the system of equations \eqref{DS}:
    \begin{equation}
        \tau_A^{(n)}(x,y,t)=(2i)^{\frac{N(N-1)}{2}}\sum_{I\in\mathcal M(A)}\Delta_I(A)S_I(\psi)e^{-i(n+\frac{N-1}{2})\Psi_I}e^{\Theta_I(x,y,t)},
    \end{equation}
    where $\mathcal M(A)$ is the matroid of $A$,
    \begin{align}
        &\Theta_I(x,y,t)=\sum_{i\in I}\theta_i(x,y,t),\qquad\Psi_I=\sum_{i\in I}\psi_i,\qquad S_I(\psi)=\prod_{j<k}\sin\left(\frac12(\psi_{i_j}-\psi_{i_k})\right).
    \end{align}
\end{lemma}
We are interested in the combinatorics of solitons in the $(x,y)$-plane arising from $Q$. Soliton solutions are regular in the whole $(x,y)$-plane. Since $Q$ is defined using the logarithm of the $\tau$-function, we ask that $A \in \mathrm{Gr}(N,M)_{\geq 0}$. Finally, we have \textit{reality constraints} $\tau$-functions must satisfy to get solutions $q$ and $Q$:
\begin{enumerate}
    \item $\tau^{(n)}$ must be real;
    \item $\tau^{(n+1)}=(\tau^{(n-1)})^*$.
\end{enumerate}
We satisfy these constraints by replacing $\tau^{(n)}$ by $(2i)^{-\frac{N(N-1)}{2}}\tau^{(n)}$ and imposing $n=-\frac{N-1}{2}$.

\begin{example}\label{example-linesolitonDS}
    Consider $N=1$ and $M=2$. Let $\tau^{(n)}(x,y,t)=e^{\theta_1-in\psi_1}+e^{\theta_2-in\psi_2}$, for $\psi_1<\psi_2$. Note that the reality constraints give us $n=0$. We have that
    \begin{align}
        Q(x,y,t)&=\frac{\partial^2}{\partial x^2}\ln(\tau^{(0)})\\
        &=\frac{\partial^2}{\partial x^2}\ln\left(2e^{\frac{1}{2}(\theta_1+\theta_2)}\cosh\left(\frac{1}{2}(\theta_1-\theta_2)\right)\right)\\
        &=(\sin(\psi_1)-\sin(\psi_2))^2\text{sech}^2\left(\frac12(\theta_1-\theta_2)\right).
    \end{align}
    For fixed $t\in\mathbb R$, $Q$ describes a wave, the crest of which is reached when the argument of the hyperbolic secant is equal to $0$, i.e. along the line $\theta_1=\theta_2$ that separates the regions of the plane where $e^{\theta_1}$ dominates with respect to $e^{\theta_2}$, and conversely. Finally, notice that the amplitude of the wave is given by $\sin(\psi_1)-\sin(\psi_2)$, which is zero for horizontal solitons. 
\end{example}

\subsection{Contour plots}\label{sec-contourplots}
In this section, the \textbf{contour plot} is formally defined for both systems. In Section \ref{sec:wronsk} we found that both the $2$DTL and the DSII equations may be solved in terms of $\tau$-functions of Wronskian type, which are expressible as  a combination of exponential functions with linear arguments, dependent only on a finite number of variables (cf. \cref{generltau}). We also found that plotting the results of these models leads to a subdivision of $\mathbb{R}^2$ into chambers. In each chamber an exponential function appearing in $\tau_A$ dominates with respect to the others. The chambers are bounded by lines that are obtained by the balancing between two dominant exponential functions. More precisely, if we call $\Theta_I$ and $\Theta_J$ two adjacent dominant exponential functions, we can approximate $\tau_A$ as the sum of those two exponential functions. Therefore, $u_A$ behaves as the square of the hyperbolic secant of $\Theta_I-\Theta_J$, whose peak is reached when the argument is $0$, providing the bounding line $\Theta_I=\Theta_J$.

It is interesting to note that the union of the chambers, lines, and vertices is a polyhedral complex satisfying balancing conditions at the vertices. In other words, the contour plot is a tropical curve in the plane \cite{maclagan2015introduction}. Now, we specialize this to the cases of study.

\subsubsection{2D Toda lattice}\label{sec-contourplotTL}

For a fixed $I \in \binom{[M]}{N}$ and any collection of real parameters  $\lambda_1, \dots, \lambda_N$, let us define the constants:
\begin{gather}
    \Lambda_I=\sum_{i\in I}\ln \lambda_i, \hspace{2em} Sh_I(\lambda)=\prod_{j<i\in I}\sinh\frac12(\ln \lambda_i-\ln \lambda_j).
\end{gather}
Then we can use the equality
\begin{equation}
    \prod_{j<i \in I}(\lambda_i-\lambda_j)=2^{\frac{N(N-1)}{2}}e^{\frac{N-1}{2}\Lambda_I}Sh_I(\lambda)
\end{equation}
to rewrite \cref{eq:tau-toda} as
\begin{equation}\label{eq-tauexpression}
    \tau^{(n)}_A(x_{-1},x_1)=2^{\frac{N(N-1)}{2}}\sum_{I\in\mathcal M(A)}e^{\frac{N-1}{2}\Lambda_I}\Delta_I(A)Sh_I(\lambda)e^{\Theta_I(x_{-1},x_1,n)},
\end{equation}
where we have set the constant term $\theta_0$ to be equal to zero. 

Let $A\in\mathrm{Gr}(N,M)_{\geq0}$ and define
\begin{equation}
    f^{(n)}_{\mathcal M(A)}(x_{-1},x_1)=\max_{I\in\mathcal M(A)}\biggl\{\Theta_I(x_{-1},x_1,n)+\biggl(\frac{N-1}{2}\biggr)\Lambda_I+\ln(\Delta_I(A)Sh_I(\lambda))\biggr\}.
\end{equation}
The \emph{corner locus}\footnote{By which we mean, the locus where the function $f^{(n)}_{\mathcal M(A)}(x_{-1},x_1)$ is not linear.} of the function $f^{(n)}_{\mathcal M(A)}(x_{-1},x_1)$ defines a tropical curve. Thus, we define the contour plot of a soliton solution as follows:

\begin{definition}\label{def-TLcontourplot}
    Let $n\in\mathbb Z$ be fixed and $V_{n,A}(x_{-1},x_1)$ be a solution of the $2$D Toda lattice equation constructed from $A\in\text{Gr}(N,M)_{\geq0}$. The \textbf{contour plot} of the solution $V_{n,A}(x_{-1},x_1)$ is the corner locus in the $(x_1, x_{-1})$--plane of the function $f^{(n)}_{\mathcal M(A)}(x_{-1},x_1)$, and is denoted by $\mathcal C(V_{n,A})$.
\end{definition}

\begin{remark}
    We are also interested in the study of the soliton contour plot when $x_{-1}$ is fixed. The definition is the same, but in this case the contour plot is sketched in the $(x_1,n)$-plane instead of the $(x_1,x_{-1})$-plane.
\end{remark}

\begin{remark}
    As pointed out in \cref{rmk-biondini} and illustrated in \cref{example-linesoliton}, our convention makes the real regular soliton solution $V_{n,A}(x_{-1},x_1)$ negative. Consequently, the contour plots of $V_{n,A}(x_{-1},x_1)$ and $|V_{n,A}(x_{-1},x_1)|$ coincide. Hence, the corner locus of $f^{(n)}_{\mathcal M(A)}(x_{-1},x_1)$ provides the contour plot of $|V_{n,A}(x_{-1},x_1)|$, and therefore also of $V_{n,A}(x_{-1},x_1)$. This justifies \cref{def-TLcontourplot}.
\end{remark}

Furthermore, we are interested in describing the \textit{asymptotic} behavior of soliton solutions, by which we mean, for large scale variables $(x_{-1},x_1,n)$. In particular, the constant terms $\displaystyle\frac{N-1}{2}\Lambda_I+\ln(\Delta_I(A)Sh_I(\lambda))$ do not contribute. We therefore approximate the function $f^{(n)}_{\mathcal M(A)}(x_{-1},x_1)$ by

\begin{equation}
    g^{(n)}_{\mathcal M(A)}(x_{-1},x_1)=\max_{I\in\mathcal M(A)}\{\Theta_I(x_{-1},x_1,n)\}=\lim_{\varepsilon\to0}\varepsilon f^{(\frac{n}{\varepsilon})}_{\mathcal M(A)}\biggl(\frac{x_{-1}}{\varepsilon},\frac{x_1}{\varepsilon}\biggr).
\end{equation}

\begin{definition} \label{def_cont_toda}
    Let $n\in\mathbb Z$ fixed and $V_{n,A}(x_{-1},x_1)$ be a solution of the $2$D-Toda lattice equation constructed from $A\in\text{Gr}(N,M)_{\geq0}$. The \textbf{asymptotic contour plot} of the solution $V_{n,A}(x_{-1},x_1)$ is the corner locus in the $(x_1,x_{-1})$-plane of the function $g^{(n)}_{\mathcal M(A)}(x_{-1},x_1)$, and is denoted by $\mathcal C^{(n)}(\mathcal M(A))$.
\end{definition}

\begin{remark}
    The contour plots $\mathcal{C}(V_{n,A})$ and $\mathcal{C}^{(n)}(\mathcal{M}(A))$ have analogues for when we fix the variable $x_{-1}$ instead of $n$, which we denote by $\mathcal{C}(V_{{x_{-1}},A})$ and $\mathcal{C}^{(x_{-1})}(\mathcal{M}(A))$, respectively
\end{remark}

\begin{remark}
    The contour plot locates the precise position of the peaks of the waves generated by the soliton solutions. However, our algorithms can only be applied in the $(x_1,x_{-1})$-plane (resp. in the $(x_1,n)$-plane) when $n$ (resp. $x_{-1}$) is sufficiently large. Hence, the object of study for this paper is the asymptotic contour plot.
\end{remark}

\subsubsection{Davey-Stewartson}

Recall how, once we impose reality constraints, solutions of the DSII system in the defocusing case are of the form
\begin{equation}
    \tau_A^{(-\frac{N-1}{2})}(x,y,t)=\tau_A(x,y,t)=\sum_{I\in\mathcal M(A)}\Delta_I(A)S_I(\psi)e^{\Theta_I(x,y,t)},
\end{equation}
where $S_I(\psi)=\prod_{j<i\in I}\sin\frac12(\psi_i-\psi_j)$ and $\Theta_I(x,y,t)=\sum_{i\in I}2(x\sin\psi_i+y\cos\psi_i-t\sin(2\psi_i))=\sum_{i\in I}\theta_i(x,y,t)$. Let $A\in\text{Gr}(N,M)_{\geq0}$, and define the two functions
\begin{align}
    &f_{\mathcal M(A)}(x,y,t)=\max_{I\in\mathcal M(A)}\{\Theta_I(x,y,t)+\ln(\Delta_I(A)S_I(\psi))\},\\
    &g_{\mathcal M(A)}(x,y,t)=\max_{I\in\mathcal M(A)}\{\Theta_I(x,y,t)\}=\lim_{\varepsilon\to0}\varepsilon f_{\mathcal{M}(A)}\left(\frac{x}{\varepsilon},\frac{y}{\varepsilon},\frac{t}{\varepsilon}\right).
\end{align}

\begin{definition}
    Let $t\in\mathbb R$ fixed and $Q_A(x,y,t)$ be a solution of the DSII equation defocusing case constructed from $A\in\text{Gr}(N,M)_{\geq0}$. The \textbf{contour plot} of the solution $Q_A(x,y,t)$ is the corner locus in the $(x,y)$--plane of the function $f_{\mathcal M(A)}(x,y,t)$, and is denoted by $\mathcal C(Q_{t,A})$.
\end{definition}

\begin{definition}
    Let $t\in\mathbb R$ be fixed and $Q_A(x,y,t)$ be a solution of the DSII equation defocusing case constructed from $A\in\text{Gr}(N,M)_{\geq0}$. The \textbf{asymptotic contour plot} of the solution $Q_A(x,y,t)$ is the corner locus in the $(x,y)$--plane of the function $g_{\mathcal M(A)}(x,y,t)$, and is denoted by $\mathcal {C}^{(t)}(\mathcal{M}(A))$.
\end{definition}

\begin{remark}
    Also in this case, our algorithm applies when $t\ll0$. Thus, we study the asymptotic contour plot of the soliton solution.
\end{remark}

\subsection{Decompositions of Grassmannians}
\label{sec:grassmannian}
In this subsection, we briefly recall the necessary facts on real Grassmannians and their decompositions. For a more detailed account, we defer to \cite{talaska2013network}. We also found \cite[Section 5]{kodama} to be an accessible reference.

We have reviewed how soliton solutions of the 2D Toda lattice and DSII integrable systems are parametrized by points in Grassmannians. As it happens, there exist sufficiently fine decompositions of the Grassmannians whereby knowing the component we source the matrix parameter from determines non-trivial information about the corresponding soliton solution. Grassmannians admit various decompositions that are governed by combinatorial data. We outline the relevant decompositions for our work in the following. 

\subsubsection{Schubert decomposition}
To begin with, we start with the classical Schubert decomposition. Each linear space $L\in\mathrm{Gr}(k,n)$ is represented by a $k\times n$ matrix $A$ of rank $k$ in reduced row echelon form with pivot set $I$, where $I\subset[n]$. For fixed $I\subset[n]$, we define the \textit{Schubert cell} $\Omega_I$ as the set of those linear spaces yielding pivot set $I$. We define the Schubert decomposition as
\begin{equation}
    \mathrm{Gr}(k,n)=\bigcup_{I\in\binom{[n]}{k}}\Omega_I.
\end{equation}

The Schubert decomposition may equivalently be indexed by permutations. Let $I=\{i_1<\dots<i_k\}\in\binom{[n]}{k}$, then we define an associated partition of $\mathrm{dim}(\Omega_I)$, $\lambda=(\lambda_1, \dots, \lambda_k)$, via 
\begin{equation}
    \lambda_j=n-k-(i_j-j).
\end{equation}

Let $S_n$ be the symmetric group on $n$ elements and define $s_j=(j\,j+1)$ for $j=1,\dots,n-1$. Moreover, we define a subgroup 
\begin{equation}
    P_k=\langle s_1,\dots,\hat{s}_{n-k},\dots,s_{n-1}\rangle\cong S_{n-k}\times S_k,
\end{equation}
and the quotient
\begin{equation}
    S_n^{(k)}=\faktor{S_n}{P_k}.
\end{equation}

The quotient $S_n^{(k)}$ may be more explicitly described. To do this, we call an expression $\sigma=s_{i_1}\cdots s_{i_l}$ for $\sigma\in S_n$ \textbf{reduced} if $l$ is minimal. We call $l$ the length of $\sigma$, and denote it by $\ell(\sigma)$.

Then, we may express $S_n^{(k)}$ as
\begin{equation}
    S_n^{(k)}=\{\sigma\in S_n\mid \ell(\sigma\tau)\ge\ell(\sigma)\,\textrm{for all}\,\tau\in P_k\}.
\end{equation}

Next, we introduce a combinatorial game to relate the partitions $\lambda$ to elements of $S_n^{(k)}$. Consider the following $k\times (n-k)$ rectangle:

\begin{table}[H]
    \begin{tabular}{|c|c|c|}
        \hline $s_{n-k}$ &  \makebox[3cm]{$\cdots$}  &  $s_1$\\
          \hline$\vdots$ & $\ddots$ & $\vdots$\\
          \hline $s_{n-1}$ & $\cdots$ & $s_k$\\
          \hline
    \end{tabular}
\end{table}

By construction, we have $\lambda\subset (n-k)^k$. We consider the Young diagram $Y_\lambda$ of $\lambda$ embedded into the rectangle above in English convention.

We define a \textbf{reading order} of $Y_\lambda$ to be a labeling of its boxes by $1,\dots,|\lambda|$, such that the labels increase going from right to left in the rows and from bottom to top in the columns. Given a reading order of $Y_\lambda$ and taking the product over all $s_j$ in the boxes chronologically along the reading order, we obtain a permutation $w$. We have $w\in S_n^{(k)}$ and the thus obtained expression is a reduced expression. In fact, all reduced expressions of $w\in S_n^{(k)}$ arise from a reading order of $Y_\lambda$. We write $\mathbf{w}$ when we want to specify a reading order expressing $w$.

\begin{example}
\label{ex-filling}
Let $n=6$ and $k=3$. Let $I=\{1,3,5\}\in\binom{[6]}{3}$. We obtain $\lambda=(3,2,1)$. We obtain the following Young diagram $Y_\lambda$ with reflections $s_j$ as follows

\begin{equation}
    \ytableaushort{{s_3}{s_2}{s_1},{s_4}{s_3},{s_5}}
\end{equation}

We consider the following two reading orders:

\begin{equation}
    \ytableaushort{{6}{5}{4},{3}{2},{1}}\quad\textrm{and}\quad \ytableaushort{{6}{5}{3},{4}{1},{2}}
\end{equation}

The first reading order yields the following reduced expression
\begin{equation}
    s_5s_3s_4s_1s_2s_3=(5\,6)(3\,4)(4\,5)(1\,2)(2\,3)(3\,4)=(1\,2\,4)(3\,6\,5),
\end{equation}
whereas the second reading order yields
\begin{equation}
    s_3s_5s_1s_4s_2s_3=(3\,4)(5\,6)(1\,2)(4\,5)(2\,3)(3\,4)=(1\,2\,4)(3\,6\,5).
\end{equation}
\end{example}

Indeed, this construction of a permutation $w$ from $\lambda$ induces a bijection between permutations $w$ in $S_n^{(k)}$ and partitions $\lambda_w$ in $(n-k)^k$. From this, we obtain a re-indexing of the Schubert decomposition:

\begin{equation}
    \mathrm{Gr}(k,n)=\bigcup_{w\in S_n^{(k)}}\Omega_w,
\end{equation}

where $\mathrm{dim}(\Omega_w)=l(w)=|\lambda_w|$. We denote the restriction of the decomposition to the totally non--negative Grassmannian as

\begin{equation}
    \mathrm{Gr}(k,n)_{\ge0}=\bigcup_{w\in S_n^{(k)}}\Omega_w^{>0}.
\end{equation}

\subsubsection{Deodhar decomposition}
The Deodhar decomposition is a refinement of the Schubert decomposition. It is indexed by pairs of permutations. We first need to introduce some combinatorial decorations of Young diagrams.

\begin{definition}
    A $\Le$--diagram is a filling of the boxes of a Young diagram by $\LeBox$ and $\LeCircle$ satisfying the $\Le$--property, i.e. a filled  box $\LeCircle$ cannot have    an empty box $\LeBox$ above \emph{and} to the left of it. A $\Le$--diagram is called \emph{irreducible} if each row and column contains at least one blank box $\LeBox$.
\end{definition}

As proved in \cite{postnikov2006total}, the set of irreducible $\Le$--diagrams contained in $k\times (n-k)$ rectangles is in bijection with the permutations without fixed points in $S_n$, called \emph{derangements}. We explain how to construct a derangement from a $\Le$--diagram. For this, we use a combinatorial game on Young diagrams called \emph{reduced pipedreams}. 

Given a $\Le$--diagram, we replace a blank box $\LeBox$ with a box containing \emph{elbow pipes} $\inlinebridge$ and a box filled $\LeCircle$ with the \emph{cross} $\inlinecross$. We label the south east boundary edges by $1,\dots,n$ proceeding from the top to the bottom.
Now, we label the north-west  by its opposing edge on the south-east boundary.
That is, we traverse along a pipe starting at the south--east boundary edge labeled $i$
take a right at a white vertex a left at a black one, giving the north-west boundary its first index. 
Then again, start at a label $i$ at the south--east boundary and follow the pipe with the following adjusted rule: at a white vertex,  turn left, while at a black vertex, turn right.  Again, we end up at a north--west border edge giving it a second label. We write $[i,j]$ at a north--west border edge that obtains the labels $i$ and $j$ through these two procedures. For example, the $\Le$--diagram below
\begin{equation}
\begin{tikzpicture}[scale=0.8, baseline=(current bounding box.center)]
\useasboundingbox (-0.8, 0.5) rectangle (3.2, -3.2);

    \draw (0,0) rectangle (1,-1); 
    \draw (1,0) rectangle (2,-1); 
    \draw (2,0) rectangle (3,-1); 
    \draw (0,-1) rectangle (1,-2); 
    \draw (0,-2) rectangle (1,-3); 
    \draw (1,-1) rectangle (2,-2); 

    \drawwhitestone{0}{0}     
    \drawwhitestone{0}{-1} 

\end{tikzpicture}
\quad \longrightarrow \quad 
\begin{tikzpicture}[scale=0.8, baseline=(current bounding box.center)]
\useasboundingbox (-0.8, 0.5) rectangle (3.2, -3.2);

    \draw (0,0) rectangle (1,-1); 
    \draw (1,0) rectangle (2,-1); 
    \draw (2,0) rectangle (3,-1); 
    \draw (0,-1) rectangle (1,-2); 
    \draw (0,-2) rectangle (1,-3); 
    \draw (1,-1) rectangle (2,-2); 

    \node[right] at (3,-0.5) {$1$};
    \node[below] at (2.5,-1) {$2$};
    \node[right] at (2,-1.5) {$3$};
    \node[below] at (1.5,-2) {$4$};
    \node[right] at (1,-2.5) {$5$};
    \node[below] at (0.5,-3) {$6$};

    \node[above] at (2.5,0) {$2$};
    \node[above] at (1.5,0) {$4$};
    \node[above] at (0.5,0) {$6$};
    \node[left] at (0,-0.5) {$1$};
    \node[left] at (0,-1.5) {$3$};
    \node[left] at (0,-2.5) {$5$};

    \drawbridge{1}{0}       
    \drawcross{0}{0}
    \drawbridge{2}{0}
    \drawcross{0}{-1}
    \drawbridge{0}{-2}
    \drawbridge{1}{-1}
\end{tikzpicture}
\quad \longrightarrow \quad 
\begin{tikzpicture}[scale=0.8, baseline=(current bounding box.center)]
\useasboundingbox (-0.8, 0.5) rectangle (3.2, -3.2);

    \draw (0,0) rectangle (1,-1); 
    \draw (1,0) rectangle (2,-1); 
    \draw (2,0) rectangle (3,-1); 
    \draw (0,-1) rectangle (1,-2); 
    \draw (0,-2) rectangle (1,-3); 
    \draw (1,-1) rectangle (2,-2); 

    \node[right] at (3,-0.5) {$1$};
    \node[below] at (2.5,-1) {$2$};
    \node[right] at (2,-1.5) {$3$};
    \node[below] at (1.5,-2) {$4$};
    \node[right] at (1,-2.5) {$5$};
    \node[below] at (0.5,-3) {$6$};

    \node[above] at (2.5,0) {$[2,1]$};
    \node[above] at (1.5,0) {$[4,2]$};
    \node[above] at (0.5,0) {$[6,5]$};
    \node[left] at (0,-0.5) {$[1,3]$};
    \node[left] at (0,-1.5) {$[3,4]$};
    \node[left] at (0,-2.5) {$[5,6]$};

    \drawbridge{1}{0}       
    \drawcross{0}{0}
    \drawbridge{2}{0}
    \drawcross{0}{-1}
    \drawbridge{0}{-2}
    \drawbridge{1}{-1}
\end{tikzpicture}
%
%
\end{equation}
produces the derangement $(1\,3\,4\,2)(5\,6)$.

  In fact, through this procedure each edge segment obtains two labels, except for the starting edges. Removing these starting segments with only one label as well as the adjacent vertices one obtains the \emph{reduced pipedreams}.
  These reduced pipedreams become important
  for our later purposes, as they are closely connected to contour plots. Additionally, we can label each region of the reduced pipedream as follows: label the southeast region by I; lines $[i,j]$ act by transposing the elements $i$ and $j$ in their respective index sets, so that adjacent regions always share $k-1$ indices\footnote{As we will see in \cref{sec-2DTodalattice} and \cref{sec:Davey-Stewartson}, these labels are related to the labels of dominant exponentials in asymptotic contour plots.}. These notions are illustrated for our running example in \cref{fig:region-labels}. We now relate these combinatorics to the permutation setting above.
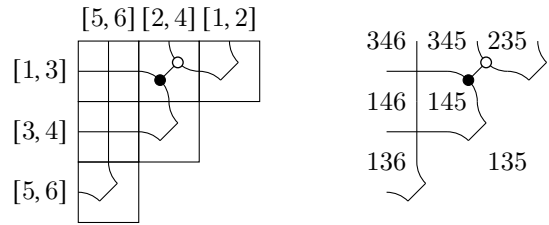
\begin{figure}[h]
    \centering
   \begin{tikzpicture}[scale=0.8, baseline=(current bounding box.center)]
\useasboundingbox (-0.8, 0.5) rectangle (3.2, -3.2);

    \draw (0,0) rectangle (1,-1); 
    \draw (1,0) rectangle (2,-1); 
    \draw (2,0) rectangle (3,-1); 
    \draw (0,-1) rectangle (1,-2); 
    \draw (0,-2) rectangle (1,-3); 
    \draw (1,-1) rectangle (2,-2);

    \node[above] at (2.5,0) {$[1,2]$};
    \node[above] at (1.5,0) {$[2,4]$};
    \node[above] at (0.5,0) {$[5,6]$};
    \node[left] at (0,-0.5) {$[1,3]$};
    \node[left] at (0,-1.5) {$[3,4]$};
    \node[left] at (0,-2.5) {$[5,6]$};

    \drawbridge{1}{0}       
    \drawcross{0}{0}
    \drawboundarybridge{2}{0}
    \drawcross{0}{-1}
    \drawboundarybridge{0}{-2}
    \drawboundarybridge{1}{-1}
\end{tikzpicture}
\hspace{2 em}
\begin{tikzpicture}[scale=0.8, baseline=(current bounding box.center)]
\useasboundingbox (-0.8, 0.5) rectangle (3.2, -3.2);

   \node at (2,-2) {$135$};
   \node at (2,0) {$235$};
   \node at (1,0) {$345$};
   \node at (0,0) {$346$};
   \node at (0,-1) {$146$};
   \node at (0,-2) {$136$};
   \node at (1,-1) {$145$};

    \drawbridge{1}{0}       
    \drawcross{0}{0}
    \drawboundarybridge{2}{0}
    \drawcross{0}{-1}
    \drawboundarybridge{0}{-2}
    \drawboundarybridge{1}{-1}
\end{tikzpicture}
    \caption{A reduced pipedream with its corresponding labeling of regions.}
    \label{fig:region-labels}
\end{figure}

\begin{definition}
    Let $w\in S_n^{(k)}$ and consider a reduced expression $\mathbf{w}=s_{i_1}\cdots s_{i_l}$. Setting some factors to $1$, we obtain a subexpression $\mathbf{v}$ of a permutation $v$. In particular, we have $v\le w$ w.r.t. the Bruhat order.

    Moreover, we define for $1\le k\le l$ the expression $\mathbf{v}_{(k)}$ as the product of the $k$ leftmost factors. We call $\mathbf{v}$ a \emph{distinguished subexpression} if

    \begin{equation}
        v_{(j-1)}\le v_{(j-1)}s_{i_j}.
    \end{equation}
    for all $j=1,\dots,l$ and write $\mathbf{v}\prec \mathbf{w}$. If the inequalities are strict, we call $\mathbf{v}$ a \emph{positive distinguished subexpression (PDS)}.
\end{definition}

If $v\le w$ and $\mathbf{w}$ is a reduced expression of $w$, then there is a unique PDS of $\mathbf{w}$ expressing $v$, which we denote by $\mathbf{v}_+$ (see \cite[Lemma 3.5]{marsh2004parametrizations}).

Let $w\in S_n^{(k)}$ be a permutation and $\Gamma$ a $\Le$--diagram on the corresponding Young diagram. We fix a reading order $\mathbf{w}=s_{i_1}\cdots s_{i_l}$. We obtain a PDS $\mathbf{v}$ by setting all factors corresponding to blank boxes to $1$. If $\Gamma$ is irreducible, we  obtain the corresponding derangement as $\pi=vw^{-1}$.

In \cite{talaska2013network} an explicit description of Deodhar decomposition via $\Le$--diagrams is derived. We need the following notions to state it.

\begin{definition}
    Let $\Gamma$ be a $\Le$--diagram on the Young diagram $Y$. Moreover, let $\mathbf{v}\prec\mathbf{w}$ be the corresponding expressions for permutations $v,w$ with $w\in S_n^{(k)}$ and $v\le w$.

    Let $b$ be a box in $\Gamma$. We denote the Young diagram weakly south--east of $b$ by $Y_b^{\textrm{in}}$, and its complement by $Y_b^{\textrm{out}}$. Recall that $\mathbf{w}$ corresponds to a reading order of $Y$. Thus, restricting to $Y_b^{\textrm{in}}$ and $Y_b^{\textrm{out}}$, we obtain reduced expressions $\mathbf{w}_b^{\textrm{in}},\mathbf{w}_b^{\textrm{out}}$ for permutations ${w}_b^{\textrm{in}},{w}_b^{\textrm{out}}$ respectively, together with PDS ${v}_b^{\textrm{in}},{v}_b^{\textrm{out}}$. (When the box $b$ is clear from the context, we may drop the notation.)
\end{definition}

The Deodhar decomposition restricted to the TNN Grassmannian is a refinement of the Schubert decomposition of the shape

\begin{equation}
    \mathrm{Gr}(k,n)_{\ge0}=\bigcup_{w\in S_n^{(k)}} \Omega_w^{>0}\quad\textrm{and}\quad \Omega_w^{>0}=\bigcup_{v\le w} \mathcal{P}_{v,w}^{>0},
\end{equation}

where $\mathcal{P}_{v,w}^{>0}$ is the \emph{Deodhar component} indexed by a permutations $v,w$ with $w\in S_n^{(k)}$ and $v\le w$. We will sometimes denote $\mathcal{P}_{v,w}^{>0}$ by $\mathcal{P}_\pi^{>0}$ if $\pi=vw^{-1}$, whenever we want to emphasize the derangement $\pi$. Deodhar components are characterized by the non--vanishing of Pl\"ucker coordinates. We use this characterization as the definition:

\begin{definition}[{\cite[Theorem 7.8]{talaska2013network}}]
    Let $w\in S_n^{(k)}$, $\mathbf{w}$ a reduced expression with PDS $\mathbf{v}\prec\mathbf{w}$. Let $\lambda$ be the partition corresponding to $w$. Moreover, let $\Gamma$ be the respective $\Le$--diagram. Let $I(\lambda)=w\{n-k+1,\dots,n\}$ the pivot set and $b$ a box in $\Gamma$. If $b$ is a blank box, we define $I_b=v^{\textrm{in}}(w^{\textrm{in}})^{-1}I(\lambda)$. If $b$ contains a white stone, we set $I_b=v^{\textrm{in}}s_b(w^{\textrm{in}})^{-1}I(\lambda)$, $s_b$ is the transposition corresponding to the box.

    Then, we define $\mathcal{P}_{v,w}^{>0}$ to be the subset of $\mathrm{Gr}(k,n)_{\ge0}$ consisting of Pl\"ucker vectors $\Delta_I$ satisfying
    \begin{enumerate}
        \item $\Delta_{I^b}=0$ if $b$ is a box with a white stone,
        \item $\Delta_{I^b}\neq0$ if $b$ is a blank box,
        \item $\Delta_{I_0}\neq0$ and
        \item $\Delta_{J}=0$ for all $J$ lexicographically smaller than $I_0$.
    \end{enumerate}
\end{definition}

\section{Duality}\label{sec-dual}

The goal of this section is to describe a correspondence between soliton solutions constructed from a point in the Grassmannian $\text{Gr}(N, M)$, and solutions constructed from a point in the dual Grassmannian $\text{Gr}(M-N,M)$.

\begin{definition}\label{def-irrmatrix}
    Let $A$ be an $N\times M$ matrix. We say that $A$ is \textbf{irreducible} if

    \begin{enumerate}
        \item there is at least one non-zero element in each column of $A$;
        \item once we put $A$ in reduced row echelon form (RREF), there is at least one non-zero element other than the pivot per row. 
    \end{enumerate}

    Equivalently, we can state the definition in terms of the matroid $\mathcal{M}(A)$: $A$ is irreducible if

    \begin{enumerate}
        \item for any $i\in[M]$, there exists $I\in\mathcal{M}(A)$ such that $i\in I$;
        \item there is no common index among the elements of $\mathcal{M}(A)$.
    \end{enumerate}
\end{definition}

Let $A \in \mathrm{Gr}(N,M)$ be any irreducible matrix. It is always expressible as
\begin{equation}
    A=
    \begin{pmatrix}
        I_N|G
    \end{pmatrix}
    P,
\end{equation}
where $I_N$ is the $N \times N$ identity matrix of pivot columns of $A$, $G$ is the $N \times (M-N)$ matrix of non-pivot columns of $A$, and $P$ is a permutation matrix satisfying that $P^T=P^{-1}.$ We define $B_A \in \mathrm{Gr}(M-N, M)$, the \textit{dual} matrix to $A$, by
\begin{equation}
    B_A=(-1)^{\sigma} \det P\cdot
    \begin{pmatrix}
        -G^T|I_{M-N}
    \end{pmatrix}
    PD, 
\end{equation}
where $\sigma=M(M+1)/2+N(N+1)/2$ and $D=\text{diag}(-1,1,\dots,(-1)^M)$. 
\begin{proposition}[Proposition 7.1, \cite{kodama}; Lemma C.4, \cite{KP2}]
    Let $A\in \mathcal{P}_{\pi}^{>0} \subset \mathrm{Gr}(N,M)$ be an irreducible matrix. Then the matrix $B_A$ satisfies that: 
    \begin{enumerate}
        \item $B_A \in \mathcal{P}_{{\pi}^{-1}}^{>0} \subset \mathrm{Gr}(M-N, M) $.
        \item $B_A$ is irreducible.
        \item If $I \in \binom{[M]}{N}$ and $J=[M]\backslash I$, then $\Delta_I(A) = \Delta_J(B_A)$.
        \item If $\{ j_1, \dots, j_{M-N}\}$ is the non-pivot set of $A$, then the pivot set of $B_A$ is given by $\{\pi(j_1), \dots, \pi(j_{M-N})\}$. 
    \end{enumerate}
\end{proposition}
We mean to exploit the relationship between $A$ and $B_A$ to classify the asymptotic structure of solitons.
\begin{proposition}\label{prop-dual}
    Let $\tau^{(n)}_A(x_{-1},x_1)$ (resp. $\tau_A(x,y,t)$) be the $\tau$-function for the $2$D Toda lattice (resp. the Davey--Stewartson) equation  associated to a matrix $A \in \text{Gr}(N,M)_{\geq 0}$. In the limit $|x_{-1}|\gg 0$ (resp. $|t|\gg0$), $\tau^{(-n)}_A(-x_{-1},-x_1)$ ( resp. $\tau_A(-x,-y,-t)$) is associated to the matrix $B_A \in \text{Gr}(M-N, M)_{\geq 0}$.
\end{proposition}
\begin{proof}
It is fundamental to notice first that both the $2$DTL equation and the DS equations are, respectively, invariant under the transformations
\begin{align}
    (x_{-1},x_1,n)&\mapsto(-x_{-1},-x_1,-n),\\
    (x,y,t)&\mapsto(-x,-y,-t).
\end{align}
This implies that $\tau^{(-n)}_A(-x,-y)$ and $\tau_A(-x,-y,-t)$ are again solutions to the $2$DTL and DS systems of equations, respectively. We will only explicitly prove the proposition for the 2D Toda lattice, as the proof for the Davey--Stewartson system is analogous. Using $\Delta_I(A)=\Delta_J(B_A)$ from above, and denoting $\Theta_J=\sum_{j\in J} (x_1 \lambda_j +\frac{x_{-1}}{\lambda_j}+n\ln{\lambda_j}) $, we may write:
\begin{align}
    \tau_A(-x_1,-x_{-1},-n)&=e^{-\Theta_{[M]}} \sum_{J\in \binom{[M]}{M-N}} \Delta_J(B_A) \prod_{i>j\in I}(\lambda_i-\lambda_j) e^{\Theta_J}\\
    &=e^{-\Theta_{[M]}} \tau'(x_1,x_{-1},n),
\end{align}
where $I=[M] \backslash J$ and the implicitly defined $\tau'$ matches $\tau_{B_A}$ up to the Vandermonde term $V_I=\prod_{i>j\in I} (\lambda_i-\lambda_j)$. 
In order to correct this in the large variable limit, let us define the following constant:
\begin{align}
    e^{\Theta_J^0}= \prod_{i\in J} \prod_{i\neq j\in [M]} (\lambda_i-\lambda_j)= V_JV_I^{-1}V_{[M]}.
\end{align}
This lets us rewrite $\tau_{B_A}$ as
\begin{align}
    \tau_{B_A}(x_1,x_{-1},n)&=V_{[M]}^{-1}\sum_{J\in \binom{[M]}{M-N}} \Delta_J(B_A) \prod_{i>j\in I}(\lambda_i-\lambda_j) e^{\Theta_J+\Theta_J^0}\\
    &\underset{|x_1| \to \infty}{=} V_{[M]}^{-1}\tau'(x_1,x_{-1},n)
\end{align}
Leading in the large variable limit to:
\begin{align}
    \tau_A(-x_1,-x_{-1},-n)= V_{[M]} e^{-\Theta_{[M]}} \tau_{B_A}(x_1,x_{-1},n)
\end{align}
\end{proof}

\begin{remark}
Note that 
\begin{align}
    \max_{I\in\mathcal M(A)}\{\Theta_I(-x_{-1}, -x_1, -n)\}&=\max_{J\in\mathcal M(B_A)}\{\Theta_{[M]}(-x_{-1},-x_1, -n)+\Theta_J(x_{-1},x_1,n)\}=\max_{J\in\mathcal M(B_A)}\{\Theta_J(x_{-1},x_1,n)\},
\end{align}
so that the asymptotic contour plots of $V_{-n,A}(-x_{-1}, -x_1)$ and $V_{n,B_A}(x_{-1}, x_1)$ match exactly (cf. \cref{def_cont_toda}). An analogous statement holds for the asymptotic contour plots of of the Davey--Stewartson system. Heuristically, the overall term $V_{[M]} e^{-\Theta_{[M]}}$ in front of $\tau_{B_A}$ does not affect the soliton graph, as $V_{n,A}$ is a second derivative of $\ln(\tau)$ (cf. equation \eqref{eq-solution}). This means that the soliton graphs given by $\tau_A(-x_1,-x_{-1},-n)$ and $ \tau_{B_A}(x_1,x_{-1},n)$ agree in the large variable limit, and similarly for $\tau_A(-x, -y, -t)$ and $\tau_{B_A}(x, y,t)$. 
\end{remark}

\section{2D Toda lattice}\label{sec-2DTodalattice}

\subsection{Asymptotics of soliton solutions}\label{sec-Asymptotic}
In the following, we will always assume that the $N\times M$ matrix $A$ is irreducible. The motivation behind this assumption concerns the form of the solution
    \begin{equation}
        V_{n,A}(x_{-1},x_1)=\frac{\partial^2}{\partial x_{-1}\partial x_1}\ln(\tau^{(n)}_A(x_{-1},x_1)).
    \end{equation}
Indeed, suppose that all the entries of the $j$-th column of $A$ are zero. Then, we can get the same solution by considering the matrix $\tilde A$, obtained from $A$ by removing the $j$-th column. If we assume that $A$, once in RREF, has a row in which the only nonzero element is the pivot, then we get the same solution by considering a new $(N-1)\times (M-1)$ matrix $\tilde A$, obtained from $A$ removing the row and column containing the pivot.

In \cref{sec-contourplotTL}, we defined the asymptotic contour plot $\mathcal C^{(n)}(\mathcal M(A))$ as the corner locus of the function
\begin{equation}
    g^{(n)}_{\mathcal M(A)}(x_{-1},x_1)=\max_{I\in\mathcal M(A)}\{\Theta_I(x_{-1},x_1,n))\}.
\end{equation}
One could find that, for a specific choice of parameters $\lambda_1, \dots, \lambda_M$, $\Theta_I=\Theta_J$ for two different multi-indices $I, J \in \binom{[M]}{N}$. This would be unlucky, as we are looking to unambiguously label regions of the plane by the multi-index extremizing $g^{(n)}_{\mathcal{M}(A)}$. Looking at equation \eqref{eq:tau-toda}, one finds that this is hardly ever the case, but in order to avoid it,  we will always assume that the parameters $\{ \lambda_1, \dots, \lambda_M  \}$ are \textit{generic}.
\begin{definition}\label{def-generic}
    A set of positive real parameters $\{\lambda_1, \dots, \lambda_M\}$ is called \textit{generic} if at least \textbf{one} of the following conditions is satisfied:
    \begin{enumerate}
        \item For all $I,J\in\binom{[M]}{m}$ with $I\neq J$, we have that $\sum_{j\in J}\lambda_{j}\neq\sum_{i\in I}\lambda_i$, for all $m=1, \dots, M$;
        \item For all $I,J\in\binom{[M]}{m}$ with $I\neq J$, we have that $\sum_{j\in J}\lambda_{j}^{-1}\neq\sum_{i\in I}\lambda_i^{-1}$, for all $m=1, \dots, M$;
        \item For all $I,J\in\binom{[M]}{m}$ with $I\neq J$, we have that $\sum_{j\in J} \ln \lambda_{j}\neq\sum_{i\in I} \ln \lambda_i$ for all $m=1, \dots, M$.\footnote{Equivalently, we can ask that
 $\prod_{j\in J}\lambda_j\neq\prod_{i\in I}\lambda_i$.}
    \end{enumerate}
\end{definition}
\begin{remark}
 Setting $m=1$, we find from the above condition that all $\lambda_i$ need to be distinct. Therefore, we will always assume them to be ordered:  $\lambda_1<\dots<\lambda_M$.
\end{remark}

The genericity of the parameters ensures that the plot $\mathcal C^{(n)}(\mathcal M(A))$ is always \textit{well-behaved}, in a sense we make precise presently. For simplicity, sometimes we will write $\phi_i$ instead of $\ln\lambda_i$.
\begin{proposition}\label{prop-dominantphases}
    Let $\{\lambda_1<\dots<\lambda_M\}$ be generic. Then the phases that dominate two adjacent regions of the soliton graph $\mathcal{C}^{(n)}(\mathcal{M}(A))$ are of the form $\Theta_I$ and $\Theta_J$, where $I,J\in\mathcal{M}(A)$ differ by only one element. 
\end{proposition}
\begin{proof}
    Firstly, notice that, given some $K\in \binom{[M]}{N}$, 
    \begin{equation}\label{Theta-Toda}
        \Theta_K(x_{-1}, x_1, n)= x_1\sum_{k\in K} \lambda_k+x_{-1}\sum_{k\in K} \frac{1}{\lambda_k}+n\sum_{k\in K} \ln \lambda_k.
    \end{equation}
    Therefore, our genericity assumption implies that the functions $\Theta_{K_1}$ and $\Theta_{K_2}$ are never equal for $K_1 \neq K_2$.

    We will prove the desired result by contradiction. Suppose that the multi-indices of two adjacent regions differ by at least two indices. Set $I=\{i,j, m_3, \dots m_N\}$ and $J=\{k,l, m_3, \dots, m_N\}$, and suppose these subsets label dominant phases in the adjacent regions $R$ and $R'$, respectively.
    It is clear that the dominance relation between $\Theta_I(x_{-1},x_1,n)$ and $\Theta_J(x_{-1},x_1,n)$ remains fixed for finite values of $n$. We will then ignore constant terms, so that the common boundary of $R$ and $R'$ is given by the line
    \begin{equation}
        L:\hspace{1 em} \theta_i +\theta_j=\theta_k+\theta_l.
    \end{equation}
    To solve the equations defining $L$, suppose that 
    \begin{equation}
        \theta_i-\theta_k=\theta_j - \theta_l=0.
    \end{equation}
    If condition 3 of \cref{def-generic} held, the slopes of the lines $\theta_i=\theta_k$ and $\theta_j=\theta_i$ would be different, and we would immediately have a contradiction. Therefore, we suppose that either condition 1 or condition 2 must hold\footnote{Or possibly both.}. Since the lines $\theta_i=\theta_k$ and $\theta_j=\theta_l$ coincide, their slopes must be equal (i.e., $\phi_i+\phi_k=\phi_j+\phi_l$) \textbf{and} there must exist $n\in\mathbb Z$ such that
    \begin{equation}\label{eq-translation}
        \frac{\phi_k-\phi_i}{\lambda_i-\lambda_k}n=\frac{\phi_l-\phi_j}{\lambda_j-\lambda_l}n.
    \end{equation}
    Assume by contradiction that there exists $n\neq0$ such that \cref{eq-translation} is satisfied. Notice that, because of our genericity assumption $\phi_i+\phi_k=\phi_j+\phi_l$ if and only if the indices are ordered in a way such that either $(\phi_j,\phi_l)\subset(\phi_i,\phi_k)$ or $(\phi_i,\phi_k)\subset(\phi_j,\phi_l)$ as intervals in the real line.\\
    Let $m\in\mathbb R$ be such that $\phi_i+\phi_k=2m=\phi_j+\phi_l$ and assume that we have the following order: $i<j<l<k$. Then, there exist $t,s\geq0$ such that
    \begin{align}
        &\phi_i=m-t,\qquad\phi_k=m+t,\\
        &\phi_j=m-s,\qquad\phi_l=m+s.
    \end{align}
    We can rewrite \cref{eq-translation} as
    \begin{equation}
        \frac{t}{\sinh(t)}=\frac{s}{\sinh(s)}.
    \end{equation}
    If we prove that the function $h(t)=\frac{t}{\sinh(t)}$ is strictly decreasing for $t\geq0$ we are done. Indeed, we have
    \begin{equation}
        h'(t)=\frac{\sinh(t)-t\cosh(t)}{\sinh^2(t)}
    \end{equation}
    and $h'(t)<0$ if and only if $\sinh(t)-t\cosh(t)<0$, which is equivalent to $\tanh(t)-t<0$. Let us consider $\tilde h(t)=\tanh(t)-t$ and notice $\tilde h(0)=0$ and $\tilde h'(t)=\text{sech}^2(t)-1\leq0$ for all $t\in\mathbb R$. In particular, $\tilde h(t)$ is a strictly decreasing function and $0=\tilde h(0)>\tilde h(t)$ for all $t>0$. Hence $h(t)$ is strictly decreasing for $t\geq0$ and we have a contradiction.

    It must then be that $\theta_i(x_{-1},x_1,n)-\theta_k(x_{-1},x_1,n)$ and $\theta_j(x_{-1},x_1,n) - \theta_l(x_{-1},x_1,n)$ have opposing signs along $L$. Suppose $\Delta_{\{k,j,m_3,\dots, m_N\}}(A)>0$ and $\Delta_{\{i,l,m_3,\dots, m_N\}}(A)>0$. In this case, both phases $\Theta_{\{k,j,m_3,\dots, m_N\}}$ and $\Theta_{\{k,j,m_3,\dots, m_N\}}$ appear in the $\tau$-function, meaning the inequalities
    \begin{gather}
        \theta_i -\theta_k = \Theta_I -\Theta_{\{k,j,m_3,\dots, m_N\}}>0,\\
        \theta_j-\theta_l= \Theta_I - \Theta_{\{i,l,m_3,\dots, m_N\}}>0
    \end{gather}
    hold in the interior of the region $R$. In this case, continuity would imply that $\theta_i -\theta_k$ and $\theta_j -\theta_l$ cannot have opposite signs along the boundary $L$ of $R$. This means that either $\Delta_{\{k,j,m_3,\dots, m_N\}}(A)$ or $\Delta_{\{i,l,m_3,\dots, m_N\}}(A)$ are null. Now we arrive at our final contradiction; the Pl\"ucker relations \ref{plucker relation} imply
    \begin{equation}
        \Delta_I(A)\Delta_J(A)=\Delta_{\{k,j,m_3,\dots, m_N\}}(A)\Delta_{\{i,l,m_3,\dots, m_N\}}(A),
    \end{equation}
    but the left-hand side cannot be zero, as $I$ and $J$ both label dominant phases.
\end{proof}

\begin{remark}\label{parallel-lines}
    An interesting feature of the contour plots arising in the proof of \cref{prop-dominantphases} is the appearance of parallel line solitons. This phenomenon is particular to the $2$D Toda lattice (and, as we will see later, also to the Davey--Stewartson equation). For the KP equation, the same phenomenon cannot occur as generic parameters immediately imply different slopes for all possible line solitons.
\end{remark}

From \cref{prop-dominantphases}, it follows that if two dominant phases differ by indices $i$ and $j$, the soliton line $L_{i,j}$ dividing them in the asymptotic contour plot is given by $\theta_{i}(x_{-1}, x_1,n)= \theta_{j}(x_{-1}, x_1,n)$. That is, $L_{i,j}$ is given by
\begin{equation}\label{eq:toda-line}
    \begin{split}
        L_{i,j}: & \hspace{1em} x_1=\frac{1}{\lambda_i \lambda_j}x_{-1}+\frac{\ln \lambda_j-\ln \lambda_i}{\lambda_i-\lambda_j}n.
    \end{split}
\end{equation}

\begin{remark}\label{label-contour}
    Notice that we are able to label the regions of $\mathcal{C}^{(n)}(\mathcal{M}(A))$ by their corresponding dominant exponential, analogously to how we labeled the regions of a reduced pipedream in \cref{fig:region-labels}. More precisely, looking at \cref{Theta-Toda}, we can see that in the region $x_1 \ll 0$ the dominant exponential is labeled by the lexicographical minimum in $\mathcal{M}(A)$\footnote{This coincides with the pivot index of $A$.}. This information is enough to fully determine all dominant labels on $\mathcal{C}^{(n)}(\mathcal{M}(A))$, as \cref{prop-dominantphases} shows that soliton lines $L_{i,j}$ act by transposing the indices $i$ and $j$ in the labels of adjacent regions.
\end{remark}

\begin{remark}
    Notice how whenever $A\in \mathrm{Gr}(N,M)_{>0}$, the proof of Proposition \ref{prop-dominantphases} follows from simple genericity arguments, as no minors of $A$ are ever null.
\end{remark}

\begin{lemma}\label{lem-dominantrel}
    Let $\{\lambda_1<\dots<\lambda_M\}$ be generic and $L_{i,j}$ be the $[i,j]$-soliton line for $i<j$.  Then,  along $L_{i,j}$, we have the following dominant relations
for $x_{-1}\gg 0$,
    \begin{itemize}
        \item $\theta_i=\theta_j<\theta_m$ for $1\leq m<i$ or $j<m\leq M$;
        \item $\theta_i=\theta_j>\theta_m$ for $i<m<j$.
    \end{itemize}
\end{lemma}
\begin{proof}
    Let $k \in \{1, \dots, M\}$. We may write the difference $\theta_k - \theta_i$ as
    \begin{equation}\label{theta_dif}
        \theta_k-\theta_i=(\lambda_k - \lambda_i)x_1+ \left(\frac{1}{\lambda_k}-\frac{1}{\lambda_i} \right)x_{-1}+(\ln \lambda_k - \ln \lambda_i) n.
    \end{equation}
    Recall that the equality $\theta_i = \theta_j$ defines a line $L_{i,j}$ in the $(x_1, x_{-1})$-plane, specified by the equality \eqref{eq:toda-line}. If we use the preceding equation to substitute for $x_1$ in (\ref{theta_dif}), then we have that, along the line $L_{i,j}$,
    \begin{equation}
        \theta_k-\theta_i=\left( \frac{\lambda_k^2-\lambda_k(\lambda_i+\lambda_j)+\lambda_i\lambda_j}{\lambda_k\lambda_i\lambda_j} \right)x_{-1} + \left( \ln \lambda_k - \ln \lambda_i + \frac{\lambda_k-\lambda_i}{\lambda_i-\lambda_j}(\ln \lambda_j - \ln \lambda_i)  \right)n.
    \end{equation}
    For a fixed $n$, and in the limit $|x_{-1}| \gg 0$, the sign of the difference $\theta_k-\theta_i$ along $L_{i,j}$ is fully determined by the coefficient of $x_{-1}$ in the preceding equation. Therefore, to prove the desired result, it is enough to study the behavior of the quadratic function $\alpha: \mathbb{R} \rightarrow \mathbb{R}$, $\alpha(\lambda)=\lambda^2 -\lambda(\lambda_i+\lambda_j)+\lambda_i\lambda_j$. We conclude the proof by noting that the two zeroes of $\alpha$ are $\lambda_i$ and $\lambda_j$, and that $\lim_{\lambda \to \infty} \alpha(\lambda)=\infty.$
\end{proof}

\begin{theorem}\label{asymptotic}
        Let $A \in \mathcal{P}_{v,w}^{>0} \subset \mathrm{Gr}\,(N,M)_{\geq0}$ be an irreducible matrix, with pivot set $\{i_1, \dots i_N\}$ and non-pivot set $\{j_1, \dots j_{M-N}\}$. Assume that the parameters $\{\lambda_1< \dots < \lambda_M\}$ of the Toda lattice soliton $V_{n,A}(x_{-1},x_1)$ are generic. Then $\mathcal{C}^{(n)}(\mathcal{M}(A))$ has the following asymptotic structure: 
        \begin{enumerate}
            \item For $x_{-1}\gg 0$, there exist $N$ line solitons of type $[i_n, p_n]$, where $p_n>i_n$.
            \item For $x_{-1}\ll 0$, there exist $(M-N)$ line solitons of type $[q_m, j_m]$, where $q_m < j_m$.
        \end{enumerate}
\end{theorem}
The proof of \cref{asymptotic} is completely analogous to the proof of the corresponding statement for KP solitons in \cite[Theorem 6.1]{kodama}. That is, the underlying combinatorial arguments are the same, and the analytics of the $2$D Toda lattice equation feature almost exclusively through Lemma \ref{lem-dominantrel}. A proof of \cref{asymptotic} is included in Appendix \ref{app-asymptotic}, for completeness.

\begin{proposition}\label{fan_is_der}
    Let $A \in \mathcal{P}_{v,w}^{>0}$ be an irreducible matrix. In the notation of Proposition \ref{asymptotic}, the permutation $\pi \in S_M$ defined by the asymptotic contour plot of $\tau_A$,
    \begin{equation}
    \begin{cases}
        \pi(i_k)=p_k & k=1, \dots, N\\
        \pi(j_k)=q_k & k=1, \dots, M-N,
    \end{cases}
    \end{equation}
    is given by $\pi = vw^{-1}$.
\end{proposition}

\begin{remark}
    Since each irreducible Deodhar component corresponds to a derangement, we deduce that the permutation $\pi$ has no fixed points.  
\end{remark}

As before, the proof of Proposition \ref{fan_is_der} is functionally the same as the proof appearing in \cite[Theorem 6.1]{kodama}. We refer the reader to Appendix \ref{app-fan} for details.

\subsection{Algorithm}\label{sec:algorithm}

The goal of this section is to provide an algorithm, based on the one presented in \cite[Chapter 8]{kodama} and \cite[Algorithm 8.3]{kodama2014kp}, whose input consists of a matrix $A\in\mathrm{Gr}(N,M)_{\geq0}$ and a parameter $n\gg0$, and which produces the contour plot of the solution $V_{n,A}(x_{-1},x_1)$. We conclude the section with an example where the algorithm is applied.

For the rest of the section we denote by $A$ a matrix in the totally nonnegative part of the Grassmannian $\mathrm{Gr}(N,M)_{\geq0}$ and we fix a set of generic parameters $\{\lambda_1,\dots,\lambda_M\}$. 

\begin{remark}
    Notice that, for all large enough $n$, the topology of the contour plot $\mathcal C^{(n)}(\mathcal M(A))$ does not change. In other words, for $n\gg0$, the bounded regions appearing in $\mathcal{C}^{(n)}(\mathcal{M}(A))$ do not change as we increase $n$. Hence, for all $n\gg0$, the resulting contour plots are equivalent, or \textbf{self-similar}. In particular, for $n\gg0$ the contour plot is denoted by $\mathcal C^+(\mathcal M(A))$. The same reasoning applies when $n\ll0$ and in this case the contour plot is denoted by $\mathcal C^-(\mathcal M(A))$.
\end{remark}

Recall from \cref{prop-dominantphases} that if $I,J\in\mathcal M(A)$ label two adjacent dominant exponentials, then there are two indices $a,b\in[M]$ such that $J=(I\setminus\{a\})\cup\{b\}$. The line separating those regions is given by the $[a,b]$-soliton and has the following form

\begin{equation}\label{eq-line-Toda}
    L_{[a,b]}^{(n)}: x_1-\frac{1}{\lambda_a\lambda_b}x_{-1}+n\frac{\phi_a-\phi_b}{\lambda_a-\lambda_b}=0.
\end{equation}
\begin{remark}
    Note that we are ignoring constant terms in the equation of the $[a,b]$-soliton since we consider $n\gg0$ or $n\ll0$.
\end{remark}

When the only interaction points appearing in $\mathcal{C}^+(\mathcal{M}(A))$ are either trivalent vertices or $X$-crossings, we will refer to $\mathcal C^+(\mathcal M(A))$ as a \textbf{generic contour plot}. Trivalent vertices are also called \textbf{resonant interaction points} and they are denoted by $v_{[a,b,c]}$ where $a,b,c\in[M]$ and $a<b<c$. This notation comes from the fact that the vertex is obtained by intersecting three line solitons that are labeled by $[a,b]$, $[b,c]$ and $[a,c]$. The coordinates of a resonant interaction point are denoted by $(x_{1,[a,b,c]},x_{-1,[a,b,c]})$ and direct computations show that
\begin{align}
    &x_{-1,[a,b,c]}=-\frac{\lambda_a\lambda_b\lambda_c(\lambda_a(\phi_b-\phi_c)+\lambda_b(\phi_c-\phi_a)+\lambda_c(\phi_a-\phi_b))}{(\lambda_b-\lambda_a)(\lambda_c-\lambda_a)(\lambda_c-\lambda_b)}n,\\
    &x_{1,[a,b,c]}=-\frac{\lambda_a\lambda_b(\phi_a-\phi_b)+\lambda_a\lambda_c(\phi_c-\phi_a)+\lambda_b\lambda_c(\phi_b-\phi_c)}{(\lambda_a-\lambda_b)(\lambda_a-\lambda_c)(\lambda_b-\lambda_c)}n.
\end{align}
However, not all triplets of indices $a<b<c$ produce trivalent vertices in $\mathcal C^+(\mathcal M(A))$. Our next goal is to understand what conditions $a<b<c$ must satisfy so that the vertex $v_{[a,b,c]}$ appears in the contour plot. To this end, we fix indices $a<b<c$ and consider planes in the coordinates $(x_{1},x_{-1},z)$ given by the equations
\begin{equation}
    z=\theta_i(x_{-1},x_1,n),\quad i=1,\dots,M.
\end{equation}
We call $z_d=\theta_d(x_{-1,[a,b,c]},x_{1,[a,b,c]},n)$ the height of the vertex $v_{[a,b,c]}$ in the hyperplane $z=\theta_d(x_{-1},x_1,n)$.
\begin{remark}
    Since $\theta_a(x_{-1,[a,b,c]},x_{1,[a,b,c]},n)=\theta_b(x_{-1,[a,b,c]},x_{1,[a,b,c]},n)=\theta_c(x_{-1,[a,b,c]},x_{1,[a,b,c]},n)$, we have

    \begin{equation}
        z_{[a,b,c]}=z_a=z_b=z_c=\frac{\phi_c\sinh(\phi_a-\phi_b)-\phi_b\sinh(\phi_a-\phi_c)+\phi_a\sinh(\phi_b-\phi_c)}{\sinh(\phi_a-\phi_b)-\sinh(\phi_a-\phi_c)+\sinh(\phi_b-\phi_c)}n.
    \end{equation}
\end{remark}
\begin{definition}\label{def-visible}
    Let $a,b,c\in[M]$ such that $a<b<c$ and consider $v_{[a,b,c]}$. We say that $v_{[a,b,c]}$ is \textbf{visible} if $z_d<z_{[a,b,c]}$ for all $d\in[M]\setminus\{a,b,c\}$.
\end{definition}
In other words, a vertex $v_{[a,b,c]}$ is visible if the height $z_d$ at the vertex is smaller than the height $z_{[a,b,c]}$ for all $d\in[M]\setminus\{a,b,c\}$. Because of \cref{def-visible}, it becomes fundamental to study the sign of the function $F_{a,b,c}(\phi_d)=z_d-z_{[a,b,c]}$. After substituting in the expressions of $z_d$ and $z_{[a,b,c]}$, we can write

\begin{equation}
    F_{a,b,c}(\phi_d)=-\frac{f_{a,b,c}(\phi_d)}{\lambda_d(\lambda_a-\lambda_b)(\lambda_a-\lambda_c)(\lambda_b-\lambda_c)}n
\end{equation}
where
\begin{align}\label{eq-functionf}
    f_{a,b,c}(\phi_d)&=\phi_c\lambda_c(\lambda_a-\lambda_b)(\lambda_a-\lambda_d)(\lambda_b-\lambda_d)+\phi_b\lambda_b(\lambda_c-\lambda_a)(\lambda_a-\lambda_d)(\lambda_c-\lambda_d)\\
    &\quad-\phi_d\lambda_d(\lambda_b-\lambda_c)(\lambda_a-\lambda_b)(\lambda_a-\lambda_c)+\phi_a\lambda_a(\lambda_b-\lambda_c)(\lambda_b-\lambda_d)(\lambda_c-\lambda_d).
\end{align}
Recall that $\lambda_i=e^{\phi_i}$. Additionally $F_{a,b,c}(\phi_a)=F_{a,b,c}(\phi_b)=F_{a,b,c}(\phi_c)=0$ and we have the following:
\begin{proposition}\label{prop-signF}
    Let $a,b,c\in[M]$ such that $a<b<c$ and $n>0$, then

    \begin{enumerate}
        \item $F_{a,b,c}(\phi_d)>0$ for $d<a$ and $b<d<c$;
        \item $F_{a,b,c}(\phi_d)<0$ for $a<d<b$ and $c<d$.
    \end{enumerate}
\end{proposition}

\begin{proof}
    Note that $\displaystyle -\frac{1}{\lambda_d(\lambda_a-\lambda_b)(\lambda_a-\lambda_c)(\lambda_b-\lambda_c)}n>0$ for each $d\in[M]$. Therefore, it is enough to study the sign of $f_{a,b,c}(\phi_d)$. First of all, observe that

    \begin{equation}
        \frac{\partial}{\partial\phi_d}=\frac{\partial}{\partial e^{\phi_d}}\frac{\partial e^{\phi_d}}{\partial d}=e^{\phi_d}\frac{\partial}{\partial e^{\phi_d}}=\lambda_d\frac{\partial}{\partial \lambda_d}.
    \end{equation}
    Hence, the derivative with respect to $\phi_d$ differs from the derivative with respect to $\lambda_d$ by a factor $\lambda_d$, which is an increasing function in the variable $\phi_d$. The second derivative with respect to $\lambda_d$ of the function $f_{a,b,c}(\phi_d)$ is given by
    \begin{equation}
        \frac{\partial^2 f_{a,b,c}(\phi_d)}{(\partial \lambda_d)^2}=2\phi_c\lambda_c(\lambda_a-\lambda_b)-2\phi_b\lambda_b(\lambda_a-\lambda_c)-\frac{1}{\lambda_d}(\lambda_b-\lambda_c)(\lambda_a-\lambda_b)(\lambda_a-\lambda_c)+2\phi_a\lambda_a(\lambda_b-\lambda_c).
    \end{equation}
    Imposing the second derivative to be zero, we have
    \begin{equation}
        \frac{1}{\lambda_d}=\frac{2\phi_c\lambda_c(\lambda_a-\lambda_b)-2\phi_b\lambda_b(\lambda_a-\lambda_c)+2\phi_a\lambda_a(\lambda_b-\lambda_c)}{(\lambda_b-\lambda_c)(\lambda_a-\lambda_b)(\lambda_a-\lambda_c)}.
    \end{equation}
    Thus, the second derivative has at most one zero, meaning that the first derivative has at most one extreme point and so at most two zeroes, namely the function $f_{a,b,c}(\phi_d)$ has at most two extreme points and therefore at most three zeroes. Since $f_{a,b,c}(\phi_a)=f_{a,b,c}(\phi_b)=f_{a,b,c}(\phi_c)=0$, the function $f_{a,b,c}(\phi_d)$ has exactly three zeroes and two extreme points. The extreme points must be a local minimum and a local maximum. One should be negative and the other positive, otherwise we have a contradiction with the fact that the function has exactly three zeroes.

    Furthermore, we get
    \begin{align}
        \lim_{\phi_d\to-\infty}f_{a,b,c}(\phi_d)&=\lambda_a\lambda_b\lambda_c(\lambda_a(\phi_c-\phi_b)+\lambda_b(\phi_a-\phi_c)+\lambda_c(\phi_b-\phi_a)).
    \end{align}
    We want to prove that this constant is positive for each $a<b<c$, so that it forces $f_{a,b,c}(\phi_d)$ to satisfy the inequalities in the statement. Fix $a<b$ and consider the function
    \begin{equation}
        g_{a,b}(\phi_c)=\lambda_a(\phi_c-\phi_b)+\lambda_b(\phi_a-\phi_c)+\lambda_c(\phi_b-\phi_a).
    \end{equation}
    We have that $g_{a,b}(\phi_a)=g_{a,b}(\phi_b)=0$. Moreover, the first derivative
    \begin{equation}
        g_{a,b}'(\phi_c)=\lambda_a-\lambda_b+\lambda_c(\phi_b-\phi_a)    
    \end{equation}
    is greater than $0$ as long as $\phi_c>\displaystyle\ln\biggl(\frac{\lambda_b-\lambda_a}{\phi_b-\phi_a}\biggr)$, meaning that $g_{a,b}(\phi_c)$ is increasing for $\phi_c>\displaystyle\ln\biggl(\frac{\lambda_b-\lambda_a}{\phi_b-\phi_a}\biggr)$. Additionally, $g_{a,b}$ is a convex function since its second derivative is greater than $0$. It remains to prove that $\displaystyle\ln\biggl(\frac{\lambda_b-\lambda_a}{\phi_b-\phi_a}\biggr)<\phi_b$: indeed, by the \textit{mean value theorem} there exists $\zeta\in(\phi_a,\phi_b)$ such that 
    \begin{equation}
        \lambda_b-\lambda_a=e^{\phi_b}-e^{\phi_a}=e^\zeta(\phi_b-\phi_a)\implies \ln\biggl(\frac{\lambda_b-\lambda_a}{\phi_b-\phi_a}\biggr)=\zeta<\phi_b.
    \end{equation}
    Hence, $g_{a,b}(\phi_c)>0$ for each $a,b,c\in[M]$ such that $a<b<c$. In other words,
    \begin{align}
        \lim_{\phi_d\to-\infty}f_{a,b,c}(\phi_d)&=\lambda_a\lambda_b\lambda_c(\lambda_a(\phi_c-\phi_b)+\lambda_b(\phi_a-\phi_c)+\lambda_c(\phi_b-\phi_a))>0\qquad\text{for all }a<b<c.\qedhere
    \end{align}
\end{proof}
\begin{remark}
    Since the functions $F_{a,b,c}$ have at most three zeroes, notice that, for generic parameters, intersection points in $\mathcal{C}^{(n)}(\mathcal{M}(A))$ will consist of either trivalent vertices or $X$-crossings.
\end{remark}

Before establishing which are the visible vertices, we attach a color to each resonant interaction point $v_{[a,b,c]}$ that appears in the contour plot. The color of a vertex could be white or black, and it is assigned following the rules outlined below:
\begin{itemize}
    \item $v_{[a,b,c]}$ is a \textbf{white vertex}, denoted by $v^{\circ}_{[a,b,c]}$, if $v_{[a,b,c]}$ is obtained by the intersection of three lines bounding regions in which the functions $\Theta_{B\cup\{a\}},\Theta_{B\cup\{b\}},\Theta_{B\cup\{c\}}$ are dominant, for some $B\subseteq[M]$ having $N-1$ elements;
    \item $v_{[a,b,c]}$ is a \textbf{black vertex}, denoted by $v^\bullet_{[a,b,c]}$, if $v_{[a,b,c]}$ is obtained by the intersection of three lines bounding regions in which the functions $\Theta_{T\setminus\{a\}},\Theta_{T\setminus\{b\}},\Theta_{T\setminus\{c\}}$ are dominant, for some $T\subseteq[M]$ having $N+1$ elements.
\end{itemize}

\begin{proposition}\label{vis_ver}
    Let $n>0$ and $A\in\mathcal P^{>0}_{v,w}$ be an irreducible matrix. Then the trivalent vertices in the reduced pipedream associated to $\mathcal{P}_{v,w}^{>0}$ are visible. In particular, if the vertex in the pipedream is white (resp. black), then it corresponds to a resonant interaction point marked by a white (resp. black) vertex.
\end{proposition}
\begin{proof}
     Let there be a white vertex on the box of the $\reflectbox L$-diagram labeled $ac$, where $a$ is a pivot index and $c$ is a non pivot index. Consider the set of indices that label the regions adjacent to the vertex\footnote{See \cref{fig:region-labels}. Essentially, we label the southeast region of the reduced pipedream by the set of pivot indices of $A$, and let lines $[i,j]$ act as transpositions.}: $I_a$, $I_b$, and $I_c$; with $a <b<c$. These index sets are of the form
    \begin{equation}
        I_a=L_0 \cup \{a\}, \hspace{1em} I_b=L_0 \cup \{b \}, \hspace{1em} I_c=L_0 \cup \{c\}, 
    \end{equation}
    for some subset $L_0$. These indices admit an intrinsic definition independent of the $\reflectbox L$-diagram. In particular, $I_a$ is the lexicographically maximal element of the indices of the form
    \begin{equation}
        \{i_1< \cdots <i_k<a=i_{k+1}<p_1<\cdots <p_l<i_{k+l+2}<\cdots <i_{N}\} \in \mathcal{M}(A),
    \end{equation}
    where the $i_n$'s are the pivot indices of the matrix $A$. We consider two separate cases:
    \begin{itemize}
        \item Suppose that there exists no $i$ such that $i \notin L_0$ and $a<i<c$. In this instance, we claim that $I_a$, $I_b$, and $I_c$ are index sets labeling dominant phases that coincide on the relevant vertex. Indeed, if  $i \notin L_0$ is such that $i<a$, then $\Delta_I(A)=0$ for any $I$ such that $I \cap \{i\} \neq\varnothing$, since the set of pivot indices of $A$ is the lexicographically minimal element of $\mathcal{M}(A)$. By assumption, there is no index $i \notin L_0$ such that $b<i<c$. Lastly, if either $a<i<b$ or $c<i$, $F_{a,b,c}(\phi_i)<0$ by Proposition \ref{prop-signF}.
        \item Suppose now that there exists an $i$ such that $i\notin L_0$ and $a<i<c$. If $\theta_i<\theta_p$ for all $p \in I_a$, then similarly to before, the indices $I_a$, $I_b$, and $I_c$ label the required dominant regions. Otherwise, suppose that there exists $q \in I_0$ such that $a<q<i$ and $\theta_q < \theta_i$. In this case,
        \begin{equation}
            (I_a -\{q\}) \cup \{i\} >I_a
        \end{equation}
        in the lexicographical order, so that 
        \begin{equation}
            \Delta_{(I_a -\{q\}) \cup \{i\}}(A)=0
        \end{equation}
        by the maximality of $I_a$. This means that if $\theta_q < \theta_i$ for $q \in L_0$, we are free to assume that $i<q<c$. We set $J_a=
        (I_a-\{q\}) \cup \{i\}$. If $\Delta_{J_a}(A)\neq 0$, then we can show that $J_a$ and the analogously defined $J_b$ and $J_c$ label the desired dominant regions. That is, the Pl\"ucker relations imply that also $\Delta_{J_b}(A) \neq 0$ and $\Delta_{J_c}(A) \neq 0$.\qedhere  
    \end{itemize}
\end{proof}

\begin{remark}
    The previous proof is entirely analogous to the one given in \cite[Proposition 8.1]{kodama}. However, it must be highlighted that we need the hypothesis $n>0$ to prove the visibility of the vertices. Here is the point in which the shape of the equation arises in the proof (for the KP equation, $t<0$ is needed for the analogous Proposition to hold).
\end{remark}
    
We are now ready to provide the algorithm. Let $n\gg0$ and $A\in\mathcal P^{>0}_{v,w}$. We adapt the algorithm proposed in \cite{kodama} and \cite[\S 8]{kodama2014kp} to construct the soliton graph $\mathcal C^{+}(\mathcal M(A))$ starting from the $\reflectbox L$-diagram associated to $A$:

\begin{algorithm}\label{algorithm-TL1}
    \begin{enumerate}
        \item[]
        \item Draw the \reflectbox{$\mathrm{L}$}-diagram from the pair $(v,w)$ or the derangement $\pi=vw^{-1}$.
        \item Draw the reduced pipedream.
        \item Plot in the $(x_{1},x_{-1})$-plane all the trivalent vertices $v_{[a,b,c]}=(x_{1,[a,b,c]},x_{-1,[a,b,c]})$, $a<b<c$, in the \reflectbox{$\mathrm{L}$}--diagram, where
        \begin{align}
            &x_{1,[a,b,c]}=-\frac{\lambda_a\lambda_b(\phi_a-\phi_b)+\lambda_a\lambda_c(\phi_c-\phi_a)+\lambda_b\lambda_c(\phi_b-\phi_c)}{(\lambda_a-\lambda_b)(\lambda_a-\lambda_c)(\lambda_b-\lambda_c)}n,\\
            &x_{-1,[a,b,c]}=-\frac{\lambda_a\lambda_b\lambda_c(\lambda_a(\phi_b-\phi_c)+\lambda_b(\phi_c-\phi_a)+\lambda_c(\phi_a-\phi_b))}{(\lambda_b-\lambda_a)(\lambda_c-\lambda_a)(\lambda_c-\lambda_b)}n.
        \end{align}
        We have two cases:
        
        \begin{itemize}
            \item if the vertex is white: draw the line $L_{[a,c]}$ for $x_{-1}>x_{-1,[a,b,c]}$ and the lines $L_{[a,b]}$ and $L_{[b,c]}$ for $x_{-1}<x_{-1,[a,b,c]}$;
            \item if the vertex is black: draw the line $L_{[a,c]}$ for $x_{-1}<x_{-1,[a,b,c]}$ and the lines $L_{[a,b]}$ and $L_{[b,c]}$ for $x_{-1}>x_{-1,[a,b,c]}$.
        \end{itemize}

        The equation of the line $[a,b]$ is given in \cref{eq-line-Toda}.
    
        \item Draw an infinite line $L_{[i,j]}$ every time there is a path labeled $[i,j]$ that does not connect two vertices in the pipedream.
    \end{enumerate}
\end{algorithm}
The proof that the previous algorithm faithfully reproduces the corresponding contour plot proceeds by induction along the number of rows of the matrix $A$. Informally, if we denote by $A'$ the matrix resulting from \textit{deleting} the top row of $A$, we prove that $\mathcal{C}^+(\mathcal{M}(A'))$ comprises a polyhedral subset that is almost the entirety of $\mathcal{C}^+(\mathcal{M}(A))$. The difference between the two is fully determined by Proposition \ref{asymptotic}, which describes the asymptotics of $\mathcal{C}^+(\mathcal{M}(A))$, and by Proposition \ref{vis_ver}. The proof of the validity of Algorithm \ref{algorithm-TL1} is functionally the same as the proofs appearing in \cite[Theorem 8.2]{kodama} and \cite[Theorem 8.5]{kodama2014kp}. Before we proceed, we need to discuss the notion of \emph{trips} along reduced pipedreams:

\begin{definition}
    Let $P_A$ be the reduced pipedream corresponding to a matrix $A \in \mathrm{Gr(N,M)_{ \geq 0}}$. Recall from \cref{sec:grassmannian} how each pipe in $P_A$ is labeled by $[i,j]$ for some $i,j \in N$. For any $k \in [N]$, we refer to the collection of pipes labeled by $k$ as the \emph{trip} $T_k$ (see Figure \ref{fig:trips}).
\end{definition}

Let $A\in \mathrm{Gr}(N,M)_{\geq 0}$ be an irreducible matrix, and let $L$ be its corresponding $\reflectbox{L}$-diagram. We denote by $L^{(k)}$ the $\reflectbox{L}$-diagram resulting from filling every box in the first $N-k$ rows of $L$ with white stones, so that $L^{(N)}=L$. Similarly, let $A^{(k)}$ be the matrix that results from setting every entry in the first $N-k$ rows of $A$ to zero, except for the pivots. We will show that $\mathcal{C}^+(\mathcal{M}(A^{(k)}))$ can be combinatorially recovered from $\mathcal{C}^+(\mathcal{M}(A^{(k+1)}))$\footnote{Note that, for $k<N$, the matrix $A^{(k)}$ will not be irreducible. The results detailed in the body of the article still hold; the caveat is that $\mathcal{C}^+(\mathcal{M}(A^{(k)}))$ may be reproduced by a matrix in a \textit{smaller} Grassmannian.}. More precisely, we show that $\mathcal{C}^+(\mathcal{M}(A^{(k)}))$ corresponds to a polyhedral subset of $\mathcal{C}^+(\mathcal{M}(A^{(k+1)}))$, the boundary of which corresponds to the trip $T_{N-k}$ on the reduced pipedream of $L^{(k+1)}$. The asymptotics of $\mathcal{C}^+(\mathcal{M}(A^{(k+1)}))$ are fully determined by Proposition \ref{asymptotic}. Thanks to this, if we assume that $\mathcal{C}^+(\mathcal{M}(A^{(k)}))$ can indeed be algorithmically derived from $L^{(k)}$, the same will have to be true of $\mathcal{C}^+(\mathcal{M}(A^{(k+1)}))$ and $L^{(k+1)}$. For clarity, we focus on the relationship between $\mathcal{C}^+(\mathcal{M}(A^{(N)}))$ and $\mathcal{C}^+(\mathcal{M}(A^{(N-1)}))$. Since $A$ is irreducible, $i_1=1.$

\begin{lemma}
    There exists a polyhedral subset $R$ of $\mathcal{C}^+(\mathcal{M}(A))$ such that every region in $R$ is labeled by an index $J$ satisfying $1 \in J$.
\end{lemma}
\begin{proof}
    Recall that $\theta_i(x_{-1},x_1,n)=\frac{1}{\lambda_i}x_{-1}+\lambda_ix_1+ (\ln \lambda)n$. It follows that, for any fixed $x_{-1}$ and $n$, there exists a large enough $x_1$ such that
    \begin{equation}
        \theta_1(x_{-1}, -x_1,n)> \theta_2(x_{-1}, -x_1,n)+> \dots >\theta_N(x_{-1}, -x_1,n).
    \end{equation}
    Therefore, any index $J$ labeling a dominant region where the $x_1$ coordinate is sufficiently negative must satisfy that $1 \in J$. The subset $R$ defined as the union of such regions is polyhedral because its boundary is necessarily comprised of soliton lines of the type $[1,j]$ for some $j$. 
\end{proof}

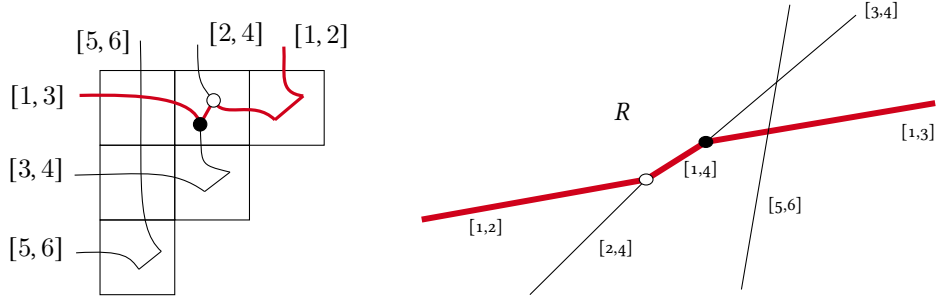
\begin{figure}[h]
    \centering
     \begin{tikzpicture}[x=0.75pt,y=0.75pt,yscale=-0.75,xscale=0.75]
            
            \draw   (400,526) -- (450,526) -- (450,576) -- (400,576) -- cycle ;
            \draw   (450,526) -- (500,526) -- (500,576) -- (450,576) -- cycle ;
            \draw   (450,476) -- (500,476) -- (500,526) -- (450,526) -- cycle ;
            \draw   (400,476) -- (450,476) -- (450,526) -- (400,526) -- cycle ;
            \draw   (400,576) -- (450,576) -- (450,626) -- (400,626) -- cycle ;
            \draw   (500,476) -- (550,476) -- (550,526) -- (500,526) -- cycle ;
            \draw    (427,456) .. controls (429,540) and (422,580) .. (441,597) ;
            \draw    (384,598) .. controls (417,596) and (420,601) .. (426,609) ;
            \draw    (426,609) -- (441,597) ;
            \draw    (383,546) .. controls (432,545) and (449,542) .. (470,557) ;
            \draw    (470,557) -- (487,544) ;
            \draw [color={rgb, 255:red, 208; green, 2; blue, 27 }, very thick  ,draw opacity=1 ]   (523,460) .. controls (524,478) and (519,491) .. (536,493) ;
            \draw  [color={rgb, 255:red, 208; green, 2; blue, 27 }, very thick  ,draw opacity=1 ]  (517,509) -- (536,493) ;
            
            \draw  [color={rgb, 255:red, 208; green, 2; blue, 27 }, very thick  ]  (467,512) -- (476,496) ;
            \draw  [color={rgb, 255:red, 208; green, 2; blue, 27 }, very thick, draw opacity=1]  (386.33,492.67) .. controls (430.33,491.67) and (463.33,494.67) .. (467,512) ;
            \draw    (467,512) .. controls (468.33,539.67) and (465.33,538.67) .. (487,544) ;
            \draw [color={rgb, 255:red, 208; green, 2; blue, 27 }, very thick]   (476,496) .. controls (482.33,509.67) and (502.33,494.67) .. (517,509) ;
            \draw    (465.33,458.67) .. controls (465.33,482.67) and (466.33,484.67) .. (476,496) ;
                
            \draw (527,443) node [anchor=north west][inner sep=0.75pt]    {$[ 1,2]$};
            \draw (470,442) node [anchor=north west][inner sep=0.75pt]    {$[ 2,4]$};
            \draw (381,447) node [anchor=north west][inner sep=0.75pt]    {$[ 5,6]$};
            \draw (337,484) node [anchor=north west][inner sep=0.75pt]    {$[ 1,3]$};
            \draw (336,533) node [anchor=north west][inner sep=0.75pt]    {$[ 3,4]$};
            \draw (336,586) node [anchor=north west][inner sep=0.75pt]    {$[ 5,6]$};
            \draw  [fill={rgb, 255:red, 255; green, 255; blue, 255 }  ,fill opacity=1 ] (471.5,496) .. controls (471.5,493.51) and (473.51,491.5) .. (476,491.5) .. controls (478.49,491.5) and (480.5,493.51) .. (480.5,496) .. controls (480.5,498.49) and (478.49,500.5) .. (476,500.5) .. controls (473.51,500.5) and (471.5,498.49) .. (471.5,496) -- cycle ;
            \draw  [fill={rgb, 255:red, 0; green, 0; blue, 0 }  ,fill opacity=1 ] (462.5,512) .. controls (462.5,509.51) and (464.51,507.5) .. (467,507.5) .. controls (469.49,507.5) and (471.5,509.51) .. (471.5,512) .. controls (471.5,514.49) and (469.49,516.5) .. (467,516.5) .. controls (464.51,516.5) and (462.5,514.49) .. (462.5,512) -- cycle ;
            
        \end{tikzpicture} \hspace{2em} 
    \begin{tikzpicture}[x=0.75pt,y=0.75pt,yscale=-1,xscale=1]

\draw    (252.07,188.79) -- (193.73,246.6) ;
\draw [color={rgb, 255:red, 208; green, 2; blue, 27 }  ,draw opacity=1 ][line width=2.25]    (252.07,188.79) -- (139.44,208.92) ;
\draw [color={rgb, 255:red, 208; green, 2; blue, 27 }  ,draw opacity=1 ][line width=2.25]    (252.07,188.79) -- (282.05,169.95) ;
\draw    (282.05,169.95) -- (357.41,106.29) ;
\draw [color={rgb, 255:red, 208; green, 2; blue, 27 }  ,draw opacity=1 ][line width=2.25]    (282.05,169.95) -- (397.12,150.46) ;
\draw    (324.19,101.09) -- (299.88,244) ;
\draw  [fill={rgb, 255:red, 255; green, 255; blue, 255 }  ,fill opacity=1 ] (248.63,188.79) .. controls (248.63,187.26) and (250.17,186.03) .. (252.07,186.03) .. controls (253.97,186.03) and (255.52,187.26) .. (255.52,188.79) .. controls (255.52,190.31) and (253.97,191.55) .. (252.07,191.55) .. controls (250.17,191.55) and (248.63,190.31) .. (248.63,188.79) -- cycle ;
\draw  [fill={rgb, 255:red, 0; green, 0; blue, 0 }  ,fill opacity=1 ] (278.61,169.95) .. controls (278.61,168.42) and (280.15,167.19) .. (282.05,167.19) .. controls (283.96,167.19) and (285.5,168.42) .. (285.5,169.95) .. controls (285.5,171.47) and (283.96,172.71) .. (282.05,172.71) .. controls (280.15,172.71) and (278.61,171.47) .. (278.61,169.95) -- cycle ;

\draw (378.64,159.55) node [anchor=north west][inner sep=0.75pt]  [font=\scriptsize] [align=left] {[1,3]};
\draw (269.38,178.24) node [anchor=north west][inner sep=0.75pt]  [font=\scriptsize] [align=left] {[1,4]};
\draw (161.66,207.23) node [anchor=north west][inner sep=0.75pt]  [font=\scriptsize] [align=left] {[1,2]};
\draw (226.83,217.87) node [anchor=north west][inner sep=0.75pt]  [font=\scriptsize] [align=left] {[2,4]};
\draw (311.93,196.26) node [anchor=north west][inner sep=0.75pt]  [font=\scriptsize] [align=left] {[5,6]};
\draw (360.16,98.36) node [anchor=north west][inner sep=0.75pt]  [font=\scriptsize] [align=left] {[3,4]};
\draw (235,149) node [anchor=north west][inner sep=0.75pt]   [align=left] {$\displaystyle R$};

\end{tikzpicture}

    \caption{Highlighted in red, the trip $T_1$ on the reduced pipedream of \cref{example-algorithm-Toda}, and its corresponding soliton lines on the associated contour plot.}
    \label{fig:trips}
\end{figure}

Notice that $\mathcal{C}^+(\mathcal{M}(A))$ and $\mathcal{C}^+(\mathcal{M}(A^{(N-1)}))$ coincide in $R$. Indeed, by construction, every dominant region $J$ in $\mathcal{C}^+(\mathcal{M}(A^{(N-1)}))$ satisfies that $1\in J$. More specifically,
\begin{equation}
    \mathcal{M}(A^{(N-1)})=\mathcal{M}(A) \cap \left( \bigcup_{\{p_1, \dots, p_{N-1}\}} \{1,p_1, \dots,p_{N-1}\}
    \right).
\end{equation}

\begin{proof}[Proof (of Algorithm \ref{algorithm-TL1}).] Instead of proving the inductive step in general, we will show that assuming Algorithm \ref{algorithm-TL1} reproduces $\mathcal{C}^+(\mathcal{M}(A^{(N-1)}))$ from $L^{(N-1)}$ means the same must be true for $\mathcal{C}^+(\mathcal{M}(A))$ and $L$.

Per Proposition \ref{asymptotic}, notice how the asymptotic structure of the unbounded solitons lines in $\mathcal{C}^+(\mathcal{M}(A^{(N)}))$ agrees with that described by the reduced pipedream of $L$. Recall how the boundary of $R$ in $\mathcal{C}^+(\mathcal{M}(A))$ is comprised of solitons of the type $[1,j]$. By Proposition \ref{vis_ver}, every vertex along the trip $T_1$ of $L$ must be visible; these vertices must be the exact same vertices appearing along the boundary of $R$. Indeed, if $\partial R$ had any extra vertices, this would either contradict our induction hypothesis or disturb the asymptotics of $\mathcal{C}^+(\mathcal{M}(A))$. To prove Algorithm \ref{algorithm-TL1} produces $\mathcal{C}^+(\mathcal{M}(A))$ from $L$, all that remains to show is that the vertices $v_{1,b,c}$ along $\partial R$ and in $T_1$ appear in the same \textit{order}. This follows from the proof of Theorem 8.5 in \cite{kodama2014kp}.
\end{proof}

\begin{example}\label{example-algorithm-Toda}
    Consider the matrix

    \begin{equation}
        A=
        \begin{pmatrix}
            1 &1 &0 &-1 &0 &0\\
            0 &0 &1 &1 &0 &0\\
            0 &0 &0 &0 &1 &1
        \end{pmatrix}
        \in\mathcal P^{>0}_{v,w}\subset\text{Gr}(3,6)_{\geq0}
    \end{equation}

    where $w=s_5s_3s_4s_1s_2s_3$ and $v=11s_411s_3$. Let $n=30$ and choose generic parameters 
    
    \begin{equation}
        \{\lambda_1=e^{-1.5},\lambda_2=e^{-1},\lambda_3=e^{-0.2},\lambda_4=e^{0.5},\lambda_5=e^{1.4},\lambda_6=e^2\}.
    \end{equation}

    This example compares the output given by the algorithm with the contour plot of the solution $V_{30,A}(x_{-1},x_1)$ associated to the matrix $A$. In \cref{fig:algorithm}, the contour plot of $V_{n,A}(x_{-1},x_1)$ is displayed at the top and the \reflectbox{L}-diagram and the output of the algorithm on the second line. We use the expressions for $\tau^{(n)}(x_{-1},x_1)$ and $V_{n,A}(x_{-1},x_1)$ from equations \eqref{eq-tauexpression} and \eqref{eq-solution}, respectively. The picture at the top of \cref{fig:algorithm} was produced using \textsc{Mathematica} \cite{Mathematica}, while the picture at the bottom right was produced using \href{https://www.geogebra.org/}{GeoGebra}. In the region bounded by the lines $[1,2]$, $[1,4]$, $[3,4]$ and $[5,6]$, the exponential function indexed by the lexicographical minimum $I=\{1,3,5\}$ in $\mathcal M(A)$ dominates. Meanwhile, in the region bounded by the lines $[1,3]$ and $[5,6]$, the dominating exponential function is indexed by the lexicographical maximum $J=\{2,4,6\}$ in $\mathcal M(A)$.

    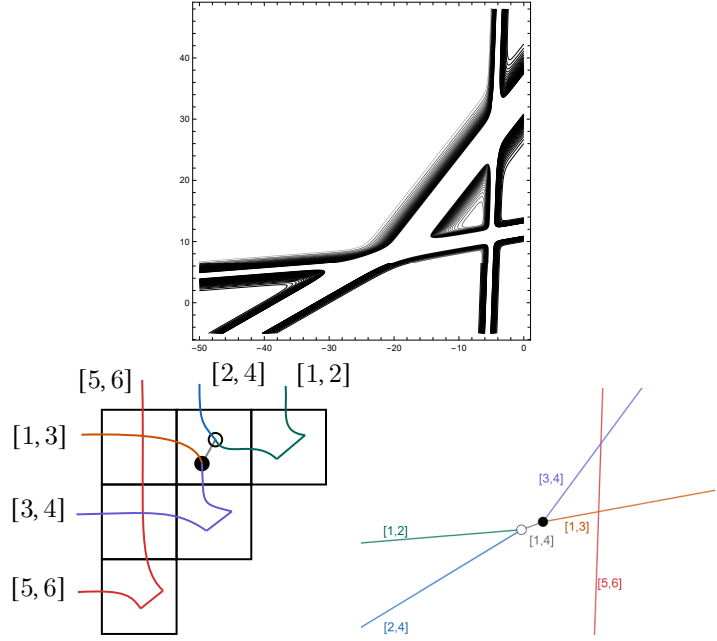
\begin{figure}
        \centering
        \includegraphics[width=0.3\linewidth]{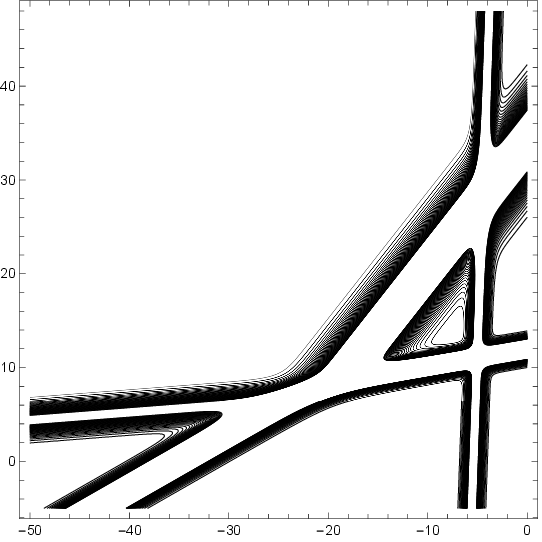}
        \tikzset{every picture/.style={line width=0.75pt}} 

        \begin{tikzpicture}[x=0.75pt,y=0.75pt,yscale=-0.75,xscale=0.75]
            
            \draw   (400,526) -- (450,526) -- (450,576) -- (400,576) -- cycle ;
            \draw   (450,526) -- (500,526) -- (500,576) -- (450,576) -- cycle ;
            \draw   (450,476) -- (500,476) -- (500,526) -- (450,526) -- cycle ;
            \draw   (400,476) -- (450,476) -- (450,526) -- (400,526) -- cycle ;
            \draw   (400,576) -- (450,576) -- (450,626) -- (400,626) -- cycle ;
            \draw   (500,476) -- (550,476) -- (550,526) -- (500,526) -- cycle ;
            \draw [color={rgb, 255:red, 211; green, 47; blue, 47 }  ,draw opacity=1 ]   (427,456) .. controls (429,540) and (422,580) .. (441,597) ;
            \draw [color={rgb, 255:red, 211; green, 47; blue, 47 }  ,draw opacity=1 ]   (384,598) .. controls (417,596) and (420,601) .. (426,609) ;
            \draw [color={rgb, 255:red, 211; green, 47; blue, 47 }  ,draw opacity=1 ]   (426,609) -- (441,597) ;
            \draw [color={rgb, 255:red, 101; green, 87; blue, 210 }  ,draw opacity=1 ]   (383,546) .. controls (432,545) and (449,542) .. (470,557) ;
            \draw [color={rgb, 255:red, 101; green, 87; blue, 210 }  ,draw opacity=1 ]   (470,557) -- (487,544) ;
            \draw [color={rgb, 255:red, 0; green, 103; blue, 88 }  ,draw opacity=1 ]   (523,460) .. controls (524,478) and (519,491) .. (536,493) ;
            \draw [color={rgb, 255:red, 0; green, 103; blue, 88 }  ,draw opacity=1 ]   (517,509) -- (536,493) ;
            \draw [color={rgb, 255:red, 128; green, 128; blue, 128 }  ,draw opacity=1 ]   (467,512) -- (476,496) ;
            \draw   (471.5,496) .. controls (471.5,493.51) and (473.51,491.5) .. (476,491.5) .. controls (478.49,491.5) and (480.5,493.51) .. (480.5,496) .. controls (480.5,498.49) and (478.49,500.5) .. (476,500.5) .. controls (473.51,500.5) and (471.5,498.49) .. (471.5,496) -- cycle ;
            \draw  [fill={rgb, 255:red, 0; green, 0; blue, 0 }  ,fill opacity=1 ] (462.5,512) .. controls (462.5,509.51) and (464.51,507.5) .. (467,507.5) .. controls (469.49,507.5) and (471.5,509.51) .. (471.5,512) .. controls (471.5,514.49) and (469.49,516.5) .. (467,516.5) .. controls (464.51,516.5) and (462.5,514.49) .. (462.5,512) -- cycle ;
            \draw [color={rgb, 255:red, 199; green, 80; blue, 0 }  ,draw opacity=1 ]   (386.33,492.67) .. controls (430.33,491.67) and (463.33,494.67) .. (467,512) ;
            \draw [color={rgb, 255:red, 101; green, 87; blue, 210 }  ,draw opacity=1 ]   (467,512) .. controls (468.33,539.67) and (465.33,538.67) .. (487,544) ;
            \draw [color={rgb, 255:red, 0; green, 103; blue, 88 }  ,draw opacity=1 ]   (476,496) .. controls (482.33,509.67) and (502.33,494.67) .. (517,509) ;
            \draw [color={rgb, 255:red, 21; green, 101; blue, 192 }  ,draw opacity=1 ]   (465.33,458.67) .. controls (465.33,482.67) and (466.33,484.67) .. (476,496) ;
                
            \draw (527,443) node [anchor=north west][inner sep=0.75pt]    {$[ 1,2]$};
            \draw (470,442) node [anchor=north west][inner sep=0.75pt]    {$[ 2,4]$};
            \draw (381,447) node [anchor=north west][inner sep=0.75pt]    {$[ 5,6]$};
            \draw (337,484) node [anchor=north west][inner sep=0.75pt]    {$[ 1,3]$};
            \draw (336,533) node [anchor=north west][inner sep=0.75pt]    {$[ 3,4]$};
            \draw (336,586) node [anchor=north west][inner sep=0.75pt]    {$[ 5,6]$};
            
        \end{tikzpicture}
        \includegraphics[width=0.3\linewidth]{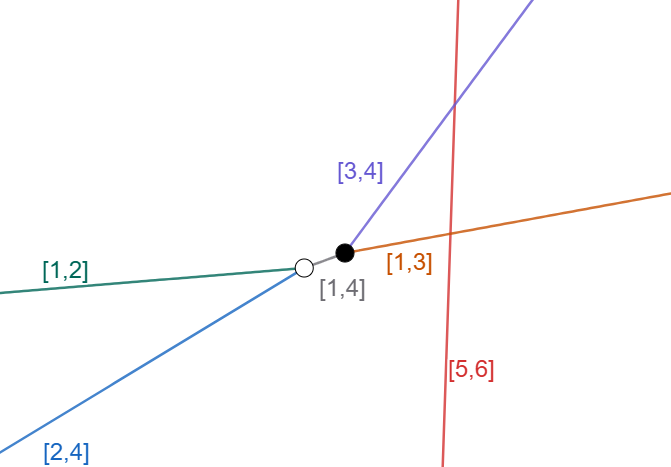}
        \caption{A comparison between the contour plot of $V_{30,A}(x_{-1},x_1)$ (top) and the output of the algorithm (bottom right).}
        \label{fig:algorithm}
    \end{figure}

\end{example}

\cref{algorithm-TL1} can also produce $\mathcal C^-(\mathcal M(A))$ in the following way: recall the definition of the dual matrix $B_A$ from \cref{sec-dual}, and consider the transformation $x_{-1}'=-x_{-1}$ and $x_1'=-x_1$. Let $n\ll0$, then

\begin{align}
    \max_{I\in\mathcal M(A)}\{\Theta_I(x_{-1},x_1,n)\}&=\max_{J\in\mathcal M(B_A)}\{\Theta_{[M]}(x_{-1},x_1,n)-\Theta_J(x_{-1},x_1,n)\}\\
    &=\max_{J\in\mathcal M(B_A)}\{\Theta_J(-x_{-1},-x_1,-n)\}\\
    &=\max_{J\in\mathcal M(B_A)}\{\Theta_J(x_{-1}',x_1',-n)\}.
\end{align}

Since $n\ll0$, it follows that $-n\gg0$, and the algorithm applied to $B_A$ in the plane $(x_{1}',x_{-1}')$ yields the contour plot $\mathcal C^-(\mathcal M(A))$.

\begin{example}\label{dual_soliton_TL}
    Here we provide the contour plot for the same $A$ and set of parameters in \cref{example-algorithm-Toda} for $n=-70$. Recall that the matrix $B_A$ is of the form 
    $\begin{pmatrix}
        -G^T|I_3
    \end{pmatrix}PD$, where
    \begin{equation}
        G=
        \begin{pmatrix}
            1 &-1 &0\\
            0 &1 &0\\
            0 &0 &1
        \end{pmatrix},
        \qquad
        P=
        \begin{pmatrix}
            1 &0 &0 &0 &0 &0\\
            0 &0 &1 &0 &0 &0\\
            0 &0 &0 &0 &1 &0\\
            0 &1 &0 &0 &0 &0\\
            0 &0 &0 &1 &0 &0\\
            0 &0 &0 &0 &0 &1
        \end{pmatrix},
        \qquad
        D=
        \begin{pmatrix}
            -1 &0 &0 &0 &0 &0\\
            0 &1 &0 &0 &0 &0\\
            0 &0 &-1 &0 &0 &0\\
            0 &0 &0 &1 &0 &0\\
            0 &0 &0 &0 &-1 &0\\
            0 &0 &0 &0 &0 &1
        \end{pmatrix}.
    \end{equation}
    The matrix $B_A$, in RREF, is the following
    \begin{equation}
        B_A=
        \begin{pmatrix}
            1 &0 &-1 &-1 &0 &0\\
            0 &1 &1 &1 &0 &0\\
            0 &0 &0 &0 &1 &1
        \end{pmatrix}.
    \end{equation}
    The reader may compare the output of the algorithm applied to the matrix $B_A$ with  top of $V_{n,A}(x_{-1},x_1)$ in \cref{fig:dualsoliton}.
    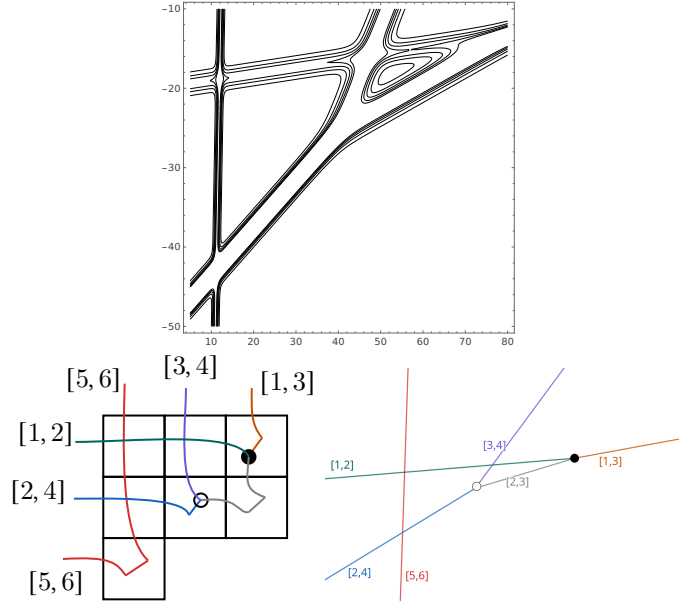
\begin{figure}[h]
        \centering
        \includegraphics[width=0.3\linewidth]{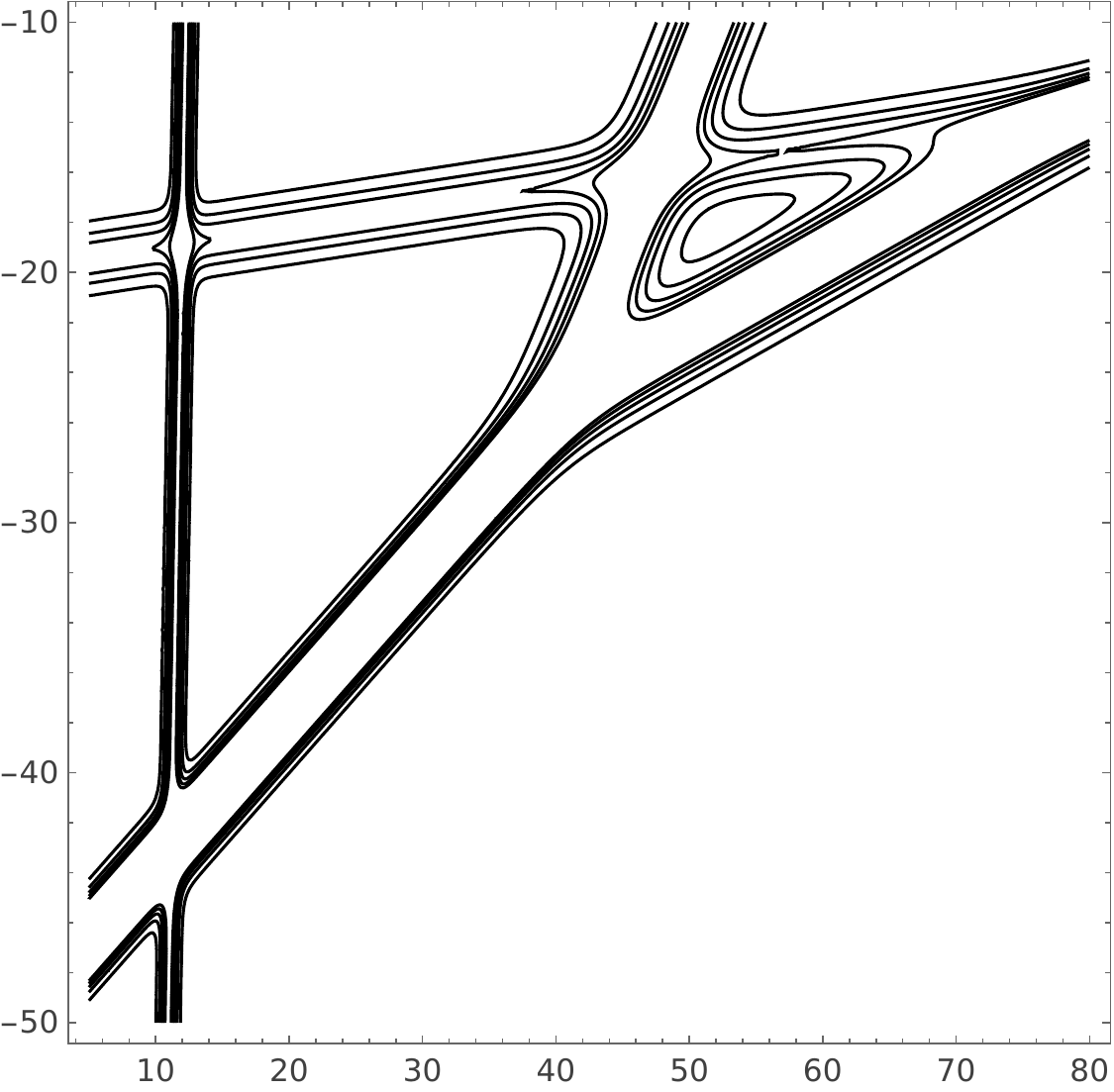}
        \tikzset{every picture/.style={line width=0.75pt}} 
        
        \begin{tikzpicture}[x=0.75pt,y=0.75pt,yscale=-0.75,xscale=0.75]

            \draw   (405,276.67) -- (446,276.67) -- (446,317.67) -- (405,317.67) -- cycle ;
            \draw   (446,276.67) -- (487,276.67) -- (487,317.67) -- (446,317.67) -- cycle ;
            \draw   (487,276.67) -- (528,276.67) -- (528,317.67) -- (487,317.67) -- cycle ;
            \draw   (405,317.67) -- (446,317.67) -- (446,358.67) -- (405,358.67) -- cycle ;
            \draw   (446,317.67) -- (487,317.67) -- (487,358.67) -- (446,358.67) -- cycle ;
            \draw   (487,317.67) -- (528,317.67) -- (528,358.67) -- (487,358.67) -- cycle ;
            \draw   (405,358.67) -- (446,358.67) -- (446,399.67) -- (405,399.67) -- cycle ;
            \draw [color={rgb, 255:red, 211; green, 47; blue, 47 }  ,draw opacity=1 ]   (419.33,255.33) .. controls (418.33,295.33) and (416,356.67) .. (435,373.67) ;
            \draw [color={rgb, 255:red, 211; green, 47; blue, 47 }  ,draw opacity=1 ]   (378,372.67) .. controls (411,370.67) and (414,375.67) .. (420,383.67) ;
            \draw [color={rgb, 255:red, 211; green, 47; blue, 47 }  ,draw opacity=1 ]   (420,383.67) -- (435,373.67) ;
            \draw [color={rgb, 255:red, 21; green, 101; blue, 192 }  ,draw opacity=1 ]   (385.33,332.33) .. controls (404.33,332.33) and (459.33,330.33) .. (462.33,343.33) ;
            \draw [color={rgb, 255:red, 21; green, 101; blue, 192 }  ,draw opacity=1 ]   (470.33,333.33) .. controls (470.33,332.33) and (462.33,343.33) .. (462.33,343.33) ;
            \draw   (465.83,333.33) .. controls (465.83,330.85) and (467.85,328.83) .. (470.33,328.83) .. controls (472.82,328.83) and (474.83,330.85) .. (474.83,333.33) .. controls (474.83,335.82) and (472.82,337.83) .. (470.33,337.83) .. controls (467.85,337.83) and (465.83,335.82) .. (465.83,333.33) -- cycle ;
            \draw [color={rgb, 255:red, 101; green, 87; blue, 210 }  ,draw opacity=1 ]   (460.33,258.33) .. controls (460.33,280.33) and (454.33,323.33) .. (470.33,333.33) ;
            \draw [color={rgb, 255:red, 199; green, 80; blue, 0 }  ,draw opacity=1 ]   (511.33,291.67) .. controls (503.33,282.33) and (504.33,272.33) .. (504.33,258.33) ;
            \draw [color={rgb, 255:red, 199; green, 80; blue, 0 }  ,draw opacity=1 ]   (502.33,304.33) .. controls (502.33,303.33) and (511.33,291.33) .. (511.33,291.67) ;
            \draw  [fill={rgb, 255:red, 0; green, 0; blue, 0 }  ,fill opacity=1 ] (497.83,304.33) .. controls (497.83,301.85) and (499.85,299.83) .. (502.33,299.83) .. controls (504.82,299.83) and (506.83,301.85) .. (506.83,304.33) .. controls (506.83,306.82) and (504.82,308.83) .. (502.33,308.83) .. controls (499.85,308.83) and (497.83,306.82) .. (497.83,304.33) -- cycle ;
            \draw [color={rgb, 255:red, 0; green, 103; blue, 88 }  ,draw opacity=1 ]   (386.33,294.33) .. controls (422.33,294.33) and (494.33,286.33) .. (502.33,304.33) ;
            \draw [color={rgb, 255:red, 128; green, 128; blue, 128 }  ,draw opacity=1 ]   (502.33,341.33) -- (513.33,331.33) ;
            \draw [color={rgb, 255:red, 128; green, 128; blue, 128 }  ,draw opacity=1 ]   (470.33,333.33) .. controls (496.33,330.33) and (494.33,342.33) .. (502.33,341.33) ;
            \draw [color={rgb, 255:red, 128; green, 128; blue, 128 }  ,draw opacity=1 ]   (502.33,304.33) .. controls (501.33,319.33) and (498.33,326.33) .. (513.33,331.33) ;
            
            \draw (345,279) node [anchor=north west][inner sep=0.75pt]    {$[ 1,2]$};
            \draw (340,318) node [anchor=north west][inner sep=0.75pt]    {$[ 2,4]$};
            \draw (377,243) node [anchor=north west][inner sep=0.75pt]    {$[ 5,6]$};
            \draw (507,245) node [anchor=north west][inner sep=0.75pt]    {$[ 1,3]$};
            \draw (442,234) node [anchor=north west][inner sep=0.75pt]    {$[ 3,4]$};
            \draw (352,379) node [anchor=north west][inner sep=0.75pt]    {$[ 5,6]$};

        \end{tikzpicture}
        \includegraphics[width=0.3\linewidth]{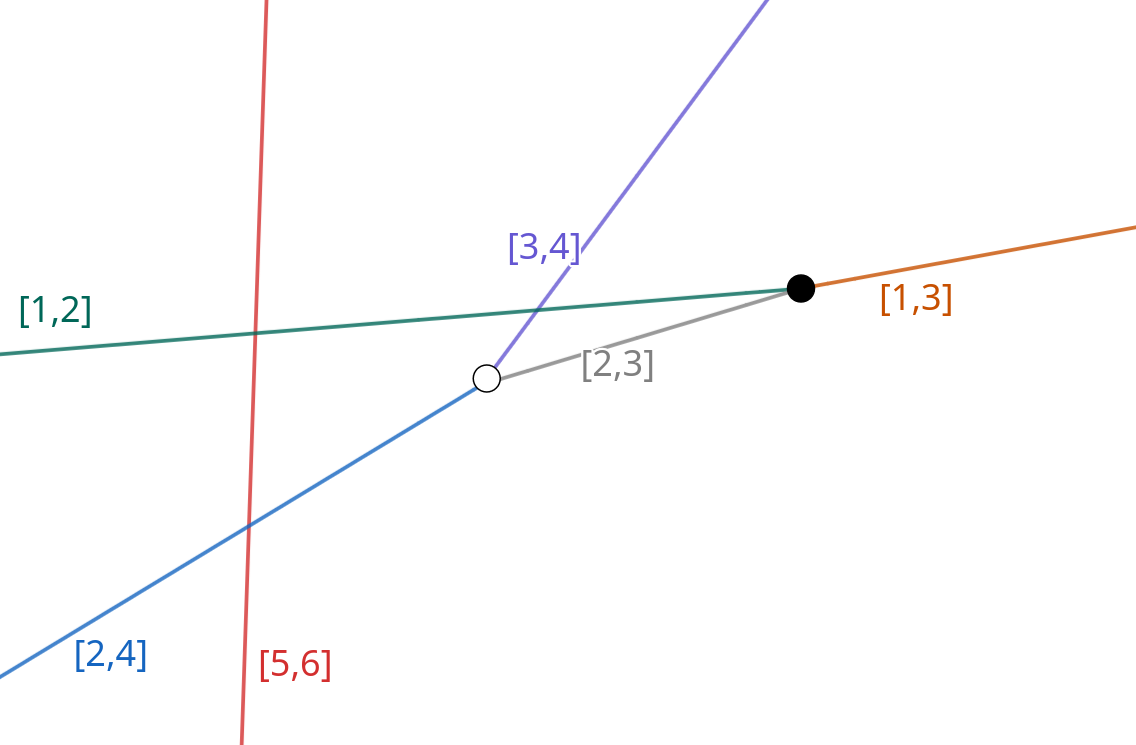}
        \caption{ A comparison between the contour plot of  $V_{-70,A}(x_{-1},x_1)$ (top) and the output of the algorithm (bottom right).}
        \label{fig:dualsoliton}
    \end{figure}
\end{example}

\begin{remark}
    Notice that the \reflectbox{L}-diagram of the matrix $B_A$ is the dual \reflectbox{L}-diagram (see \cite[Section 7.1.1]{kodama}) of the \textit{time dual \reflectbox{$\mathrm L$}-diagram} introduced in \cite[Section 8.2.2]{kodama}.
\end{remark}

\subsection{Solitons in the \texorpdfstring{$(x\textsubscript{$1$},n)$}{(x 1,n)}--plane}

In \cite{Biondini_2010}, the authors describe the asymptotic structure of 2D-Toda Lattice solitons in the $(x_1, n)$--plane, while fixing the variable $x_{-1}$. In this section, we first reproduce their results with our methods. We also provide an algorithm to obtain the corresponding asymptotic contour plots, in the same fashion as \cref{sec:algorithm}. 

Recall that, if the variables are large enough, a soliton line $[a,b]$ has equation
\begin{equation}
    L_{[a,b]}^{(n)}: x_1+n\frac{\phi_a-\phi_b}{\lambda_a-\lambda_b}-\frac{1}{\lambda_a\lambda_b}x_{-1}=0,
\end{equation}
where $\phi_i= \ln \lambda_i$, so that in the $(x_1,n)$-plane, the slope of $L_{[a,b]}^{(n)}$ is given by
\begin{equation}\label{eq-slope}
\frac{\lambda_b-\lambda_a}{\phi_a-\phi_b}.
\end{equation}
Also in this case, we will say a family of parameters $\{\lambda_1, \dots, \lambda_M\}$ is generic if the condition in Definition \ref{def-generic} is satisfied.
As before, for generic $\{\lambda_1<\dots<\lambda_M\}$, the labels of adjacent dominant phases share all but one index:
\begin{proposition}[\cite{Biondini_2010}, Theorem 3.5]
    Let $\{\lambda_1<\dots<\lambda_M\}$ be generic. Then the phases that dominate two adjacent regions of the soliton graph $\mathcal{C}^{(x_{-1})}(\mathcal{M}(A))$ are of the form $\Theta_I$ and $\Theta_J$, where $I,J\in\mathcal{M}(A)$ differ by only one element.
\end{proposition}
\begin{proof}
    The same proof of \cref{prop-dominantphases} applies. We only have to be careful when we prove that if the lines $\theta_i=\theta_k$ and $\theta_j=\theta_l$ are parallel, they cannot coincide for $x_{-1}\neq0$. Indeed, in this case, for the lines to match we must have that
    \begin{equation}
        \frac{\phi_i-\phi_k}{\lambda_i-\lambda_k}=\frac{\phi_j-\phi_l}{\lambda_j-\lambda_l},
    \end{equation}
    \textbf{and} for there to exist $x_{-1}\in\mathbb R-\{0\}$ such that
    \begin{equation}
        \frac{x_{-1}}{\lambda_i\lambda_k}=\frac{x_{-1}}{\lambda_j\lambda_l}.
    \end{equation}
    If we assume $x_{-1}\neq0$, the last equation is equivalently written as $\phi_i+\phi_k=\phi_j+\phi_l$. The same argument used in the proof of \cref{prop-dominantphases} then holds, proving the statement.
\end{proof}

In \cite{Biondini_2010}, the authors prove a result analogous to Lemma \ref{lem-dominantrel}:
\begin{lemma}[Lemma 3.6, \cite{Biondini_2010}]
    Let $\{ \lambda_1 < \dots < \lambda_M\}$ be generic parameters, and $L_{i,j}^{(n)}$ a soliton line in the $(x_1, n)$-plane for $i<j$. Then, along $L_{i,j}^{(n)}$, we have the following relations for $n \gg 0:$
    \begin{itemize}
        \item $\theta_i=\theta_j < \theta_m$ for $i<m<j$;
        \item $\theta_i=\theta_j>\theta_m$ for $1\leq m <i$ or $j<m\leq M$.
    \end{itemize}
\end{lemma}
The following result follows from the previous lemma, in the same way that \cref{asymptotic} follows from \cref{lem-dominantrel}:


\begin{theorem}[Theorem 3.8 \cite{Biondini_2010}]\label{thm-asymptotic_n}
    Let $A \in \mathcal{P}_{v,w}^{>0} \subset \text{Gr}\,(N,M)_{>0}$ be an irreducible matrix, with pivot set $\{i_1, \dots i_N\}$ and non-pivot set $\{j_1, \dots j_{M-N}\}$. Assume also that the parameters $\{\lambda_1, \dots , \lambda_M\}$ of the TL soliton $V_{n,A}(x_{-1},x_1)$ are generic. Then $\mathcal{C}^{(x_{-1})}(\mathcal{M}(A))$ has the following asymptotic structure: 
    \begin{enumerate}
        \item For $n\gg 0$, there exist $M-N$ line solitons of type $[q_m, j_m]$, where $q_m<j_m$.
        \item For $n\ll 0$, there exist $N$ line solitons of type $[i_s, p_s]$, where $i_s<p_s$.
    \end{enumerate}
    Additionally, the map $\pi:[M]\to[M]$ defined by $\pi(i_s)=p_s$ and $\pi(j_m)=q_m$ satisfies $\pi=vw^{-1}$.
\end{theorem}
\begin{remark}
In \cite[Theorem 3.8]{Biondini_2010}, the authors had already described the asymptotic structure of soliton solutions in the $(x_1,n)$-plane. However, we emphasize that the result relating the permutation $\pi$ with the Deodhar component of the matrix parameter is new. 
\end{remark}

Following \cref{sec:algorithm}, in order to provide an algorithm we need to compute the coordinates of resonant interaction points $v_{[a,b,c]}=(x_{1,[a,b,c]},n_{[a,b,c]})$ and understand when they are visible. We take $A\in\text{Gr}(N,M)_{\geq0}$ and $a,b,c\in[M]$ such that $a<b<c$. The resonant interaction point $v_{[a,b,c]}$ obtained by the intersection of the line solitons labeled by $[a,b],[b,c]$, and $[a,c]$, has coordinates given by the following formulas:
\begin{align}
    &x_{1,[a,b,c]}=\frac{(\phi_c-\phi_b)\lambda_a^{-1}+(\phi_a-\phi_c)\lambda_b^{-1}+(\phi_b-\phi_a)\lambda_c^{-1}}{(\phi_b-\phi_c)\lambda_a+(\phi_c-\phi_a)\lambda_b+(\phi_a-\phi_b)\lambda_c}x_{-1},\\
    &n_{[a,b,c]}=\frac{(\lambda_b-\lambda_a)(\lambda_a-\lambda_c)(\lambda_b-\lambda_c)}{\lambda_a\lambda_b\lambda_c((\phi_c-\phi_b)\lambda_a+(\phi_a-\phi_c)\lambda_b+(\phi_b-\phi_a)\lambda_c)}x_{-1}.
\end{align}
As in the $(x_{1}, x_{-1})$-plane, we study the sign of the function
\begin{equation}
    G_{a,b,c}(\phi_d)=z_d-z_{[a,b,c]}=\frac{f_{a,b,c}(\phi_d)}{\lambda_a\lambda_b\lambda_c\lambda_d((\phi_c-\phi_b)\lambda_a+(\phi_a-\phi_c)\lambda_b+(\phi_b-\phi_a)\lambda_c)}x_{-1},
\end{equation}
where the function $f_{a,b,c}(d)$ is defined in \cref{eq-functionf}.
\begin{proposition}\label{prop-signofG}
    Let $a,b,c\in[M]$ such that $a<b<c$ and $x_{-1}>0$, then
    \begin{enumerate}
        \item $G_{a,b,c}(\phi_d)>0$ for $d<a$ and $b<d<c$;
        \item $G_{a,b,c}(\phi_d)<0$ for $a<d<b$ and $c<d$.
    \end{enumerate}
\end{proposition}
\begin{proof}
    At the end of the proof of \cref{prop-signF} we proved that
    \begin{equation}
        (\phi_c-\phi_b)\lambda_a+(\phi_a-\phi_c)\lambda_b+(\phi_b-\phi_a)\lambda_c>0\qquad\text{for }a<b<c.
    \end{equation}
    Moreover, $\lambda_a\lambda_b\lambda_c\lambda_d>0$, and therefore
    \begin{equation}
        \frac{1}{\lambda_a\lambda_b\lambda_c\lambda_d((\phi_c-\phi_b)\lambda_a+(\phi_a-\phi_c)\lambda_b+(\phi_b-\phi_a)\lambda_c)}x_{-1}>0.
    \end{equation}
    Hence, the sign of the function $G_{a,b,c}(\phi_d)$ is the same as that of the function $f_{a,b,c}(\phi_d)$; therefore, the proof follows from \cref{prop-signF}.
\end{proof}
Since the functions $G_{a,b,c}$ have exactly three zeroes, it follows that, with respect to generic parameters, asymptotic contour plots in the $(x_1,n)$-plane also only admit X-crossings and trivalent vertices as interaction points. The following result follows from \cref{prop-signofG}, analogously to how \cref{vis_ver} follows from \cref{prop-signF}:
\begin{proposition}
    Let $x_{-1}>0$ and $A\in\mathcal P^{>0}_{v,w}$ be an irreducible matrix. Then the trivalent vertices in the reduced pipedream are visible. In particular, if the vertex in the pipedream is white (black), then it corresponds to a resonant interaction point marked by a white (black) vertex.
\end{proposition}
Recall how it was enough to guarantee the visibility of vertices in the $\reflectbox{L}$-diagram to produce an algorithm in the $(x_1,x_{-1})$--plane; we are therefore ready to write down the main result of this subsection. Let $x_{-1}\gg0$ and $A\in\mathcal P^{>0}_{v,w}$. The following procedure yields the contour plot $\mathcal{C}^+(\mathcal{M}(A))$ in the $(x_1, n)$-plane:
 \begin{algorithm}\label{algorithm-TL2}
    \begin{enumerate}
        \item[]
        \item Draw the \reflectbox{$\mathrm L$}-diagram from the pair $(v,w)$ or the derangement $\pi=vw^{-1}$.
        \item Draw the reduced pipedream.
        \item Plot in the $(x_1,n)$-plane all the trivalent vertices $v_{[a,b,c]}=(x_{1,[a,b,c]},n_{[a,b,c]})$, $a<b<c$, in the \reflectbox{$\mathrm{L}$}--diagram, where
        \begin{align}
            &x_{1,[a,b,c]}=\frac{(\phi_c-\phi_b)\lambda_a^{-1}+(\phi_a-\phi_c)\lambda_b^{-1}+(\phi_b-\phi_a)\lambda_c^{-1}}{(\phi_b-\phi_c)\lambda_a+(\phi_c-\phi_a)\lambda_b+(\phi_a-\phi_b)\lambda_c}x_{-1}\\
            &n_{[a,b,c]}=\frac{(\lambda_b-\lambda_a)(\lambda_a-\lambda_c)(\lambda_b-\lambda_c)}{\lambda_a\lambda_b\lambda_c((\phi_c-\phi_b)\lambda_a+(\phi_a-\phi_c)\lambda_b+(\phi_b-\phi_a)\lambda_c)}x_{-1}.
        \end{align}
        We have two cases:
        \begin{itemize}
            \item if the vertex is white: draw the line $L_{[a,c]}$ for $n>n_{[a,b,c]}$ and the lines $L_{[a,b]}$ and $L_{[b,c]}$ for $n<n_{[a,b,c]}$;
            \item if the vertex is black: draw the line $L_{[a,c]}$ for $n<n_{[a,b,c]}$ and the lines $L_{[a,b]}$ and $L_{[b,c]}$ for $n>n_{[a,b,c]}$.
        \end{itemize}
        The equation of the line $[a,b]$ is given in \cref{eq-line-Toda}.
        \item Draw an infinite line $L_{[i,j]}$ every time there is a path labeled $[i,j]$ that does not connect two vertices in the pipedream.
    \end{enumerate}
\end{algorithm}

\begin{example}
    Consider the same data in \cref{example-algorithm-Toda} and $x_{-1}=10$. In \cref{fig:algorithm2}, we can see how the contour plot of the solution $V_{n,A}(10,x_1)$ and the output of the algorithm match.
    \begin{figure}
        \centering

        \includegraphics[width=0.3\linewidth]{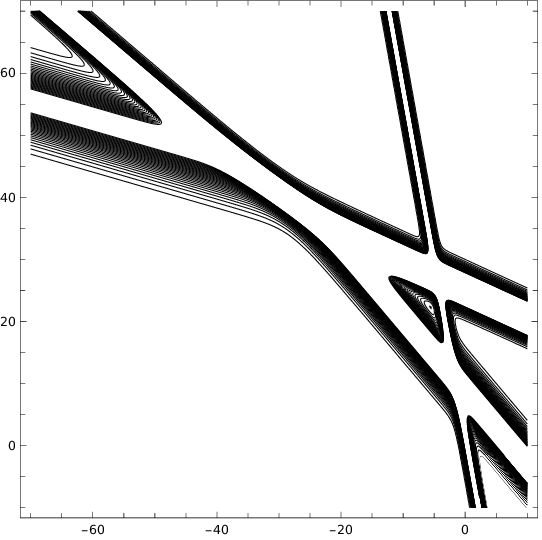}
        \tikzset{every picture/.style={line width=0.75pt}} 

        \begin{tikzpicture}[x=0.75pt,y=0.75pt,yscale=-0.75,xscale=0.75]
        
            \draw   (400,526) -- (450,526) -- (450,576) -- (400,576) -- cycle ;
            \draw   (450,526) -- (500,526) -- (500,576) -- (450,576) -- cycle ;
            \draw   (450,476) -- (500,476) -- (500,526) -- (450,526) -- cycle ;
            \draw   (400,476) -- (450,476) -- (450,526) -- (400,526) -- cycle ;
            \draw   (400,576) -- (450,576) -- (450,626) -- (400,626) -- cycle ;
            \draw   (500,476) -- (550,476) -- (550,526) -- (500,526) -- cycle ;
            \draw [color={rgb, 255:red, 211; green, 47; blue, 47 }  ,draw opacity=1 ]   (427,456) .. controls (429,540) and (422,580) .. (441,597) ;
            \draw [color={rgb, 255:red, 211; green, 47; blue, 47 }  ,draw opacity=1 ]   (384,598) .. controls (417,596) and (420,601) .. (426,609) ;
            \draw [color={rgb, 255:red, 211; green, 47; blue, 47 }  ,draw opacity=1 ]   (426,609) -- (441,597) ;
            \draw [color={rgb, 255:red, 101; green, 87; blue, 210 }  ,draw opacity=1 ]   (383,546) .. controls (432,545) and (449,542) .. (470,557) ;
            \draw [color={rgb, 255:red, 101; green, 87; blue, 210 }  ,draw opacity=1 ]   (470,557) -- (487,544) ;
            \draw [color={rgb, 255:red, 0; green, 103; blue, 88 }  ,draw opacity=1 ]   (523,460) .. controls (524,478) and (519,491) .. (536,493) ;
            \draw [color={rgb, 255:red, 0; green, 103; blue, 88 }  ,draw opacity=1 ]   (517,509) -- (536,493) ;
            \draw [color={rgb, 255:red, 128; green, 128; blue, 128 }  ,draw opacity=1 ]   (467,512) -- (476,496) ;
            \draw   (471.5,496) .. controls (471.5,493.51) and (473.51,491.5) .. (476,491.5) .. controls (478.49,491.5) and (480.5,493.51) .. (480.5,496) .. controls (480.5,498.49) and (478.49,500.5) .. (476,500.5) .. controls (473.51,500.5) and (471.5,498.49) .. (471.5,496) -- cycle ;
            \draw  [fill={rgb, 255:red, 0; green, 0; blue, 0 }  ,fill opacity=1 ] (462.5,512) .. controls (462.5,509.51) and (464.51,507.5) .. (467,507.5) .. controls (469.49,507.5) and (471.5,509.51) .. (471.5,512) .. controls (471.5,514.49) and (469.49,516.5) .. (467,516.5) .. controls (464.51,516.5) and (462.5,514.49) .. (462.5,512) -- cycle ;
            \draw [color={rgb, 255:red, 199; green, 80; blue, 0 }  ,draw opacity=1 ]   (386.33,492.67) .. controls (430.33,491.67) and (463.33,494.67) .. (467,512) ;
            \draw [color={rgb, 255:red, 101; green, 87; blue, 210 }  ,draw opacity=1 ]   (467,512) .. controls (468.33,539.67) and (465.33,538.67) .. (487,544) ;
            \draw [color={rgb, 255:red, 0; green, 103; blue, 88 }  ,draw opacity=1 ]   (476,496) .. controls (482.33,509.67) and (502.33,494.67) .. (517,509) ;
            \draw [color={rgb, 255:red, 21; green, 101; blue, 192 }  ,draw opacity=1 ]   (465.33,458.67) .. controls (465.33,482.67) and (466.33,484.67) .. (476,496) ;
            
            \draw (527,443) node [anchor=north west][inner sep=0.75pt]    {$[ 1,2]$};
            \draw (470,442) node [anchor=north west][inner sep=0.75pt]    {$[ 2,4]$};
            \draw (381,447) node [anchor=north west][inner sep=0.75pt]    {$[ 5,6]$};
            \draw (337,484) node [anchor=north west][inner sep=0.75pt]    {$[ 1,3]$};
            \draw (336,533) node [anchor=north west][inner sep=0.75pt]    {$[ 3,4]$};
            \draw (336,586) node [anchor=north west][inner sep=0.75pt]    {$[ 5,6]$};
                
        \end{tikzpicture}
        \includegraphics[width=0.3\linewidth]{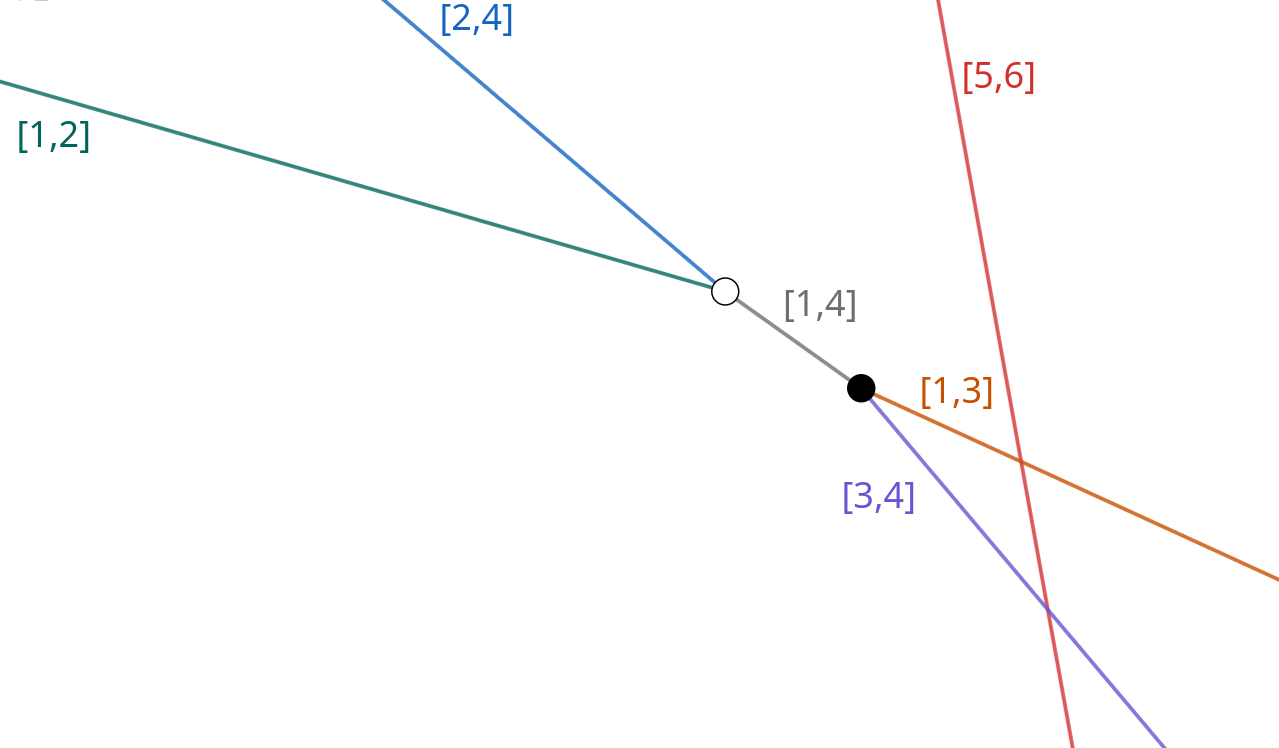}

        \caption{A comparison between the contour plot of $V_{n,A}(10,x_1)$ (top) and the output of the algorithm (bottom right).}
        \label{fig:algorithm2}
    \end{figure}
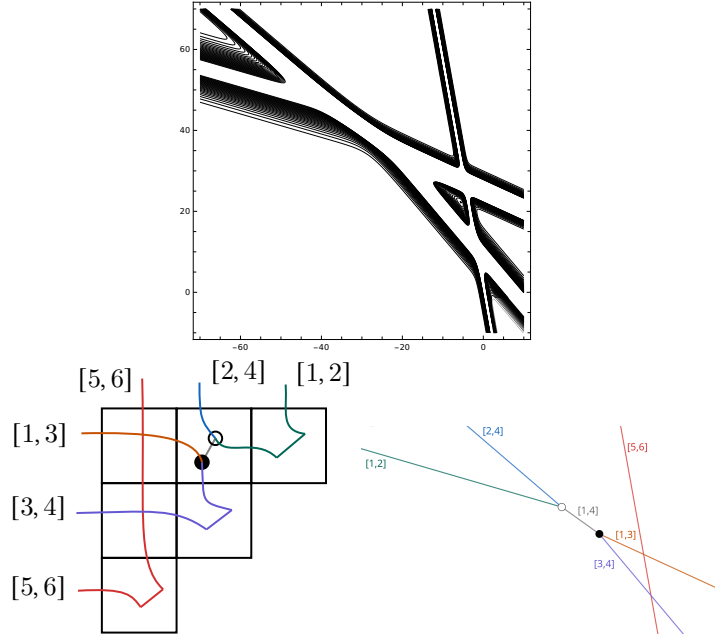
\end{example}

\section{Davey-Stewartson equation}\label{sec:Davey-Stewartson}

\subsection{Asymptotics of soliton solutions}

Analogously to when we studied Toda lattice solitons, we will assume that our Davey-Stewartson $\tau$-functions are parametrized by irreducible matrices. We will also assume that the parameters $\{\psi_1, \dots , \psi_M\}$ that appear in our equations are \textit{generic} and belong to the interval $(-\frac\pi2,\frac\pi2)$. For this equation we have a different set of genericity conditions that we write presently.
\begin{definition}\label{DS-genericity}
    A set of real parameters $\{\psi_1, \dots , \psi_M\}$ is generic if it satisfies the following conditions:
    \begin{itemize}
        \item $-\frac{\pi}{2}<\psi_i+\psi_j+\psi_k<\frac{\pi}{2}$ for all distinct $i,j,k \in [M]$;
        \item the sums $\sum_{j=1}^m\sin\psi_{i_j}$ are all distinct for $m=1,\dots,M$.
    \end{itemize}
\end{definition}

\begin{proposition}
    Suppose that $\mathcal{C}(Q_{t,A})$ depends on a set of generic parameters $\{\psi_j\}_{j=1}^M$. Then the labels of any two adjacent regions on $\mathcal{C}^{(t)}(\mathcal{M}(A))$ share $N-1$ indices.
\end{proposition}
\begin{proof}
    Firstly, notice that, for any $K \in \binom{[M]}{N}$,
    \begin{equation}
        \Theta_K(x,y,t)=2x\sum_{k\in K}\sin \psi_k+2y \sum_{k\in K} \cos \psi_k-2t\sum_{k\in K} \sin 2 \psi_k.
    \end{equation}
    Therefore, our genericity assumption implies that the functions $\Theta_{K_1}$ and $\Theta_{K_2}$ are never equal for $K_1 \neq K_2$.
    Looking at the proof of Proposition \ref{prop-dominantphases}, we only have to prove that the lines $\theta_i=\theta_k$ and $\theta_j=\theta_l$ cannot coincide. They coincide if the slopes are equal, i.e. $\psi_i+\psi_k=\psi_j+\psi_l$, and there exists $t\in\mathbb R$ such that
    \begin{equation}\label{translationDS}
        t\frac{\sin(2\psi_k)-\sin(2\psi_i)}{\sin(\psi_k)-\sin(\psi_i)}=t\frac{\sin(2\psi_l)-\sin(2\psi_j)}{\sin(\psi_l)-\sin(\psi_j)}.
    \end{equation}
    Following the same steps of the proof of \cref{prop-dominantphases}, assume by contradiction that there exists $t\neq0$ such that \cref{translationDS} holds true. Let $\psi_i+\psi_k=2m=\psi_j+\psi_l$ with $-\pi< 2m<\pi$. Then there exist $\alpha,\beta \in \mathbb{R}$ such that
    \begin{align}
        &\psi_i=m-\alpha,\qquad\psi_k=m+\alpha,\\
        &\psi_j=m-\beta,\qquad\psi_l=m+\beta.
    \end{align}
    Note that $-\frac{\pi}{2}< m-\alpha<\frac\pi2$ and $-\frac\pi2< m+\alpha<\frac\pi2$, therefore $-\frac\pi2<\alpha<\frac\pi2$. The same holds for $\beta$. Algebraic manipulation of \cref{translationDS} leads to the equality
    \begin{equation}
        \cos(\alpha)=\cos(\beta),
    \end{equation}
    which is satisfied if $\alpha=\beta$ or $\alpha=-\beta$ leading to $\psi_i=\psi_j$ and $\psi_k=\psi_l$ or $\psi_i=\psi_l$ and $\psi_k=\psi_j$. But this contradicts our genericity assumption from \cref{DS-genericity}.
\end{proof}
In the proof above we are using the fact that (for large parameters) the line $\theta_a=\theta_b$ has the following equation
\begin{equation}\label{DS-solitonline}
    L_{[a,b]}: x-\tan\left(\frac{\psi_a+\psi_b}{2}\right)y-t\frac{\sin(2\psi_a)-\sin(2\psi_b)}{\sin(\psi_a)-\sin(\psi_b)}=0.
\end{equation}

\begin{remark}
    An interesting feature of the Davey--Stewartson solution $Q(x,y,t)$ is the appearance of V-shaped and L-shaped solitons \cite[Section 4.3]{biondini2022soliton}. Their contour plots are not balanced because of the presence of two-valent vertices. As already mentioned in \cref{example-linesolitonDS}, this situation occurs when $Q(x,y,t)$ localizes along a line $L_{[a,b]}$ such that $\sin\psi_a=\sin\psi_b$, i.e. along a horizontal line. However, the genericity conditions imposed in this work exclude this behavior of the solution.
\end{remark}

Analogously to the Toda lattice case, understanding the asymptotics of Davey--Stewartson solitons is equivalent to understanding the dominant relations among the $\theta$-phases.
\begin{lemma}\label{lem-dominantrel-DS}
    Let $\{\psi_1<\dots<\psi_M\}$ be generic and $L_{i,j}$ be the $[i,j]$-soliton line for $i<j$.  Then,  along $L_{i,j}$, we have the following dominant relations for $y\gg 0$:
    \begin{itemize}
        \item $\theta_i=\theta_j>\theta_m$ for $1\leq m<i$ or $j<m\leq M$;
        \item $\theta_i=\theta_j<\theta_m$ for $i<m<j$.
    \end{itemize}
\end{lemma}
\begin{proof}
    The proof proceeds analogously to that of Lemma \ref{lem-dominantrel}. Along $L_{[i,j]}$, we can use equation \eqref{DS-solitonline} to substitute for $x$ in the difference $\theta_k(x,y,t)-\theta_i(x,y,t)$. In the limit $y \gg 0$, the sign of this difference will be determined by the coefficient of $y$, which is given by evaluating the following function on $\psi_k$:
    \begin{equation}
        \begin{split}
            \alpha: \hspace{1em}\biggl(-\frac{\pi}{2}, \frac{\pi}{2}\biggr) &\longrightarrow \mathbb{R}\\
            \psi &\longmapsto 2 \left[ \tan \frac{\psi_i + \psi_j}{2} (\sin \psi -\sin \psi_i)+\cos\psi-\cos \psi_i \right].
        \end{split}
    \end{equation}
    To complete the proof, it is enough to list the behavior of the function $\alpha:$
    \begin{itemize}
        \item $\alpha(\psi_i)=\alpha(\psi_j)=0$;
        \item The function $\eta(\psi)$ has a maximum at the point $\psi=\frac{\psi_i+\psi_j}{2}$;
        \item The function $\alpha$ is concave for $-\frac{\pi}{2}<\psi\leq\frac{\psi_i+\psi_j+\pi}{2}$ and convex for $\frac{\psi_i+\psi_j+\pi}{2}<\psi<\frac{\pi}{2}$ if $\frac{\psi_i+\psi_j-\pi}{2}<-\frac{\pi}{2}$, otherwise it is convex for $-\frac{\pi}{2}<\psi\leq\frac{\psi_i+\psi_j-\pi}{2}$ and concave for $\frac{\psi_i+\psi_j-\pi}{2}\leq\psi<\frac{\pi}{2}$.\qedhere
    \end{itemize}
\end{proof}

The proof of the following Theorem is functionally that of \cref{asymptotic} and \cref{fan_is_der} , in the sense that the combinatorial arguments are identical.
\begin{theorem}\label{DS-asymptotic}
    Let $A \in \mathcal{P}_{v,w}^{>0} \subset \mathrm{Gr}\,(N,M)_{\geq0}$ be an irreducible matrix, with pivot set $\{i_1, \dots i_N\}$ and non-pivot set $\{j_1, \dots j_{M-N}\}$. Let $\{\psi_j\}_{j=1}^M$ be generic parameters. Then $\mathcal{C}^{(t)}(\mathcal{M}(A))$ has the following asymptotic structure 
        \begin{itemize}
            \item For $y\ll 0$, there exist $N$ line solitons of type $[i_k, p_k]$, where $p_k>i_k$.
            \item For $y\gg 0$, there exist $M-N$ line solitons of type $[q_l, j_l]$, where $q_l < j_l$.
        \end{itemize}
        Furthermore, the permutation $\pi\in S_M$ given by $\pi(i_k)=p_k$ for all $k=1,\dots,N$ and $\pi(j_l)=q_l$ for all $l=1,\dots,M-N$ satisfies that $\pi=vw^{-1}$.
\end{theorem}

\subsection{Algorithm}

As in \cref{sec:algorithm}, we provide an algorithm to produce the soliton graph of soliton solutions of the DSII equation.\\
Let $A\in\text{Gr}(N,M)_{\geq0}$, $\{\psi_1,\dots,\psi_M\}$ be a set of generic parameters, and $I,J\in\mathcal M(A)$ two index sets such that $\Theta_I$ and $\Theta_J$ are dominant in two adjacent regions. Recall that there exist $a,b\in[M]$ such that $J=(I\setminus\{a\})\cup\{b\}$ and the line bounding the regions is given by the $[a,b]$-soliton,
\begin{equation}\label{eq-line}
    L_{[a,b]}: x-\tan\left(\frac{\psi_a+\psi_b}{2}\right)y-t\frac{\sin(2\psi_a)-\sin(2\psi_b)}{\sin(\psi_a)-\sin(\psi_b)}=0.
\end{equation}
\begin{remark}
    Note that we are ignoring constant terms in the equation of the $[a,b]$-soliton since we consider $t\gg0$ or $t\ll0$.
\end{remark}
The coordinates of a resonant interaction point $v_{[a,b,c]}=(x_{[a,b,c]},y_{[a,b,c]})$ are given by the following formulas:
\begin{align}
    &x_{[a,b,c]}=t(\cos(\psi_a)+\cos(\psi_b)+\cos(\psi_c)-\cos(\psi_a+\psi_b+\psi_c)),\\
    &y_{[a,b,c]}=t(\sin(\psi_a)+\sin(\psi_b)+\sin(\psi_c)+\sin(\psi_a+\psi_b+\psi_c)).
\end{align}
Consider now the hyperplanes in the $(x,y,z)$-space given by the equations
\begin{equation}
    z=\theta_i(x,y,t)\qquad\forall i\in[M],
\end{equation}
and define $z_d=\theta_d(x_{[a,b,c]},y_{[a,b,c]},t)$.
\begin{remark}
    Since $\theta_a(x_{[a,b,c]},y_{[a,b,c]},t)=\theta_b(x_{[a,b,c]},y_{[a,b,c]},t)=\theta_c(x_{[a,b,c]},x_{[a,b,c]},t)$, we have
    \begin{equation}
        z_{[a,b,c]}=z_a=z_b=z_c=t(\sin(\psi_a+\psi_b)+\sin(\psi_a+\psi_c)+\sin(\psi_b+\psi_c)).
    \end{equation}
\end{remark}
\begin{definition}
    Let $a,b,c\in[M]$ such that $a<b<c$ and consider $v_{[a,b,c]}$. We say that $v_{[a,b,c]}$ is \textbf{visible} if $z_d<z_{[a,b,c]}$ for all $d\in[M]\setminus\{a,b,c\}$.
\end{definition}
Therefore, determining visibility is equivalent to determining the sign of the function
\begin{equation}
    F_{a,b,c}(\psi_d)=z_d-z_{[a,b,c]}=8t\sin\left(\frac{\psi_a-\psi_d}{2}\right)\sin\left(\frac{\psi_c-\psi_d}{2}\right)\sin\left(\frac{\psi_d-\psi_b}{2}\right)\cos\left(\frac{\psi_a+\psi_b+\psi_c+\psi_d}{2}\right).
\end{equation}
Notice that our genericity assumption \eqref{DS-genericity} guaranties the positivity of the rightmost factor in $F_{a,b,c}(\psi_d)$. Therefore, $F_{a,b,c}$ has only three zeros: $\psi_d=\psi_k$ for $k\in\{a,b,c\}$. In particular, this implies the only interaction points featuring in asymptotic contour plots are X-crossings and trivalent vertices. We have the following:
\begin{proposition}
    Let $a,b,c\in[M]$ such that $a<b<c$ and $t<0$, then
    \begin{enumerate}
        \item $F_{a,b,c}(\psi_d)>0$ for $d<a$ and $b<d<c$;
        \item $F_{a,b,c}(\psi_d)<0$ for $a<d<b$ and $c<d$.
    \end{enumerate}
\end{proposition}
\begin{proof}
    The proof is a case-by-case analysis of the function $F_{a,b,c}$. If $d<a$, we have that
    \begin{align}
        &0<\frac{\psi_a-\psi_d}{2}<\frac{\pi}{2}\implies\sin\left(\frac{\psi_a-\psi_d}{2}\right)>0;\\
        &0<\frac{\psi_c-\psi_d}{2}<\frac{\pi}{2}\implies\sin\left(\frac{\psi_c-\psi_d}{2}\right)>0;\\
        &0<\frac{\psi_b-\psi_d}{2}<\frac{\pi}{2}\implies\sin\left(\frac{\psi_d-\psi_b}{2}\right)<0.
    \end{align}
    Since $t>0$, we have that $F_{a,b,c}(\psi_d)>0$ for $d<a$. Analogously, we get all the other inequalities.
\end{proof}
\begin{proposition}
    Let $t<0$ and $A\in\mathcal P^{>0}_{v,w}$ be an irreducible matrix. Then the trivalent vertices in the reduced pipedream are visible. In particular, if the vertex in the pipedream is white (black), then it corresponds to a resonant interaction point marked by a white (black) vertex.
\end{proposition}
The proof of this result is identical to the one of Proposition \ref{vis_ver} and for this reason it is omitted.\\
Let $t\ll0$ and $A\in\mathcal P^{>0}_{v,w}$. We now adapt the algorithms proposed in \cite{kodama} \cite{kodama2014kp} to construct the soliton graph $\mathcal C^-(\mathcal M(A))$ starting from the $\reflectbox L$-diagram associated to $A$:
\begin{algorithm}\label{algorithm_DS}
    \begin{enumerate}
        \item[] 
        \item Draw the \reflectbox{$\mathrm L$}-diagram from the pair $(v,w)$ or the derangement $\pi=vw^{-1}$.
        \item Draw the reduced pipedream.
        \item Plot in the $(x,y)$-plane all the trivalent vertices $v_{[a,b,c]}=(x_{[a,b,c]},y_{[a,b,c]})$, $a<b<c$, in the \reflectbox{$\mathrm{L}$}--diagram, where
        \begin{align}
            &x_{[a,b,c]}=t(\cos(\psi_a)+\cos(\psi_b)+\cos(\psi_c)-\cos(\psi_a+\psi_b+\psi_c)),\\
            &y_{[a,b,c]}=t(\sin(\psi_a)+\sin(\psi_b)+\sin(\psi_c)+\sin(\psi_a+\psi_b+\psi_c)).
        \end{align}
        We have two cases:    
        \begin{itemize}
            \item if the vertex is white: draw the line $L_{[a,c]}$ for $y<y_{[a,b,c]}$ and the lines $L_{[a,b]}$ and $L_{[b,c]}$ for $y>y_{[a,b,c]}$;
            \item if the vertex is black: draw the line $L_{[a,c]}$ for $y>y_{[a,b,c]}$ and the lines $L_{[a,b]}$ and $L_{[b,c]}$ for $y<y_{[a,b,c]}$.
        \end{itemize}
        The equation of the line $[a,b]$ is given in \cref{eq-line}.
        \item Draw an infinite line $L_{[i,j]}$ every time there is a path labeled $[i,j]$ that does not connect two vertices in the pipedream.
    \end{enumerate}
\end{algorithm}
The proof of the faithfulness of the previous algorithm is analogous to that Algorithm \ref{algorithm-TL1}, and is also omitted.

\begin{remark}
    Note an analogous argument to the one used for the 2D Toda Lattice allows us to construct the soliton graph $\mathcal{C}^+(\mathcal{M}(A))$ from the previous algorithm. Let $t\gg 0$, and consider the change of variables $x'=-x$, $y'=-y$. Since the function $\Theta_I(x, y,t)$ is linear in its variables,
    \begin{equation}
        \max_{I\in \mathcal{M}(A)} \{ \Theta_I(x,y,t)\}=\max_{J\in \mathcal{M}(B_A)} \{ -\Theta_J(x,y,t)\}=\max_{J \in \mathcal{M}(B_A)}\{\Theta(x',y',-t)\},
    \end{equation}
    where $B_A$ is the \textit{dual matrix} to $A$ defined in Section \ref{sec-dual}. Since $-t \ll 0$ by assumption, $\mathcal C^+(\mathcal{M}(A))$ can be recovered from the soliton graph $\mathcal C^-(\mathcal{M}(B_A))$ drawn in the $(x',y')$-plane.
\end{remark}

\begin{example}\label{example-algorithm}
    We show the output of \cref{algorithm_DS} applied to the matrix
    \begin{equation}
        A=
        \begin{pmatrix}
            1 &1 &0 &-1 &0 &0\\
            0 &0 &1 &1 &0 &0\\
            0 &0 &0 &0 &1 &1
        \end{pmatrix}
        \in\mathcal P^{>0}_{v,w}\subset\text{Gr}(3,6)_{\geq0}
    \end{equation}
    from \cref{example-algorithm-Toda}, where $w=s_5s_3s_4s_1s_2s_3$ and $v=11s_411s_3$. Let $t=-12$ and choose generic parameters    
    \begin{equation}
        \left\{\psi_1=-\frac{2\pi}{5},\psi_2=-\frac{\pi}{4},\psi_3=-\frac{\pi}{12},\psi_4=\frac{\pi}{6},\psi_5=\frac{2\pi}{7},\psi_6=\frac{3\pi}{7}\right\}.
    \end{equation}
    In \cref{fig:algorithm-Davey}, the reader can verify how the contour plot of the soliton solution and the output of the algorithm coincide. Recall the form of the matrix $B_A$ in \cref{dual_soliton_TL}:
    \begin{equation}
        B_A=
        \begin{pmatrix}
            1 &0 &-1 &-1 &0 &0\\
            0 &1 &1 &1 &0 &0\\
            0 &0 &0 &0 &1 &1
        \end{pmatrix}.
    \end{equation}
    In \cref{fig:dual_DS}, we instead compare the contour plot of the soliton solution $Q(x,y,20)$ and the output of the algorithm applied to the dual matrix $B_A$.
    
    \begin{figure}
        \centering

        \includegraphics[width=0.3\linewidth]{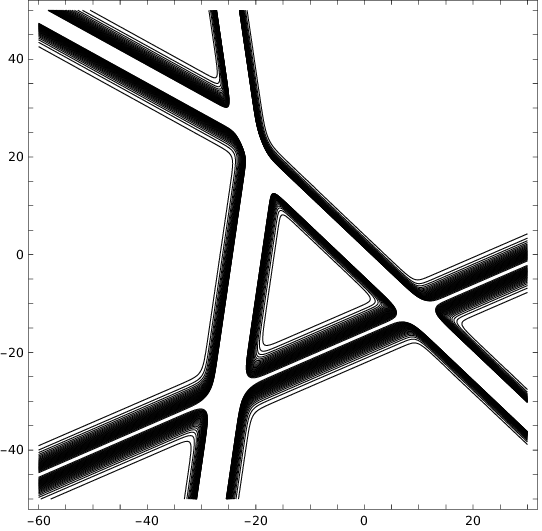}
        \tikzset{every picture/.style={line width=0.75pt}} 
        
        \begin{tikzpicture}[x=0.75pt,y=0.75pt,yscale=-0.75,xscale=0.75]
            
            \draw   (400,526) -- (450,526) -- (450,576) -- (400,576) -- cycle ;
            \draw   (450,526) -- (500,526) -- (500,576) -- (450,576) -- cycle ;
            \draw   (450,476) -- (500,476) -- (500,526) -- (450,526) -- cycle ;
            \draw   (400,476) -- (450,476) -- (450,526) -- (400,526) -- cycle ;
            \draw   (400,576) -- (450,576) -- (450,626) -- (400,626) -- cycle ;
            \draw   (500,476) -- (550,476) -- (550,526) -- (500,526) -- cycle ;
            \draw [color={rgb, 255:red, 211; green, 47; blue, 47 }  ,draw opacity=1 ]   (427,456) .. controls (429,540) and (422,580) .. (441,597) ;
            \draw [color={rgb, 255:red, 211; green, 47; blue, 47 }  ,draw opacity=1 ]   (384,598) .. controls (417,596) and (420,601) .. (426,609) ;
            \draw [color={rgb, 255:red, 211; green, 47; blue, 47 }  ,draw opacity=1 ]   (426,609) -- (441,597) ;
            \draw [color={rgb, 255:red, 101; green, 87; blue, 210 }  ,draw opacity=1 ]   (383,546) .. controls (432,545) and (449,542) .. (470,557) ;
            \draw [color={rgb, 255:red, 101; green, 87; blue, 210 }  ,draw opacity=1 ]   (470,557) -- (487,544) ;
            \draw [color={rgb, 255:red, 0; green, 103; blue, 88 }  ,draw opacity=1 ]   (523,460) .. controls (524,478) and (519,491) .. (536,493) ;
            \draw [color={rgb, 255:red, 0; green, 103; blue, 88 }  ,draw opacity=1 ]   (517,509) -- (536,493) ;
            \draw [color={rgb, 255:red, 128; green, 128; blue, 128 }  ,draw opacity=1 ]   (467,512) -- (476,496) ;
            \draw   (471.5,496) .. controls (471.5,493.51) and (473.51,491.5) .. (476,491.5) .. controls (478.49,491.5) and (480.5,493.51) .. (480.5,496) .. controls (480.5,498.49) and (478.49,500.5) .. (476,500.5) .. controls (473.51,500.5) and (471.5,498.49) .. (471.5,496) -- cycle ;
            \draw  [fill={rgb, 255:red, 0; green, 0; blue, 0 }  ,fill opacity=1 ] (462.5,512) .. controls (462.5,509.51) and (464.51,507.5) .. (467,507.5) .. controls (469.49,507.5) and (471.5,509.51) .. (471.5,512) .. controls (471.5,514.49) and (469.49,516.5) .. (467,516.5) .. controls (464.51,516.5) and (462.5,514.49) .. (462.5,512) -- cycle ;
            \draw [color={rgb, 255:red, 199; green, 80; blue, 0 }  ,draw opacity=1 ]   (386.33,492.67) .. controls (430.33,491.67) and (463.33,494.67) .. (467,512) ;
            \draw [color={rgb, 255:red, 101; green, 87; blue, 210 }  ,draw opacity=1 ]   (467,512) .. controls (468.33,539.67) and (465.33,538.67) .. (487,544) ;
            \draw [color={rgb, 255:red, 0; green, 103; blue, 88 }  ,draw opacity=1 ]   (476,496) .. controls (482.33,509.67) and (502.33,494.67) .. (517,509) ;
            \draw [color={rgb, 255:red, 21; green, 101; blue, 192 }  ,draw opacity=1 ]   (465.33,458.67) .. controls (465.33,482.67) and (466.33,484.67) .. (476,496) ;
                
            \draw (527,443) node [anchor=north west][inner sep=0.75pt]    {$[ 1,2]$};
            \draw (470,442) node [anchor=north west][inner sep=0.75pt]    {$[ 2,4]$};
            \draw (381,447) node [anchor=north west][inner sep=0.75pt]    {$[ 5,6]$};
            \draw (337,484) node [anchor=north west][inner sep=0.75pt]    {$[ 1,3]$};
            \draw (336,533) node [anchor=north west][inner sep=0.75pt]    {$[ 3,4]$};
            \draw (336,586) node [anchor=north west][inner sep=0.75pt]    {$[ 5,6]$};
            
        \end{tikzpicture}
        \includegraphics[width=0.3\linewidth]{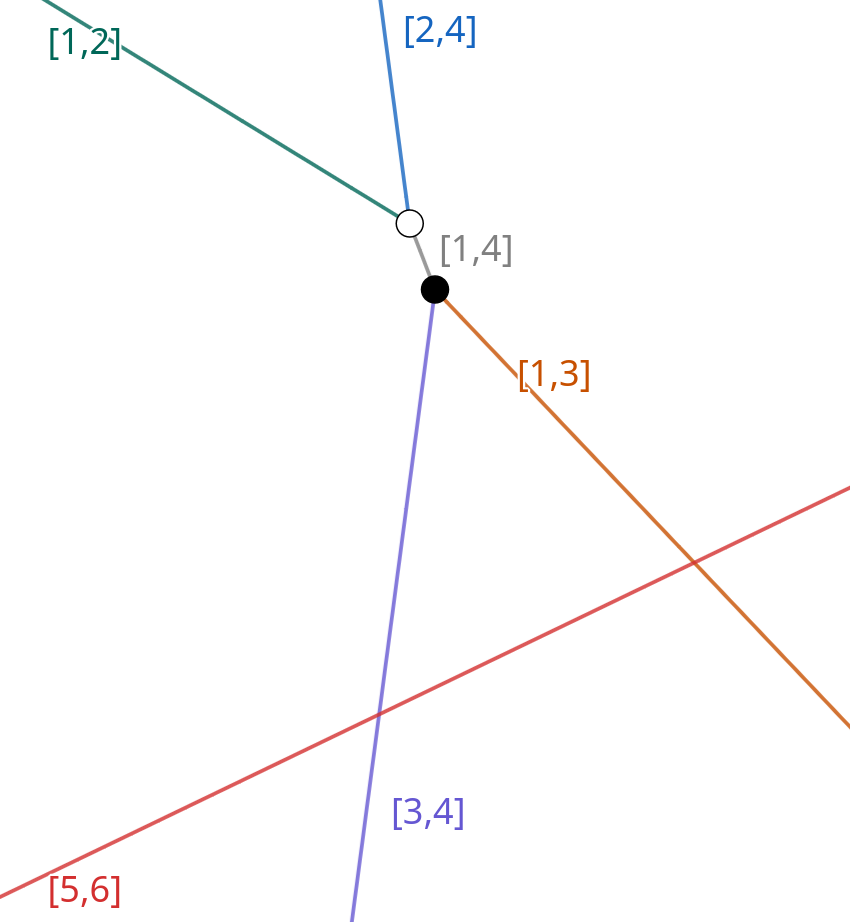}
        \caption{Comparison between the contour plot of $Q(x,y,-12)$ (top) and the output of the algorithm (bottom right).}
        \label{fig:algorithm-Davey}
    \end{figure}
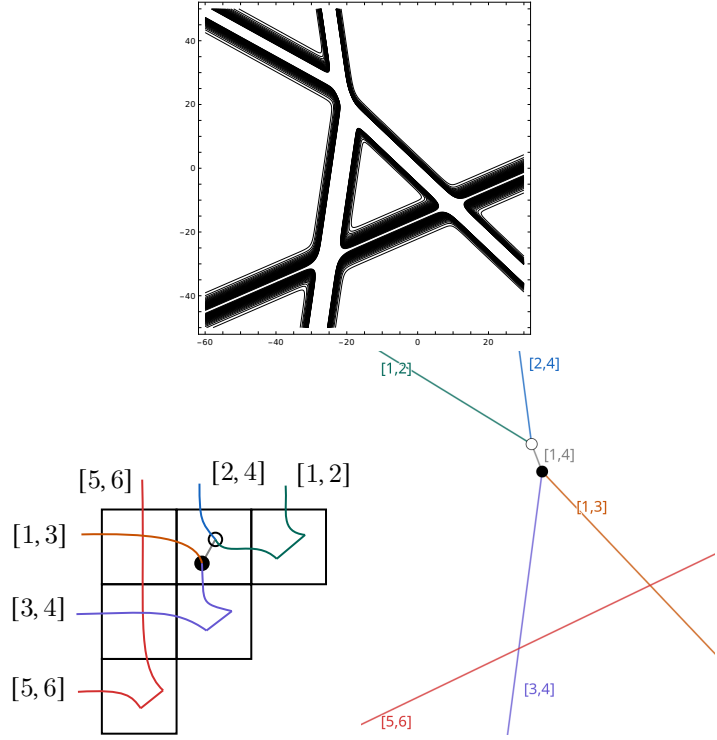
    \begin{figure}
        \centering
        \includegraphics[width=0.3\linewidth]{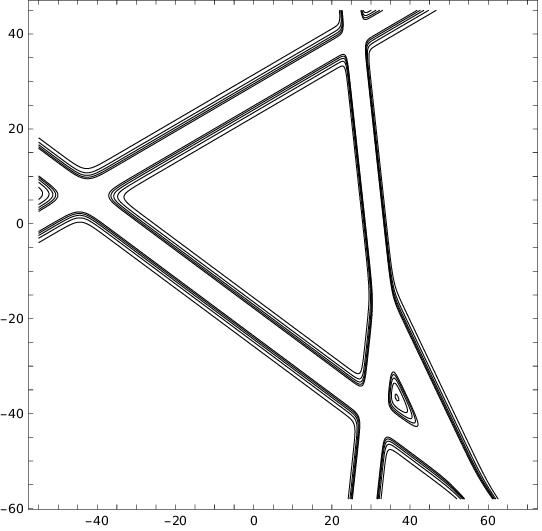}
        \tikzset{every picture/.style={line width=0.75pt}} 
        
        \begin{tikzpicture}[x=0.75pt,y=0.75pt,yscale=-0.75,xscale=0.75]

            \draw   (405,276.67) -- (446,276.67) -- (446,317.67) -- (405,317.67) -- cycle ;
            \draw   (446,276.67) -- (487,276.67) -- (487,317.67) -- (446,317.67) -- cycle ;
            \draw   (487,276.67) -- (528,276.67) -- (528,317.67) -- (487,317.67) -- cycle ;
            \draw   (405,317.67) -- (446,317.67) -- (446,358.67) -- (405,358.67) -- cycle ;
            \draw   (446,317.67) -- (487,317.67) -- (487,358.67) -- (446,358.67) -- cycle ;
            \draw   (487,317.67) -- (528,317.67) -- (528,358.67) -- (487,358.67) -- cycle ;
            \draw   (405,358.67) -- (446,358.67) -- (446,399.67) -- (405,399.67) -- cycle ;
            \draw [color={rgb, 255:red, 211; green, 47; blue, 47 }  ,draw opacity=1 ]   (419.33,255.33) .. controls (418.33,295.33) and (416,356.67) .. (435,373.67) ;
            \draw [color={rgb, 255:red, 211; green, 47; blue, 47 }  ,draw opacity=1 ]   (378,372.67) .. controls (411,370.67) and (414,375.67) .. (420,383.67) ;
            \draw [color={rgb, 255:red, 211; green, 47; blue, 47 }  ,draw opacity=1 ]   (420,383.67) -- (435,373.67) ;
            \draw [color={rgb, 255:red, 21; green, 101; blue, 192 }  ,draw opacity=1 ]   (385.33,332.33) .. controls (404.33,332.33) and (459.33,330.33) .. (462.33,343.33) ;
            \draw [color={rgb, 255:red, 21; green, 101; blue, 192 }  ,draw opacity=1 ]   (470.33,333.33) .. controls (470.33,332.33) and (462.33,343.33) .. (462.33,343.33) ;
            \draw   (465.83,333.33) .. controls (465.83,330.85) and (467.85,328.83) .. (470.33,328.83) .. controls (472.82,328.83) and (474.83,330.85) .. (474.83,333.33) .. controls (474.83,335.82) and (472.82,337.83) .. (470.33,337.83) .. controls (467.85,337.83) and (465.83,335.82) .. (465.83,333.33) -- cycle ;
            \draw [color={rgb, 255:red, 101; green, 87; blue, 210 }  ,draw opacity=1 ]   (460.33,258.33) .. controls (460.33,280.33) and (454.33,323.33) .. (470.33,333.33) ;
            \draw [color={rgb, 255:red, 199; green, 80; blue, 0 }  ,draw opacity=1 ]   (511.33,291.67) .. controls (503.33,282.33) and (504.33,272.33) .. (504.33,258.33) ;
            \draw [color={rgb, 255:red, 199; green, 80; blue, 0 }  ,draw opacity=1 ]   (502.33,304.33) .. controls (502.33,303.33) and (511.33,291.33) .. (511.33,291.67) ;
            \draw  [fill={rgb, 255:red, 0; green, 0; blue, 0 }  ,fill opacity=1 ] (497.83,304.33) .. controls (497.83,301.85) and (499.85,299.83) .. (502.33,299.83) .. controls (504.82,299.83) and (506.83,301.85) .. (506.83,304.33) .. controls (506.83,306.82) and (504.82,308.83) .. (502.33,308.83) .. controls (499.85,308.83) and (497.83,306.82) .. (497.83,304.33) -- cycle ;
            \draw [color={rgb, 255:red, 0; green, 103; blue, 88 }  ,draw opacity=1 ]   (386.33,294.33) .. controls (422.33,294.33) and (494.33,286.33) .. (502.33,304.33) ;
            \draw [color={rgb, 255:red, 128; green, 128; blue, 128 }  ,draw opacity=1 ]   (502.33,341.33) -- (513.33,331.33) ;
            \draw [color={rgb, 255:red, 128; green, 128; blue, 128 }  ,draw opacity=1 ]   (470.33,333.33) .. controls (496.33,330.33) and (494.33,342.33) .. (502.33,341.33) ;
            \draw [color={rgb, 255:red, 128; green, 128; blue, 128 }  ,draw opacity=1 ]   (502.33,304.33) .. controls (501.33,319.33) and (498.33,326.33) .. (513.33,331.33) ;
            
            \draw (345,279) node [anchor=north west][inner sep=0.75pt]    {$[ 1,2]$};
            \draw (340,318) node [anchor=north west][inner sep=0.75pt]    {$[ 2,4]$};
            \draw (377,243) node [anchor=north west][inner sep=0.75pt]    {$[ 5,6]$};
            \draw (507,245) node [anchor=north west][inner sep=0.75pt]    {$[ 1,3]$};
            \draw (442,234) node [anchor=north west][inner sep=0.75pt]    {$[ 3,4]$};
            \draw (352,379) node [anchor=north west][inner sep=0.75pt]    {$[ 5,6]$};

        \end{tikzpicture}
        \includegraphics[width=0.4\linewidth]{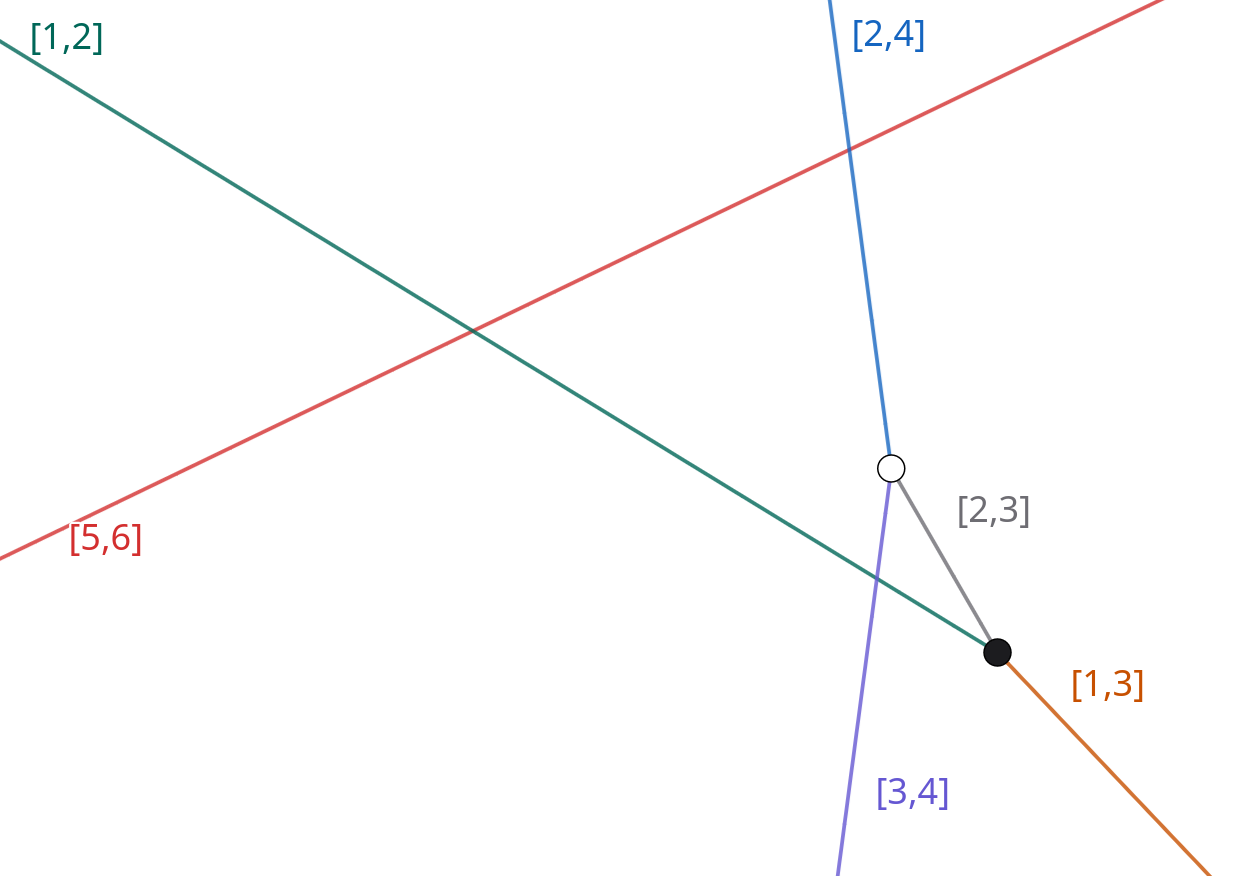}
        \caption{Comparison between the contour plot of $Q(x,y,20)$ (top) and the output of the algorithm (bottom right).}
        \label{fig:dual_DS}
    \end{figure}
\end{example}

\section{Inverse problem}
\label{sec-inverseproblem}
Let $n \in \mathbb{Z}$ be fixed. In the Toda lattice system, the goal of the inverse problem is to recover the point $A \in \mathrm{Gr}(N,M)_{>0}$ from the contour plot\footnote{Throughout this section, we always consider Toda lattice soliton plots in the $(x_1, x_{-1})$-plane} $\mathcal{C}(V_{n,A})$ and the fixed $n$. We are also implicitly assuming knowledge of the generic constants fed into $\mathcal{C}(V_{n,A})$. The inverse problem can be analogously defined in the DSII case, where the task is now to recover $A\in \mathrm{Gr}(N,M)_{>0}$ from $\mathcal{C}(Q_{t,A})$ with fixed $t$.

In \cite{kodama2014kp}, Kodama and Williams solve the inverse problem for contour plots of KP soliton solutions. The authors relied on the theory of cluster algebras, and used general arguments so that their methods could be translated to other integrable systems. In this section, we will show that Kodama and Williams' methods allow us to solve the inverse problem for the 2D Toda Lattice and Davey--Stewartson systems. In what follows, we always assume  that our contour plots depend on generic parameters (see \cref{def-generic} and \cref{DS-genericity}).

A keystone in Kodama and Williams' argument is that soliton graphs arising from the KP system satisfy the \textit{resonance property}, which we define in the following proposition:
\begin{proposition}
    Let $A\in \mathrm{Gr}(M,N)_{\geq 0}$. All vertices $v$  appearing in 
     $\mathcal{C}(V_{n,A})$ for $|n| \gg 0$ (resp. in $\mathcal{C}(Q,t)$ for $|t| \gg 0$) satisfy the \textit{resonance property}. That is, there exist $i<j<l$ such that the soliton lines incident on $v$ appear as $[i,j]$, $[j,l]$ and $[i,l]$, in clockwise (resp. counterclockwise) order.
\end{proposition}
\begin{proof}
    By large enough $|n|$, we mean that $\mathcal{C}(V_{n,A})$ has the same topology as the corresponding asymptotic contour plot; large enough $|t|$ is meant to be understood similarly. Therefore, the labels of adjacent dominant regions share $N-1$ indices. We can then assume that the soliton lines intersecting at $v$ are of the form $[i,j]$, $[j,l]$ and $[i,l]$. If furthermore $i<j<l$, then the result follows from the following ordering of the slopes of the soliton lines:
    \begin{equation}
        \lambda_i \lambda_j < \lambda_i \lambda_l < \lambda_j \lambda_l.
    \end{equation}
    We have an analogous ordering of the slopes for the corresponding soliton lines in $\mathcal{C}(Q,t)$, since
    \begin{equation}
        \tan \frac{\psi_i+ \psi_j}{2} < \tan \frac{\psi_i+ \psi_l}{2} < \tan \frac{\psi_j+ \psi_l}{2}.
    \end{equation}
\end{proof}

\begin{definition}
    A plabic graph is an undirected graph embedded in a disk subject to the following requirements:
\begin{itemize}
    \item There are $M$ labeled vertices $1, \dots, M$ located in the boundary of the disk, arranged in either counterclockwise or clockwise\footnote{This is a matter of convention; if the outer edges are labeled in counterclockwise (resp. clockwise) order, internal edges should satisfy the resonance property in clockwise (resp. counterclockwise) order.} order.
    \item Each vertex on the boundary of the disk is incident to a single edge.
    \item Every internal vertex is colored either black or white.
\end{itemize}
\end{definition}

\begin{figure}[h]
    \centering
    \tikzset{every picture/.style={line width=0.75pt}} 

    \begin{tikzpicture}[x=0.75pt,y=0.75pt,yscale=-0.75,xscale=0.75]
    
        \draw   (226,2381.17) .. controls (226,2318.11) and (277.11,2267) .. (340.17,2267) .. controls (403.22,2267) and (454.33,2318.11) .. (454.33,2381.17) .. controls (454.33,2444.22) and (403.22,2495.33) .. (340.17,2495.33) .. controls (277.11,2495.33) and (226,2444.22) .. (226,2381.17) -- cycle ;
        \draw  [fill={rgb, 255:red, 0; green, 0; blue, 0 }  ,fill opacity=1 ] (320,2436) .. controls (320,2433.24) and (322.24,2431) .. (325,2431) .. controls (327.76,2431) and (330,2433.24) .. (330,2436) .. controls (330,2438.76) and (327.76,2441) .. (325,2441) .. controls (322.24,2441) and (320,2438.76) .. (320,2436) -- cycle ;
        \draw  [fill={rgb, 255:red, 0; green, 0; blue, 0 }  ,fill opacity=1 ] (340,2368) .. controls (340,2365.24) and (342.24,2363) .. (345,2363) .. controls (347.76,2363) and (350,2365.24) .. (350,2368) .. controls (350,2370.76) and (347.76,2373) .. (345,2373) .. controls (342.24,2373) and (340,2370.76) .. (340,2368) -- cycle ;
        \draw  [fill={rgb, 255:red, 0; green, 0; blue, 0 }  ,fill opacity=1 ] (396,2420) .. controls (396,2417.24) and (398.24,2415) .. (401,2415) .. controls (403.76,2415) and (406,2417.24) .. (406,2420) .. controls (406,2422.76) and (403.76,2425) .. (401,2425) .. controls (398.24,2425) and (396,2422.76) .. (396,2420) -- cycle ;
        \draw  [fill={rgb, 255:red, 0; green, 0; blue, 0 }  ,fill opacity=1 ] (294,2346) .. controls (294,2343.24) and (296.24,2341) .. (299,2341) .. controls (301.76,2341) and (304,2343.24) .. (304,2346) .. controls (304,2348.76) and (301.76,2351) .. (299,2351) .. controls (296.24,2351) and (294,2348.76) .. (294,2346) -- cycle ;
        \draw  [fill={rgb, 255:red, 0; green, 0; blue, 0 }  ,fill opacity=1 ] (393,2343) .. controls (393,2340.24) and (395.24,2338) .. (398,2338) .. controls (400.76,2338) and (403,2340.24) .. (403,2343) .. controls (403,2345.76) and (400.76,2348) .. (398,2348) .. controls (395.24,2348) and (393,2345.76) .. (393,2343) -- cycle ;
        \draw  [fill={rgb, 255:red, 255; green, 255; blue, 255 }  ,fill opacity=1 ] (272,2413) .. controls (272,2410.24) and (274.24,2408) .. (277,2408) .. controls (279.76,2408) and (282,2410.24) .. (282,2413) .. controls (282,2415.76) and (279.76,2418) .. (277,2418) .. controls (274.24,2418) and (272,2415.76) .. (272,2413) -- cycle ;
        \draw  [fill={rgb, 255:red, 255; green, 255; blue, 255 }  ,fill opacity=1 ] (335,2324) .. controls (335,2321.24) and (337.24,2319) .. (340,2319) .. controls (342.76,2319) and (345,2321.24) .. (345,2324) .. controls (345,2326.76) and (342.76,2329) .. (340,2329) .. controls (337.24,2329) and (335,2326.76) .. (335,2324) -- cycle ;
        \draw  [fill={rgb, 255:red, 255; green, 255; blue, 255 }  ,fill opacity=1 ] (300,2390) .. controls (300,2387.24) and (302.24,2385) .. (305,2385) .. controls (307.76,2385) and (310,2387.24) .. (310,2390) .. controls (310,2392.76) and (307.76,2395) .. (305,2395) .. controls (302.24,2395) and (300,2392.76) .. (300,2390) -- cycle ;
        \draw  [fill={rgb, 255:red, 255; green, 255; blue, 255 }  ,fill opacity=1 ] (356,2422) .. controls (356,2419.24) and (358.24,2417) .. (361,2417) .. controls (363.76,2417) and (366,2419.24) .. (366,2422) .. controls (366,2424.76) and (363.76,2427) .. (361,2427) .. controls (358.24,2427) and (356,2424.76) .. (356,2422) -- cycle ;
        \draw  [fill={rgb, 255:red, 255; green, 255; blue, 255 }  ,fill opacity=1 ] (383,2383) .. controls (383,2380.24) and (385.24,2378) .. (388,2378) .. controls (390.76,2378) and (393,2380.24) .. (393,2383) .. controls (393,2385.76) and (390.76,2388) .. (388,2388) .. controls (385.24,2388) and (383,2385.76) .. (383,2383) -- cycle ;
        \draw    (256.33,2303) -- (299,2346) ;
        \draw    (326.33,2267) -- (340,2319) ;
        \draw    (398,2343) -- (433.33,2314) ;
        \draw    (299,2346) -- (335,2324) ;
        \draw    (345,2368) -- (340,2329) ;
        \draw    (345,2324) -- (398,2343) ;
        \draw    (398,2343) -- (388,2378) ;
        \draw    (345,2368) -- (383,2383) ;
        \draw    (310,2390) -- (345,2368) ;
        \draw    (299,2346) -- (277,2408) ;
        \draw    (299,2346) -- (305,2385) ;
        \draw    (282,2413) -- (325,2436) ;
        \draw    (305,2395) -- (325,2436) ;
        \draw    (325,2436) -- (356,2422) ;
        \draw    (366,2422) -- (401,2420) ;
        \draw    (388,2388) -- (401,2420) ;
        \draw    (401,2420) -- (426.33,2456) ;
        \draw    (361,2427) -- (354.33,2495) ;
        \draw    (240.33,2435) -- (272,2413) ;
        
        \draw (241,2281) node [anchor=north west][inner sep=0.75pt]    {$1$};
        \draw (314,2244) node [anchor=north west][inner sep=0.75pt]    {$6$};
        \draw (435,2294) node [anchor=north west][inner sep=0.75pt]    {$5$};
        \draw (431.33,2458) node [anchor=north west][inner sep=0.75pt]    {$4$};
        \draw (348,2498) node [anchor=north west][inner sep=0.75pt]    {$3$};
        \draw (221,2433) node [anchor=north west][inner sep=0.75pt]    {$2$};

    \end{tikzpicture}

    \caption{Example of plabic graph for $M=6$.}
    \label{fig:plabic_graph}
\end{figure}
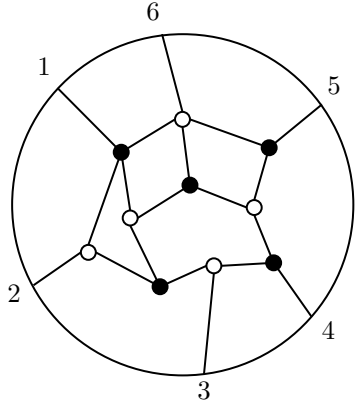

For a fixed $n$ with $|n| \gg 0$ and $A \in \mathrm{Gr}(N,M)_{\geq 0}$, the contour plot $\mathcal{C}(V_{n,A})$ defines a \textit{generalized} plabic graph $G_{n,A}$ in the following way:
\begin{itemize}
    \item Embed $\mathcal{C}(V_{n,A})$ in a disk;
    \item Label the boundary vertex defined by the soliton line $\{i,\pi(i)\}$ by $\pi(i)$.
\end{itemize}
An analogous statement is true for $\mathcal{C}(Q_{t,A})$ with $|t| \gg 0$. In a generalized plabic graph, edge vertices are not necessarily ordered. We also allow for internal edges to intersect in X-crossings, which we do not consider vertices.\\
Kodama and Williams showed that if all the vertices of a plabic graph satisfy the resonance property, then it is \textit{reduced}, in the sense of \cite{postnikov2006total}. 

\begin{theorem}[\cite{kodama2014kp}] \label{theo-cluster} Let $A \in \mathrm{Gr}(N,M)_{>0}$. If $\mathcal{C}(V_{n,A})$ for $|n| \gg 0$ has no X crossings, then the set of Pl\"ucker coordinates labeling a generic soliton graph is a cluster for the cluster algebra associated to the Grassmannian. The same is true for $\mathcal{C}(Q_{t,A})$ with $|t| \gg 0$.
\end{theorem}
\begin{proof}
    If $A \in \mathrm{Gr}(N,M)_{>0}$ then, per \cref{asymptotic}, the boundary edges of $G_{n,A}$ will be arranged in counterclockwise order. Furthermore, if $\mathcal{C}(V_{n,A})$ has no X-crossings then $G_{n,A}$ will be a veritable plabic graph. As the vertices of $G_{n,A}$ satisfy the resonance property, it will be a reduced plabic graph (cf. \cite[Theorem 10.5]{kodama2014kp}). The result follows, as it is known that the Pl\"ucker coordinates corresponding to the labels of a reduced plabic graph are a cluster for the coordinate ring of $\mathrm{Gr}(N,M)$ (see \cite{Scott2006}).
    
    On the other hand, per \cref{DS-asymptotic}, the boundary edges of $G_{t,A}$ will be ordered in clockwise order, so that it will be a reduced plabic graph in the opposite convention if no X-crossings are present.
\end{proof}

In \cref{x-cross}, we show that whenever $A\in \mathrm{Gr}(N,M)_{>0}$, then $\mathcal{C}(V_{n,A})$ and $\mathcal{C}(Q_{t,A})$ have no X-crossings for $|n| \gg 0$ and $|t| \gg 0$, respectively (see \cref{prop-xcross}). Kodama and Williams' methods then allow us to conclude the following:
\begin{theorem}\label{inverse-prob}
    Let $A \in \mathrm{Gr}(N,M)_{>0}
    $ and $|n| \gg 0$ be fixed. Then given a generic contour plot $\mathcal{C}(V_{n,A})$ and $n$ we are able to recover the point $A$. The same is true when given a generic $\mathcal{C}(Q_{t,A})$ and $|t| \gg 0$.
\end{theorem}
\begin{proof}
    Starting from $\mathcal{C}(V_{n,A})$ we may compute the label of the dominant exponentials in each region (cf. \cref{label-contour}). Furthermore, we are able to compute the Pl\"ucker coordinate associated to each region, up to a scalar factor. Indeed, notice that, since for $|n| \gg 0$ the labels of adjacent dominant regions share $N-1$ indices, a soliton line in $\mathcal{C}(V_{n,A})$ has equation
    \begin{equation}
        x_1(\lambda_i-\lambda_j) + x_{-1}\left(\frac{1}{\lambda_i}-\frac{1}{\lambda_j} \right)+n(\phi_i-\phi_j)+\frac{N-1}{2}(\Lambda_I-\Lambda_J)+\ln\left(\frac{\Delta_I(A)Sh_I(\lambda)}{\Delta_J(A)Sh_J(\lambda)} \right)=0
    \end{equation}
    for two multi-indices $I,J \in \binom{[M]}{N}$. It follows that we can recover the ratios of the pertinent Pl\"ucker coordinates if we know the generic parameters. The same is true for $\mathcal{C}(Q_{t,A})$ with $|t| \gg 0$, as soliton lines now have equation
    \begin{equation}
        x(\sin \psi_i - \sin \psi_j)+y (\cos \psi_i - \cos \psi_j)-t(\sin2\psi_i-\sin 2\psi_j)+\ln \left( \frac{\Delta_I(A)S_I(\psi)}{\Delta_J(A)S_J(A)}\right)=0.
    \end{equation}
    By \cref{theo-cluster}, Pl\"ucker coordinates corresponding to dominant exponentials form a cluster $\mathbf{c}$ for the cluster algebra structure of the coordinate ring of the Grassmannian. This concludes the proof, as it implies we can express each Pl\"ucker coordinate of $A$ as a Laurent polynomial in $\mathbf{c}$.
\end{proof}
We have phrased the previous proof in the language of cluster algebras to match with the work of Kodama and Williams. Alternatively, the proof could also be phrased in the language of \textit{matroids}. More precisely, since the contour plots corresponding to matrices $A \in \mathrm{Gr}(N,M)_{>0}$ have no X-crossings, the labels of each region match the labels of the associated $\reflectbox{L}$-diagram (see \cref{fig:region-labels}). This ends the proof, as the Pl\"ucker coordinates corresponding to the labels of the  $\reflectbox{L}$-diagram are a basis  for the matroid $\mathcal{M}(A)$ (cf. \cite[Corollary 4.2]{TALASKA201158}).

\section{Conclusion and Outlook}

Soliton solutions to the KP equation are parametrized by Grassmannian points. In seminal work, Kodama and Williams \cite{kodama2014kp} showed that the (asymptotic) contour plots of soliton solutions can be algorithmically recovered from the Deodhar component of the Grassmannian parameter. In this paper we have shown that this algorithm for KP-solitons can be extended to two more models, the 2DTL and the DSII (cf. \cref{algorithm_DS} and \cref{algorithm-TL1}). Whether this machinery is useful for other models, and how it would have to be adapted, is an open question. A notable difference between the KP case and the integrable systems we considered is that our asymptotic contour plots admit the appearance of parallel soliton lines (cf. \cref{parallel-lines}).

Our work on the asymptotics of 2DTL and DSII solitons agrees with that of \cite{Biondini_2010} and \cite{biondini2022soliton}, respectively. We emphasize that results such as \cref{fan_is_der} and \cref{DS-asymptotic}, which tie these asymptotics to Grassmannian combinatorics, are new.

As postulated by Kodama and Williams, the \textit{inverse problem} for soliton solutions involves recovering the Grassmannian parameter starting from its associated contour plot. In this paper, we solved the inverse problem for parameters in the totally positive Grassmannian, for the 2DTL as well as for the DSII (cf. \cref{inverse-prob}). Extending this to the totally non-negative setting remains an open problem, but should be accessible.

\printbibliography

\appendix
\section{Proof of Proposition \ref{asymptotic}}\label{app-asymptotic}

The proof of Proposition \ref{asymptotic} is entirely analogous to the proof of Theorem 6.1 in \cite[\S 6.2]{kodama}; we include it for the sake of completeness. We mean to prove that:

\begin{proposition}
        Let $A \in \mathcal{P}_{v,w}^{>0} \subset \mathrm{Gr}\,(N,M)_{\geq0}$ be an irreducible matrix, with pivot set $\{i_1, \dots, i_N\}$ and non-pivot set $\{j_1, \dots, j_{M-N}\}$. Assume also that the parameters $\{\phi_1, \dots , \phi_M\}$ of the TL soliton $V_{n,A}(x_{-1},x_1)$ are generic. Then $\mathcal{C}(V_{n,A})$ has the following asymptotic structure: 
        \begin{enumerate}
            \item For $x_{-1}\gg 0$, there exist $N$ line solitons of type $[i_n, p_n]$, where $p_n>i_n$.\\
            \item For $x_{-1}\ll 0$, there exist $M-N$ line solitons of type $[q_m, j_m]$, where $q_m < jm$.
        \end{enumerate}
\end{proposition}

\begin{proof}
    First we will consider the case $x_{-1}\gg 0$. Completing the proof amounts to verifying each of the following statements: 
    \begin{enumerate}
        \item \textit{If an $[i,j]$-soliton exists, then $i$ is a pivot of $A$:} Suppose that $\Theta_I$ and $\Theta_J$ are two adjacent dominant phases, so that we can write $I=\{i, m_2, \dots , m_M\}$ and $J=\{j, m_2, \dots m_M\}$. By contradiction, assume that the $i$-th column $A_i$ of $A$ is not a pivot column. Then necessarily
        \begin{equation}
            A_i=\sum_{r=1}^{s} c_r A_{i_r}
        \end{equation}
        for $c_r \neq 0$ and some $i_r < i$. But notice now that along the line $L_{i,j}$, in accordance with Lemma \ref{lem-dominantrel}, we have
        \begin{equation}
            \overline{\theta}_{i_r}(x_{-1}, x_1,n) > \overline{\theta}_i(x_{-1},x_1,n).
        \end{equation}
        This would imply $\Theta_I < \Theta_{\{i_r, m_2, \dots m_M\}}$ along $L_{i,j}$, contradicting that $\Theta_I$ is the dominant phase.
        \item \textit{There exists an $[i,j]$-soliton for every pivot $i$ of A:} Let $i$ be a pivot of $A$. Let us define an index $j^* \in \{i+1, \dots, M\}$ as the maximum of the $j>i$ satisfying that\footnote{Such a maximum does indeed exist. Since $A$ is irreducible, there is another non-zero entry on the row of the pivot $i$. Also, since $i$ is a pivot index, $\mathrm{rank}\, A[1, \dots, i-1] <N$.}
        \begin{equation}
            \begin{cases}
                \mathrm{rank}\, A[1, \dots, i-1,j,j+1, \dots,M]=N,\\
                \mathrm{rank}\, A[1, \dots, i-1,j+1, \dots,M]=N-1,
            \end{cases}
        \end{equation}
        where $A[k_1, k_2, \dots, k_n]$ is the submatrix of $A$ comprised of the columns with indices $k_1, k_2, \dots, k_n$.
        Then by assumption, $\mathrm{rank}\, A[1, \dots, i-1,i,j^*+1, \dots M]$ equals either $N$ or $N-1$. We consider each of these two cases separately.
        \begin{itemize}
            \item Suppose that $\mathrm{rank}\, A[1, \dots, i-1,i,j^*+1, \dots M]=N$. We have chosen $j^*$ such that $\mathrm{rank}\, A[1, \dots,i-1,j^*+1, \dots, M]=N-1$, so it follows we may choose $N-1$ linearly independent columns $j_1, \dots j_{N-1} \in \{1,  \dots, i-1,j^*+1, \dots, M\}$ such that $\mathrm{rank}\, A[j_1, \dots, j_{N-1}]=N-1$. We understand the behavior of the functions $\theta_k$ along the line $L_{i,j^*}$; namely,
            \begin{equation}
                \overline{\theta}_1 > \overline{\theta}_2 > \dots >\overline{\theta}_{i-1} \hspace{1 em} \text{and} \hspace{1 em} \overline{\theta}_{j^*+1}<\overline{\theta}_{j^*+2}< \dots <\overline{\theta}_M.
            \end{equation}
            This means we may choose these indices such that
            \begin{equation}
                \sum_{l=1}^{N-1} \overline{\theta}_{j_l} \geq\sum_{l=1}^{N-1}\overline{\theta}_{k_l}
            \end{equation}
            for all $\{k_1, \dots, k_{N-1}\} \subset \{i, \dots, i-1, j^*+1, \dots, M \}$. If we define $J=\{j^*, j_1, \dots, j_{N-1}\}$ and $I=\{i, j_1, \dots, j_n\} $, then $\Theta_I$ and $\Theta_J$ are dominant phases labeling regions separated by the soliton line $L_{i,j^*}$. Indeed, since $i$ is a pivot index and the columns we have chosen are linearly independent, $\Delta_I(A) \neq0.$ Similarly, $\Delta_J(A) \neq 0$ since we assumed $\mathrm{rank}\, A[1,\dots, i-1,j^*,j^*+1, \dots, M]=N$. If there were a multi-index $J'=\{p, j_1, \dots, j_{N-1}\}$ such that $\Delta_{J'}(A) \neq 0$ and $\Theta_{J'} > \Theta_J$ along $L_{i,j^*}$, this would contradict our choice of $j^*$. Because this would mean $\overline{\theta}_p > \overline{\theta}_{j^*}$, Lemma \ref{lem-dominantrel} implies that $p >j^*$, but this contradicts the maximality of $j^*$.
            \item Suppose that $\mathrm{rank}\, A[1, \dots , i-1, i, j^*+1, \dots M]=N-1$. Since $i$ is a pivot index, the column $A[i]$ cannot be spanned by its preceding columns alone. Let $r^*$ be the maximum of the indices $r$ satisfying
            \begin{equation}
                \begin{cases}
                    A[i] \in \mathrm{span}\, \{A[1], \dots, A[i-1], A[r], \dots A[M]\},\\
                    A[i] \notin \mathrm{span}\, \{A[1], \dots, A[i-1], A[r+1], \dots A[M]\}.
                \end{cases}
            \end{equation}
            Analogously to before, we may choose indices $j_1, \dots, j_{N}-1$ such that, if we define $I= \{ i, j_1, \dots, j_{N-1}\}$, then $\Delta_I(A) \neq0$ and $\Theta_I$ is a dominant phase. With this choice of indices, $J=\{ r^*, j_1, \dots, j_{N-1}\}$ labels another dominant phase. Indeed, by assumption $\Delta_J(A) \neq 0$. If $K=\{ k, j_1, \dots, j_{N-1}\}$ were such that $\Theta_K > \Theta_{J}$, then our choice of $j's$ and Lemma \ref{lem-dominantrel} would imply that $r^* < k$. Then necessarily $\Delta_K(A)=0$, by the maximality of $r^*$. Therefore, as in the previous case, $\Theta_I$ and $\Theta_J$ are dominant phases, and their corresponding regions on the contour plot are separated by the soliton line $L_{i,r^*}$. 
        \end{itemize}
        \item \textit{Uniqueness of the $[i, p]$-soliton}: By contradiction, suppose that there are two indices $p,\, p'$ such that $[i,p]$ and $[i, p']$ -solitons are present in the contour plot. Assume $p > p'$, so that the $[i,p]$ -soliton is left of the $[i, p']$-soliton. The dominant phase of the region $x_1 \ll 0$ is indexed by the lexicographically minimal element in $\mathcal{M}(A)$. We know this phase must be $\Theta_I$, where $I=\{i_1, \dots, i_N\}$ is the set of pivot indices of $A$. Since the pivot index $i$ is replaced by the $[i,p]$-soliton, there must exist an $[i', i]$- soliton with a pivot $i' < i$ between the $[i,p]$ and $[i, p']$ solitons. However, the inequality $i'< i < p'$ implies
        \begin{equation}
            \lambda_i \lambda_{i'} < \lambda_i \lambda_{p'}.
        \end{equation}
        This is a contradiction, as it places the $[i', i]$-soliton to the right of the $[i, p']$-soliton instead of to its left.
    \end{enumerate}
    
    We have proved the first statement of Theorem \ref{asymptotic}. That is, we have shown that if $A \in \text{Gr}(N, M)$, then, in the $x_{-1} \gg 0$ limit, the contour plot of $\tau_A$ consists of $N$ solitons corresponding to the pivot indices of $A$. To prove the matching statement for the limit $x_{-1} \ll 0$, we can rely on duality.

    Keeping with the notation of Section \ref{sec-dual}, we write $B_A$ for the dual matrix of $A$. Recall that $B_A \in \text{Gr}(M-N, M)$. The pivot indices of $B_A$ are $\pi(J_0) = \{\pi(j_1), \dots , \pi(j_{M-N}) \}$, where $J_0$ are the non-pivot indices of $A$. By the first leg of the proof, we understand the arrangement of the contour plot of $\tau_{B_A}$ in the limit $x_{-1} \gg 0$. But according to Proposition \ref{prop-dual}, the soliton fan of $\tau_{B_A}^{(n)}(x_{-1},x_1)$ matches the soliton fan of $\tau_A^{(-n)}(-x_{-1},-x_1)$ in the limit $|x_{-1}| \gg 0$.

    It follows that in the limit $x_{-1}\ll 0$, the contour plot of $\tau_A^{(n))}(x_{-1},x_1)$ can be produced by reflecting with respect to the origin the contour plot of $\tau_{B_A}^{(n)}(x_{-1},x_1)$ in the limit $x_{-1}\gg  0$.
\end{proof}

\section{Proof of Proposition \ref{fan_is_der}}\label{app-fan}

\begin{proposition}
    Let $A \in \mathcal{P}_{v,w}^{>0}$ be an irreducible matrix. The permutation $\pi \in S_M$ defined by the soliton fan of $\tau_A$,
    \begin{equation}
    \begin{cases}
        \pi(i_k)=p_k & k=1, \dots, N\\
        \pi(j_k)=q_k & k=1, \dots, M-N,
    \end{cases}
    \end{equation}
   is given by $\pi = vw^{-1}$.
\end{proposition}
\begin{proof}
   Recall that the lexicographically minimal nonzero Pl\"ucker coordinate $\Delta_I(A)$ is given by the set of pivot indices of A, $I_1=\{i_1, \dots, i_N\}$. Since $A \in \mathcal{P}_{v,w}^{>0}$,  the lexicographically maximal Pl\"ucker coordinate $\Delta_{I'}(A)$ is given by $I'=vw^{-1}(I)$. In other words, the last nonzero entry in the $k$-th row of $A$ is given by $i_k'=vw^{-1}(i_k)$. The matrix $A$ is irreducible, so $i_k < i_k'$. 

Since it is the lexicographical minimum in $\mathcal{M}(A)$, notice how the set of pivot indices $I_1$ corresponds to the dominant phase $\Theta_{I_1}$ in the region $x_1 \ll 0$. Without loss of generality, we can assume the slopes of the asymptotic soliton solutions are ordered as
    \begin{equation}
            \lambda_{i_1}\lambda_{p_1} < \dots <\lambda_{i_N} \lambda_{p_N}.
    \end{equation}
This implies the leftmost soliton is of type $[i_N, p_N]$, and that the dominant phase $\Theta_{I_2}$ adjacent to $\Theta_{I_1}$ on the right is given by $I_2=\{i_1, \dots, i_{N-1}, p_N\}$. The key step is noticing that $I_2$ can also be characterized as the lexicographical maximum of the set 
    \begin{equation}
            \mathcal{M}(A) \cap \{ i_1, \dots i_{N-1}, q \}\, | \,  i_N < q \leq M \}.
    \end{equation}
Indeed, suppose by contradiction that there exists $p$ such that $p_N < p \leq M$ and $\Delta_{\{i_1, \dots, i_{N-1}, p\}}(A) \neq 0$. Then Lemma \ref{lem-dominantrel}, implies that along the line $L_{i_N, p_N}$ we have $\overline{\theta}_{p_N} < \overline{\theta}_{p}$, so that also
    \begin{equation}
            \Theta_{I_2}= \sum_{k=1}^{N-1} \overline{\theta}_{i_k}+\overline{\theta}_{p_N} <\sum_{k=1}^{N-1} \overline{\theta}_{i_k} + \overline{\theta}_p = \Theta_{\{i_1, \dots i_{N-1}, p\}}, 
    \end{equation}
which contradicts the assumption that $I_2$ is a dominant phase. This new characterization of $I_2$ allows us to conclude $p_N= vw^{-1}(i_N)$ must be labeling the last nonzero entry on the $N$-th row of $A$. 
        
Now the proof can proceed inductively. We assume that $I_{k-1}=\{i_1, \dots , i_k,p_{k+1},\dots , p_N \}$ and\\ 
$I_k = \{i_1, \dots , i_{k-1}, p_{k}, \dots , p_N \}$ label adjacent dominant regions $\Theta_{I_{k-1}}$ and $\Theta_{I_k}$. Lemma \ref{lem-dominantrel} still implies that that $I_k$ is the lexicographically maximal element of the set
    \begin{equation}
             \mathcal{M}(A) \cap \{ i_1, \dots, i_{k-1}, q, p_{k+1}, \dots p_N\}\, | \,  i_N < q \leq M \},
    \end{equation}
and consequently that $p_k= vw^{-1}(i_k)$. This proves $\pi(i_l)= vw^{-1}(i_l)$ for all the pivot indices of $A$. The proof for the non-pivot indices concerns the asymptotic region $x_{-1} \ll 0$, and is analogous.
\end{proof}

\section{X-crossings in asymptotic contour plots}\label{x-cross}

Our main goal is to prove the following result:
\begin{proposition}\label{prop-xcross}
    Let $A \in \mathrm{Gr}(N,M)_{>0}$. Then generic contour plots $\mathcal{C}(V_{n,A})$ and $\mathcal{C}(Q_{t,A})$ have no X-crossings for $|n| \gg 0$ and $|t| \gg 0$, respectively.
\end{proposition}
We have relegated the proof to an appendix because the methods used are exactly those appearing in \cite[\S 9]{kodama2014kp}. \cref{prop-xcross} is a corollary of the following proposition, which states the presence of an X-crossing indicates the vanishing of a Pl\"ucker coordinate.
\begin{proposition}\label{x-vanish}
    Let $A \in \mathrm{Gr}(N,M)_{\geq 0}$. Suppose we have an X-crossing involving indices $1 \leq h <i<j<l\leq M$ in $\mathcal{C}(V_{n,A})$ for $|n| \gg 0$ (resp. in $\mathcal{C}(Q_{t,A})$ for $|t| \gg 0).$ Then there exists $S \in \binom{[M]}{N-2}$ disjoint from $\{h,i,j,l\}$ such that we land in one of the following cases:
    \begin{enumerate}
        \item The X-crossing involves the soliton lines $[i,j]$ and $[h,l]$:
        \begin{enumerate}
            \item If $\lambda_i\lambda_j > \lambda_h \lambda_l$, (resp. $\tan \frac{\psi_i +\psi_j}{2}< \tan \frac{\psi_h+ \psi_l}{2})$ then $S \cup \{h,l\} \notin \mathcal{M}(A)$ for $n > 0$ (resp. $t>0)$ and $S \cup \{i,j\} \notin \mathcal{M}(A)$ for $n<0$ (resp. $t<0)$.
            \item If $\lambda_i\lambda_j < \lambda_h \lambda_l$, (resp. $\tan \frac{\psi_i +\psi_j}{2}> \tan \frac{\psi_h+ \psi_l}{2})$ then $S \cup \{i,j\} \notin \mathcal{M}(A)$ for $n > 0$ (resp. $t>0)$ and $S \cup \{h,l\} \notin \mathcal{M}(A)$ for $n<0$ (resp. $t<0)$.
        \end{enumerate}
        \item The X-crossing involves the soliton lines $[h,i]$ and $[j,l]$: If $n >0$ (resp. $t<0$) $S\cup \{j,l\} \notin \mathcal{M}(A)$. If $n<0$ (resp. $t>0$), then $S \cup \{h,i\} \notin \mathcal{M}(A)$.
        \item The X-crossing involves the soliton lines $[h,j]$ and $[i,l]$. If $n>0$ (resp. $t<0$), then $S \cup \{h,j\} \notin \mathcal{M}(A)$. If $n<0$ (resp. $t>0)$, then $S \cup \{i,l\} \notin \mathcal{M}(A)$.
    \end{enumerate}
\end{proposition}

\begin{figure}[H]
    \centering
\tikzset{every picture/.style={line width=0.75pt}} 

\begin{tikzpicture}[x=0.75pt,y=0.75pt,yscale=-1,xscale=1]

\draw    (120.32,50.04) -- (210.95,241.42) ;
\draw    (109.81,238.99) -- (218.88,52.97) ;
\draw    (80.14,85.03) -- (321.07,162.09) ;
\draw    (76.4,188.93) -- (325,132.8) ;
\draw    (146.17,49.12) -- (149.12,241.74) ;
\draw    (188.42,50.95) -- (166.8,245.4) ;
\draw  [fill={rgb, 255:red, 0; green, 0; blue, 0 }  ,fill opacity=1 ] (145.19,105.79) .. controls (145.19,105.03) and (145.84,104.42) .. (146.66,104.42) .. controls (147.47,104.42) and (148.13,105.03) .. (148.13,105.79) .. controls (148.13,106.55) and (147.47,107.16) .. (146.66,107.16) .. controls (145.84,107.16) and (145.19,106.55) .. (145.19,105.79) -- cycle ;
\draw  [fill={rgb, 255:red, 0; green, 0; blue, 0 }  ,fill opacity=1 ] (179.58,116.78) .. controls (179.58,116.02) and (180.24,115.4) .. (181.05,115.4) .. controls (181.86,115.4) and (182.52,116.02) .. (182.52,116.78) .. controls (182.52,117.53) and (181.86,118.15) .. (181.05,118.15) .. controls (180.24,118.15) and (179.58,117.53) .. (179.58,116.78) -- cycle ;
\draw  [fill={rgb, 255:red, 0; green, 0; blue, 0 }  ,fill opacity=1 ] (147.15,172.62) .. controls (147.15,171.86) and (147.81,171.25) .. (148.62,171.25) .. controls (149.44,171.25) and (150.1,171.86) .. (150.1,172.62) .. controls (150.1,173.38) and (149.44,173.99) .. (148.62,173.99) .. controls (147.81,173.99) and (147.15,173.38) .. (147.15,172.62) -- cycle ;
\draw  [fill={rgb, 255:red, 0; green, 0; blue, 0 }  ,fill opacity=1 ] (174.66,166.21) .. controls (174.66,165.45) and (175.32,164.84) .. (176.14,164.84) .. controls (176.95,164.84) and (177.61,165.45) .. (177.61,166.21) .. controls (177.61,166.97) and (176.95,167.58) .. (176.14,167.58) .. controls (175.32,167.58) and (174.66,166.97) .. (174.66,166.21) -- cycle ;

\draw (109.82,134.12) node [anchor=north west][inner sep=0.75pt]  [font=\tiny] [align=left] {$hijl$};
\draw (148.81,151.91) node [anchor=north west][inner sep=0.75pt]  [font=\tiny,rotate=-271.23] [align=left] {$hjil$};
\draw (156.82,118.73) node [anchor=north west][inner sep=0.75pt]  [font=\tiny] [align=left] {$hjli$};
\draw (153.85,158.95) node [anchor=north west][inner sep=0.75pt]  [font=\tiny] [align=left] {$jhil$};
\draw (177.72,136.68) node [anchor=north west][inner sep=0.75pt]  [font=\tiny,rotate=-89.98] [align=left] {$jhli$};
\draw (217.94,52.09) node [anchor=north west][inner sep=0.75pt]  [font=\scriptsize] [align=left] {$[h,j]$};
\draw (187.33,41.41) node [anchor=north west][inner sep=0.75pt]  [font=\scriptsize] [align=left] {$[h,l]$};
\draw (144.66,37.84) node [anchor=north west][inner sep=0.75pt]  [font=\scriptsize] [align=left] {$[i,j]$};
\draw (94.57,42.3) node [anchor=north west][inner sep=0.75pt]  [font=\scriptsize] [align=left] {$[i,l]$};
\draw (293.07,119.77) node [anchor=north west][inner sep=0.75pt]  [font=\scriptsize] [align=left] {$[h,i]$};
\draw (66.74,86.82) node [anchor=north west][inner sep=0.75pt]  [font=\scriptsize] [align=left] {$[j,l]$};

\end{tikzpicture}
\quad
\begin{tikzpicture}[x=0.75pt,y=0.75pt,yscale=-1,xscale=1]

\draw    (120.32,50.04) -- (210.95,241.42) ;
\draw    (109.81,238.99) -- (218.88,52.97) ;
\draw    (80.14,85.03) -- (321.07,162.09) ;
\draw    (76.4,188.93) -- (325,132.8) ;
\draw    (146.17,49.12) -- (149.12,241.74) ;
\draw    (188.42,50.95) -- (166.8,245.4) ;
\draw  [fill={rgb, 255:red, 0; green, 0; blue, 0 }  ,fill opacity=1 ] (145.19,105.79) .. controls (145.19,105.03) and (145.84,104.42) .. (146.66,104.42) .. controls (147.47,104.42) and (148.13,105.03) .. (148.13,105.79) .. controls (148.13,106.55) and (147.47,107.16) .. (146.66,107.16) .. controls (145.84,107.16) and (145.19,106.55) .. (145.19,105.79) -- cycle ;
\draw  [fill={rgb, 255:red, 0; green, 0; blue, 0 }  ,fill opacity=1 ] (179.58,116.78) .. controls (179.58,116.02) and (180.24,115.4) .. (181.05,115.4) .. controls (181.86,115.4) and (182.52,116.02) .. (182.52,116.78) .. controls (182.52,117.53) and (181.86,118.15) .. (181.05,118.15) .. controls (180.24,118.15) and (179.58,117.53) .. (179.58,116.78) -- cycle ;
\draw  [fill={rgb, 255:red, 0; green, 0; blue, 0 }  ,fill opacity=1 ] (147.15,172.62) .. controls (147.15,171.86) and (147.81,171.25) .. (148.62,171.25) .. controls (149.44,171.25) and (150.1,171.86) .. (150.1,172.62) .. controls (150.1,173.38) and (149.44,173.99) .. (148.62,173.99) .. controls (147.81,173.99) and (147.15,173.38) .. (147.15,172.62) -- cycle ;
\draw  [fill={rgb, 255:red, 0; green, 0; blue, 0 }  ,fill opacity=1 ] (174.66,166.21) .. controls (174.66,165.45) and (175.32,164.84) .. (176.14,164.84) .. controls (176.95,164.84) and (177.61,165.45) .. (177.61,166.21) .. controls (177.61,166.97) and (176.95,167.58) .. (176.14,167.58) .. controls (175.32,167.58) and (174.66,166.97) .. (174.66,166.21) -- cycle ;

\draw (109.82,134.12) node [anchor=north west][inner sep=0.75pt]  [font=\tiny] [align=left] {$hijl$};
\draw (148.81,151.91) node [anchor=north west][inner sep=0.75pt]  [font=\tiny,rotate=-271.23] [align=left] {$hjil$};
\draw (156.82,118.73) node [anchor=north west][inner sep=0.75pt]  [font=\tiny] [align=left] {$jhil$};
\draw (153.85,158.95) node [anchor=north west][inner sep=0.75pt]  [font=\tiny] [align=left] {$hjli$};
\draw (177.72,136.68) node [anchor=north west][inner sep=0.75pt]  [font=\tiny,rotate=-89.98] [align=left] {$jhli$};
\draw (217.94,52.09) node [anchor=north west][inner sep=0.75pt]  [font=\scriptsize] [align=left] {$[i,l]$};
\draw (187.33,41.41) node [anchor=north west][inner sep=0.75pt]  [font=\scriptsize] [align=left] {$[h,l]$};
\draw (144.66,37.84) node [anchor=north west][inner sep=0.75pt]  [font=\scriptsize] [align=left] {$[i,j]$};
\draw (94.57,42.3) node [anchor=north west][inner sep=0.75pt]  [font=\scriptsize] [align=left] {$[h,j]$};
\draw (293.07,119.77) node [anchor=north west][inner sep=0.75pt]  [font=\scriptsize] [align=left] {$[j,l]$};
\draw (66.74,86.82) node [anchor=north west][inner sep=0.75pt]  [font=\scriptsize] [align=left] {$[h,i]$};

\end{tikzpicture}

    \caption{To the left (resp. right), the neighborhood of the $(L_{[h,j]}$, $L_{[i,l]})$ X-crossing  for $n \gg0$ (resp. $t \ll 0).$}
    \label{fig:X-n>0,t<0}
\end{figure}
Notice how, in the first case, the instances $(a)$ and $(b)$ are indeed exhaustive, as our assumption prohibits the equality $\lambda_i \lambda_j=\lambda_h \lambda_l$ (or $\tan \frac{\psi_i + \psi_j}{2} = \tan \frac{\psi_h+\psi_l}{2}$). Indeed, since we are assuming that the $[i,j]$ and $[h,l]$-solitons form an X-crossing, these two lines cannot possibly be parallel. We will only explicitly prove the result for Case $(2)$. We need the following auxiliary lemma:
\begin{lemma}
    Let $A \in \mathrm{Gr}(N,M)_{\geq 0}$, and take $1 \leq h < i< j<l \leq M$. For $|n| \gg 0$ (resp. $|t| \gg 0)$, we have the following relationship among trivalent vertices in $\mathcal{C}(V_{n,A})$ (resp. $\mathcal{C}(Q_{t,A})$):
    \begin{itemize}
        \item If $n>0$ (resp. $t>0$), then $x_{-1[h,i,j]} < x_{-1,[h,i,l]}< x_{-1,[h,j,l]}<x_{-1,[i,j,l]}$ (resp. $y_{[h,i,j]}<y_{[h,i,l]} < y_{[h,j,l]}< y_{[i,j,l]}$).
        \item If $n<0$ (resp. $t<0$), then $x_{-1[h,i,j]} > x_{-1,[h,i,l]}> x_{-1,[h,j,l]}>x_{-1,[i,j,l]}$ (resp. $y_{[h,i,j]}>y_{[h,i,l]} > y_{[h,j,l]}> y_{[i,j,l]}$).
    \end{itemize}
\end{lemma}
\begin{proof}
    The proof follows from examining the following formulas:
    \begin{gather}
        x_{-1,[a,b,c]}-x_{-1,[a,b,d]}=\frac{\lambda_a \lambda_b }{\lambda_b-\lambda_a}\left((\lambda_d -\lambda_c)(\phi_a\lambda_a - \phi_b \lambda_b)+(\lambda_a - \lambda_b)(\phi_d \lambda_d-\phi_c \lambda_c) \right)n;\\
        y_{[a,b,c]}-y_{[a,b,d]}=(\sin \psi_c - \sin \psi_d+\sin (\psi_a+\psi_b+\psi_c)-\sin(\psi_a+\psi_b+\psi_d))t.
    \end{gather}
    We just emphasize that, because of our genericity assumptions on the parameters $\{\psi_1, \dots, \psi_M\}$, $\sin (\psi_a + \psi_b + \psi_c) < \sin (\psi_a + \psi_b + \psi_d) $ if $\psi_a+\psi_b+\psi_c < \psi_a+ \psi_b+\psi_d$.
\end{proof}

\begin{proof}[Proof (of \cref{x-vanish}, Case (2)).]
    In the notation of \cref{x-vanish}, suppose that soliton lines $L_{[h,j]}$ and $L_{[i,l]}$ intersect in an X-crossing $v$ on $\mathcal{C}(V_{n,A})$ (resp. $\mathcal{C}(Q_{t,A})$) for $n \gg 0$ (resp. $t \ll 0).$ Since $\lambda_h \lambda_i < \lambda_h \lambda_j$ (resp. $\tan \frac{\psi_h+\psi_i}{2}< \tan \frac{\psi_h+\psi_j}{2}$) and $x_{-1,[h,i,j]}<x_{-1,[h,i,l]}$ (resp. $y_{[h,i,l]}<y_{h,i,j}$), the $L_{[h,i]}$ soliton line must intersect the $[h,j]$ and $[i,l]$-solitons below (resp. above) $v$. Similar considerations place the $L_{[j,l]}$ soliton line on $\mathcal{C}(V_{n,A})$ (resp. $\mathcal{C}(Q_{t,A})$) above (resp. below) $v$, in the manner of Figure \ref{fig:X-n>0,t<0}. This fixes the placement of the $L_{[i,j]}$ and $L_{[h,l]}$ as well\footnote{Figures \ref{fig:X-n>0,t<0} and \ref{fig:X-n<0,t>0} are schematic in the sense that all generic soliton lines have positive slope and the $[i,j]$ and $[h,l]$-solitons may be parallel.}.

    Notice that in the region $x_1 \ll 0$ (resp. $x \ll 0)$, $\theta_h(x_{-1}, x_1, n)>\theta_i(x_{-1}, x_1, n)>\theta_j(x_{-1}, x_1, n)>\theta_l(x_{-1}, x_1, n)$ (resp. $\theta_h(x,y,t)>\theta_i(x,y,t)>\theta_j(x,y,t)>\theta_l(x,y,t)$); to reflect this, we label this region by $hijl$. Because we know the placement of the soliton lines in a neighborhood of $v$, we can similarly label the regions immediately adjacent to the X-crossing in this fashion. Notice that, on these regions, $\theta_h + \theta_j> \theta_l + \theta_i$. This implies that, if $\{h,j\}\cup R \in \mathcal{M}(A)$ for all $R \in \binom{[M]}{N-2}$, the X-crossing $v$ would not be visible. Indeed, there would exist some multi-index $R'$ such that $\Theta_{\{j,j\}\cup R'}$ is the dominant exponential on a region containing $v$. Therefore, there must exist $S \in \binom{[M]}{N-2}$ such that $\Delta_{\{h,j\} \cup S}(A) =0$. Similar considerations for the case in which $n \ll 0$ (resp. $t \gg 0$) on $\mathcal{C}(V_{n,A})$ (resp. $\mathcal{C}(Q_{t,A}))$ ensure there must exist $S$ such that $\Delta_{\{l,i\}\cup S}(A)=0$, as can be gleaned from Figure \ref{fig:X-n<0,t>0}.
\end{proof}

\begin{figure}[H]
    \centering

\tikzset{every picture/.style={line width=0.75pt}} 

\begin{tikzpicture}[x=0.75pt,y=0.75pt,yscale=-1,xscale=1]

\draw    (265.42,52.04) -- (174.79,243.42) ;
\draw    (275.93,240.99) -- (166.86,54.97) ;
\draw    (305.6,87.03) -- (64.67,164.09) ;
\draw    (309.34,190.93) -- (60.74,134.8) ;
\draw    (239.57,51.12) -- (236.62,243.74) ;
\draw    (197.32,52.95) -- (218.94,247.4) ;
\draw  [fill={rgb, 255:red, 0; green, 0; blue, 0 }  ,fill opacity=1 ] (240.55,107.79) .. controls (240.55,107.03) and (239.89,106.42) .. (239.08,106.42) .. controls (238.27,106.42) and (237.61,107.03) .. (237.61,107.79) .. controls (237.61,108.55) and (238.27,109.16) .. (239.08,109.16) .. controls (239.89,109.16) and (240.55,108.55) .. (240.55,107.79) -- cycle ;
\draw  [fill={rgb, 255:red, 0; green, 0; blue, 0 }  ,fill opacity=1 ] (206.16,118.78) .. controls (206.16,118.02) and (205.5,117.4) .. (204.69,117.4) .. controls (203.88,117.4) and (203.22,118.02) .. (203.22,118.78) .. controls (203.22,119.53) and (203.88,120.15) .. (204.69,120.15) .. controls (205.5,120.15) and (206.16,119.53) .. (206.16,118.78) -- cycle ;
\draw  [fill={rgb, 255:red, 0; green, 0; blue, 0 }  ,fill opacity=1 ] (238.59,174.62) .. controls (238.59,173.86) and (237.93,173.25) .. (237.12,173.25) .. controls (236.3,173.25) and (235.64,173.86) .. (235.64,174.62) .. controls (235.64,175.38) and (236.3,175.99) .. (237.12,175.99) .. controls (237.93,175.99) and (238.59,175.38) .. (238.59,174.62) -- cycle ;
\draw  [fill={rgb, 255:red, 0; green, 0; blue, 0 }  ,fill opacity=1 ] (211.08,168.21) .. controls (211.08,167.45) and (210.42,166.84) .. (209.6,166.84) .. controls (208.79,166.84) and (208.13,167.45) .. (208.13,168.21) .. controls (208.13,168.97) and (208.79,169.58) .. (209.6,169.58) .. controls (210.42,169.58) and (211.08,168.97) .. (211.08,168.21) -- cycle ;

\draw (69,144) node [anchor=north west][inner sep=0.75pt]  [font=\tiny] [align=left] {hijl};
\draw (207.13,153.18) node [anchor=north west][inner sep=0.75pt]  [font=\tiny,rotate=-271.28] [align=left] {ilhj};
\draw (213,121) node [anchor=north west][inner sep=0.75pt]  [font=\tiny] [align=left] {lihj};
\draw (236.57,136.56) node [anchor=north west][inner sep=0.75pt]  [font=\tiny,rotate=-90.57] [align=left] {lijh};
\draw (215,158) node [anchor=north west][inner sep=0.75pt]  [font=\tiny] [align=left] {iljh};
\draw (56,119) node [anchor=north west][inner sep=0.75pt]  [font=\scriptsize] [align=left] {[j,l]};
\draw (147,55) node [anchor=north west][inner sep=0.75pt]  [font=\scriptsize] [align=left] {[i,l]};
\draw (174.32,47.95) node [anchor=north west][inner sep=0.75pt]  [font=\scriptsize] [align=left] {[h,l]};
\draw (221,46) node [anchor=north west][inner sep=0.75pt]  [font=\scriptsize] [align=left] {[i,j]};
\draw (265.42,49.04) node [anchor=north west][inner sep=0.75pt]  [font=\scriptsize] [align=left] {[i,j]};
\draw (301,88) node [anchor=north west][inner sep=0.75pt]  [font=\scriptsize] [align=left] {[h,i]};
\end{tikzpicture}
\quad
\begin{tikzpicture}[x=0.75pt,y=0.75pt,yscale=-1,xscale=1]

\draw    (272.42,50.04) -- (181.79,241.42) ;
\draw    (282.93,238.99) -- (173.86,52.97) ;
\draw    (312.6,85.03) -- (71.67,162.09) ;
\draw    (316.34,188.93) -- (67.74,132.8) ;
\draw    (246.57,49.12) -- (243.62,241.74) ;
\draw    (204.32,50.95) -- (225.94,245.4) ;
\draw  [fill={rgb, 255:red, 0; green, 0; blue, 0 }  ,fill opacity=1 ] (247.55,105.79) .. controls (247.55,105.03) and (246.89,104.42) .. (246.08,104.42) .. controls (245.27,104.42) and (244.61,105.03) .. (244.61,105.79) .. controls (244.61,106.55) and (245.27,107.16) .. (246.08,107.16) .. controls (246.89,107.16) and (247.55,106.55) .. (247.55,105.79) -- cycle ;
\draw  [fill={rgb, 255:red, 0; green, 0; blue, 0 }  ,fill opacity=1 ] (213.16,116.78) .. controls (213.16,116.02) and (212.5,115.4) .. (211.69,115.4) .. controls (210.88,115.4) and (210.22,116.02) .. (210.22,116.78) .. controls (210.22,117.53) and (210.88,118.15) .. (211.69,118.15) .. controls (212.5,118.15) and (213.16,117.53) .. (213.16,116.78) -- cycle ;
\draw  [fill={rgb, 255:red, 0; green, 0; blue, 0 }  ,fill opacity=1 ] (245.59,172.62) .. controls (245.59,171.86) and (244.93,171.25) .. (244.12,171.25) .. controls (243.3,171.25) and (242.64,171.86) .. (242.64,172.62) .. controls (242.64,173.38) and (243.3,173.99) .. (244.12,173.99) .. controls (244.93,173.99) and (245.59,173.38) .. (245.59,172.62) -- cycle ;
\draw  [fill={rgb, 255:red, 0; green, 0; blue, 0 }  ,fill opacity=1 ] (218.08,166.21) .. controls (218.08,165.45) and (217.42,164.84) .. (216.6,164.84) .. controls (215.79,164.84) and (215.13,165.45) .. (215.13,166.21) .. controls (215.13,166.97) and (215.79,167.58) .. (216.6,167.58) .. controls (217.42,167.58) and (218.08,166.97) .. (218.08,166.21) -- cycle ;

\draw (74,141) node [anchor=north west][inner sep=0.75pt]  [font=\tiny] [align=left] {hijl};
\draw (214.13,150.18) node [anchor=north west][inner sep=0.75pt]  [font=\tiny,rotate=-270.91] [align=left] {ilhj};
\draw (220,118) node [anchor=north west][inner sep=0.75pt]  [font=\tiny] [align=left] {iljh};
\draw (222,155) node [anchor=north west][inner sep=0.75pt]  [font=\tiny] [align=left] {lihj};
\draw (242.51,134.51) node [anchor=north west][inner sep=0.75pt]  [font=\tiny,rotate=-90.11] [align=left] {lijh};
\draw (65,117) node [anchor=north west][inner sep=0.75pt]  [font=\scriptsize] [align=left] {[h,i]};
\draw (155,53) node [anchor=north west][inner sep=0.75pt]  [font=\scriptsize] [align=left] {[h,j]};
\draw (182,50) node [anchor=north west][inner sep=0.75pt]  [font=\scriptsize] [align=left] {[h,l]};
\draw (224,48) node [anchor=north west][inner sep=0.75pt]  [font=\scriptsize] [align=left] {[i,j]};
\draw (272.42,50.04) node [anchor=north west][inner sep=0.75pt]  [font=\scriptsize] [align=left] {[i,l]};
\draw (301.42,88.04) node [anchor=north west][inner sep=0.75pt]  [font=\scriptsize] [align=left] {[j,l]};

\end{tikzpicture}

     \caption{To the left (resp. right), the neighborhood of the $(L_{[h,j]}$, $L_{[i,l]})$ X-crossing  for $n \ll0$ (resp. $t \gg 0).$ }
    \label{fig:X-n<0,t>0}
\end{figure}

\end{document}